\documentclass[twoside,leqno,twocolumn]{article}

\usepackage[letterpaper]{geometry}
\usepackage{siamproceedings}

\input{styles/include.sty}
\input{styles/problems.sty}
\input{styles/variables.sty}

\makeatletter
\newcommand\affilfootnote[1]{%
  \begingroup
  \renewcommand\thefootnote{}%
  \renewcommand\@makefntext[1]{\noindent##1}%
  \footnotetext{#1}%
  \endgroup}
\makeatother
\begin{document}

\title{\Large An Evaluation of Vertex Elimination Strategies for Algorithmic Differentiation}
\author{Alex Crane\thanks{University of Utah, USA}\and
Pål Grønås Drange\thanks{University of Bergen, Norway}\and 
Eli Friedman\footnotemark[1]\and 
Paul D. Hovland\thanks{Argonne National Laboratory, USA}\and
Jan Hückelheim\footnotemark[3]\and 
Andrew Lyons\and 
Yosuke Mizutani\footnotemark[1]\and 
Macéo Ottavy\thanks{LIP, ENS de Lyon, France}\and 
Blair D. Sullivan\footnotemark[1]\textsuperscript{  }\footnotemark[4]\textsuperscript{  }\thanks{Collegium de Lyon, France}\textsuperscript{  }}
\date{}

{\setlength{\tabcolsep}{2pt}\maketitle}







\begin{abstract}
    The algorithmic differentiation (AD) of mathematical functions can be interpreted as a sequence of vertex eliminations in an underlying directed acyclic graph.
    The problem of determining a minimum-cost elimination ordering, which we call \StrOJA{}, is \cclass{NP}-complete, and much effort has been devoted to the design of heuristics.
    Many of these heuristics are widely believed to perform well in practice, but this hypothesis has so far been difficult to test due to the lack of scalable exact methods.

    We design and engineer new integer programming formulations for \StrOJA{} and for a related objective we call~\MinEC{}. Our implementations scale to graphs one-to-two orders of magnitude larger than existing techniques, facilitating the assembly of a corpus of medium-sized graphs for which optimal solutions are known. We give a comprehensive evaluation of existing heuristics on this corpus.

    We also make several theoretical contributions, including a tight analysis of the forward and reverse modes of AD, a simple algorithm for \StrOJA{} with approximation ratio parameterized by the size of a minimum source-sink separator, and 
    the first approximation lower bounds for both problems.
\end{abstract}
\section{Introduction.}\label{sec:intro}

For over fifty years, the (linearized)
computational graph model has been used to guide the algorithmic differentiation (AD) of functions~\cite{bauer1974computational,linnainmaa1970representation}. 
In a directed acyclic graph (DAG) $D = (V = \sources \uplus \internals \uplus \sinks, E)$,
source vertices $\sources$ and sink vertices $\sinks$ represent independent and dependent variables, respectively.
The internal vertices $\internals$ represent intermediate values, while directed edges model data dependencies.
By associating with each edge $(u, v)$ a local partial derivative $\partial{v}/\partial{u}$, we obtain a chain-rule-based interpretation of the Jacobian:
the derivative of a dependent variable $t$
with respect to an independent variable $s$ is equal to the sum
over all $st$-paths of the product of the edge
weights (partial derivatives) along the path. Many popular procedures for computing derivatives can be interpreted as
a sequence of graph modifications that preserve correctness of this chain-rule-based interpretation.
In this paper we focus on a particular graph modification called \emph{vertex elimination}.\looseness=-1    

The \emph{elimination} of an internal vertex $v \in \internals$ is defined by the deletion of~$v$ and the addition of the edge $(u, w)$, if it does not already exist, for every in-neighbor $u \in \innbrs{D}{v}$ and out-neighbor $w \in \outnbrs{D}{v}$.
To maintain correctness of the chain-rule-based interpretation, the product of the local partial derivatives $\partial{v}/\partial{u}$ and $\partial{w}/\partial{v}$ is added to the local partial derivative $\partial{w}/\partial{u}$, which is zero if the edge $(u,w)$ did not previously exist.
We write $\gElimVertex{D}{v}$ for the resulting DAG,
and generalize this notation to a \emph{sequence} $\sigma = (v_1, v_2, \ldots, v_\ell)$ of eliminations, where $\gElimSequence{D}{\sigma} \vcentcolon = (((\gElimVertex{D}{v_1})_{/v_2})\ldots)_{v_\ell}$.
Viewed this way, the ``forward mode''
of AD computes derivatives by eliminating all internal vertices in topological order. The terminal
graph is bipartite on vertices $\sources$ and $\sinks$, and the derivative associated with each edge captures directly a non-zero entry in the Jacobian.
The ``reverse mode'' of AD eliminates vertices in reverse topological order,
resulting in the same bipartite graph. Indeed, it is possible to eliminate the internal vertices in any order, obtaining an exact (up to round-off error) computation of the Jacobian. 
However, not all elimination sequences require the same computational effort.

\begin{figure*}
    \centering
    \begin{tikzpicture}[
    yscale=0.70,
    xscale=0.9,
    >={Stealth[length=2mm]},
    vertex/.style={circle, draw, fill=black, minimum size=2mm, inner sep=0pt},
    every edge/.style={draw, ->}
]

\begin{scope}[xshift=0cm]
    \node[vertex] (s1-0) at (0, 1) {};
    \node[vertex] (s2-0) at (0, 0) {};
    \node[vertex] (s3-0) at (0, -1) {};

    \node[vertex, label=above:$u$] (u-0) at (2, 0.5) {};
    \node[vertex, label=below:$v$] (v-0) at (2, -0.5) {};

    \node[vertex] (t-0) at (4, 0) {};

    \draw[->] (s1-0) -- (u-0);
    \draw[->] (s2-0) -- (v-0);
    \draw[->] (s3-0) -- (v-0);
    \draw[->] (u-0) -- (v-0);
    \draw[->] (u-0) -- (t-0);
    \draw[->] (v-0) -- (t-0);

    \node at (2, -1.5) {$D$};
\end{scope}

\begin{scope}[xshift=6cm]
    \node[vertex] (s1-1) at (0, 1) {};
    \node[vertex] (s2-1) at (0, 0) {};
    \node[vertex] (s3-1) at (0, -1) {};

    \node[vertex, label=below:$v$] (v-1) at (2, -0.5) {};

    \node[vertex] (t-1) at (4, 0) {};

    \draw[->] (s1-1) -- (v-1);
    \draw[->] (s2-1) -- (v-1);
    \draw[->] (s3-1) -- (v-1);
    \draw[->] (s1-1) -- (t-1);
    \draw[->] (v-1) -- (t-1);

    \node at (2, -1.5) {$D_{/u}$};
\end{scope}

\begin{scope}[xshift=12cm]
    \node[vertex] (s1-2) at (0, 1) {};
    \node[vertex] (s2-2) at (0, 0) {};
    \node[vertex] (s3-2) at (0, -1) {};

    \node[vertex, label=above:$u$] (u-2) at (2, 0.5) {};

    \node[vertex] (t-2) at (4, 0) {};

    \draw[->] (s1-2) -- (u-2);
    \draw[->] (u-2) -- (t-2);
    \draw[->] (s2-2) -- (t-2);
    \draw[->] (s3-2) -- (t-2);

    \node at (2, -1.5) {$D_{/v}$};
\end{scope}

\end{tikzpicture}
    \caption{On the left is a toy DAG~$D$ with two internal vertices. In $D$, $u$ has Markowitz degree~$2$ and $v$ has Markowitz degree~$3$. In the center is $\gElimVertex{D}{u}$, the DAG resulting from eliminating $u$. On the right is $\gElimVertex{D}{v}$, the DAG resulting from eliminating $v$. Eliminating in forward topological order ($u$ before $v$) results in a total cost of $2 + 3 = 5$, while eliminating in reverse topological order ($v$ before $u$) results in total cost $3 + 1 = 4$.}
    \label{fig:toy-dag}
\end{figure*}
The \emph{cost} of eliminating a vertex $v$ is the \emph{Markowitz degree} $\markdeg{v} \vcentcolon = |\innbrs{D}{v}|\cdot|\outnbrs{D}{v}|$, computed at the time of $v$'s elimination. This cost reflects the number of multiplications required to maintain correctness of the chain-rule-based interpretation of $D$. The \emph{cost} of a sequence of eliminations is the sum of the individual costs.
While forward and reverse modes both eliminate every internal vertex (i.e., they produce \emph{total elimination sequences}), the costs of these sequences may differ, as the elimination of a vertex may affect the Markowitz degrees of its not-yet-eliminated neighbors.
See~\Cref{fig:toy-dag} for a simple example.
Since the number of multiplications greatly affects the computational effort associated with derivative evaluations, a central problem in AD is to find low-cost elimination sequences; see the standard textbook~\cite{griewank2008evaluating}, or a recent open problems paper~\cite{Aksoy2023Seven}.

\begin{problembox}{\StrOJA{}}
    \Input & A DAG $D = (\sources \uplus \internals \uplus \sinks, E)$. \\
    \Prob & Find a total elimination sequence of minimum cost.
\end{problembox}

The original formulation of this problem is due to Griewank and Reese~\cite{Griewank1991OtC},
who also proposed a greedy heuristic:
at each step, the internal vertex with minimum Markowitz degree is eliminated.
Many other greedy heuristics
 have been proposed; see~\Cref{sec:heuristics}.

Additionally, a variety of randomized, search-based algorithms have been
developed, including algorithms based on simulated
annealing~\cite{naumann1999efficient,naumann2002SA}
and Monte Carlo tree
search~\cite{lohoff2024optimizing}. Due to their running times, these methods are best used in situations where the search cost can
be amortized over many derivative evaluations.

Another randomized algorithm, based on Markov Chain Monte Carlo, is
described in~\cite{lyons2012randomized}.
This algorithm targets a different objective, introduced
in~\cite{Griewank2003Mathematical} as \emph{scarcity}, and formalized
in~\cite{Aksoy2023Seven} as the {\MinEC} problem. In this
problem, the elimination sequence $\sigma$ need not be total, i.e., not all internal vertices need to be eliminated, and the objective changes from the cost of $\sigma$ to the number of edges in $\gElimSequence{D}{\sigma}$.
We refer to the original paper~\cite{Griewank2003Mathematical} for a detailed motivation.

\begin{problembox}{\MinEC{}}
    \Input & A DAG $D = (\sources \uplus \internals \uplus \sinks, E)$. \\
    \Prob & Find a (not necessarily total) elimination sequence $\sigma$ minimizing $|E(\gElimSequence{D}{\sigma})|.$
\end{problembox}

Ideally, one would evaluate heuristics by comparing them to an
optimal elimination sequence or (tight) lower bounds. Unfortunately,
the decision problems underlying both~\StrOJA{} and~\MinEC{} are \cclass{NP}-complete~\cite{bentert2025structural}, and finding
good lower bounds has remained elusive except in very restricted structural classes. Nonetheless, in an
effort to provide a baseline for comparison, Chen \emph{et al.}~\cite{chen2012integer,chen2012scarcity} developed
integer linear programming (ILP) formulations for both problems. Unfortunately, these formulations remain computationally prohibitive
for all but very small graphs (around~$20$ vertices).

In~\Cref{sec:ilp}, we introduce new, more efficient ILP formulations for both problems. These
enable us to find optimal solutions for medium-sized
graphs (a few hundred vertices for~\StrOJA{}; over~$1000$ for~\MinEC{}). 
This in turn facilitates a comprehensive assessment
of heuristics and search-based
algorithms on a broad range of real and synthetic computational
graphs; see \Cref{sec:experiments}. 

To further aid in understanding the performance of heuristics, in~\Cref{sec:separator-analysis} we give a tight bound on the approximation ratios achievable by the forward and reverse modes of AD. We then generalize our techniques to provide a parameterized approximation guarantee for~\StrOJA{}. The algorithm is based on eliminating ``outwards'' from a small $S - T$ separator, and the guarantee is parameterized by the size of this separator.
In~\Cref{sec:approximation-hardness}, we complement our algorithmic results by giving the first approximation lower bounds for both~\StrOJA{} and~\MinEC{}.
We also provide a (stronger) approximation lower bound for the dual of the latter problem,~\MaxER{}, which asks us to maximize $|E(D)| - |E(D_\sigma)|$, rather than minimizing $|E(D_\sigma)|$.
Throughout, we assume standard graph-theoretic notation, e.g., as seen in~\cite{diestel2012graph}.

\section{Integer Linear Programming.}\label{sec:ilp}
\label{sec:ILP}

In this section we present new ILPs for \StrOJA{} and \MinEC{}, both of which are asymptotically smaller than those proposed in~\cite{chen2012integer,chen2012scarcity}. As demonstrated in~\Cref{sec:experiments}, our formulations are significantly more efficient in practice, allowing us to obtain optimal solutions for most graphs in our corpus.

\subsection{ILP for \StrOJA{}.}\label{sec:ilp-flops}

The ILP formulation provided by~\cite{chen2012integer,chen2012scarcity} has $\Theta(n^3)$ variables and $\Theta(n^4)$ constraints. The ILP tracks the structure of the graph at timesteps $t = 0, 1, \ldots, |\internals| + 1$, where the graph $D(t) = (V(t), E(t))$ reflects the first~$t$ vertex eliminations. The bound on the number of constraints then comes from the need to encode events such as the creation of edge $(i, j)$ via the elimination of vertex $k$ at timestep $t$. 
Here, we propose a new ILP with only $\Theta(n^3)$ constraints. Our main innovation is to abandon the timestep abstraction (though it remains useful for analysis) and encode only the relative order of eliminations. Formally, our first set of variables is:
\begin{align*}\label{eq:ilp:order:x-vars}
    x_{ij} = 
    \begin{cases}
        1 \quad\text{ if $i, j \in \internals$ and $i$ is eliminated before $j$.}\\
        1 \quad\text{ if $i \in \internals$ and $j \notin \internals$.} \\
        0 \quad\text{ otherwise.} 
    \end{cases}
\end{align*}

We need two more variable sets. One tracks whether the edge $(i, j)$ exists at any point in the elimination procedure. The other tracks the cost of eliminations.
We say that a vertex triplet $(i, j, k)$ is a \emph{multiplication at} $k$ if $i$ and $j$ are in- and out-neighbors (respectively) of $k$ when $k$ is eliminated. Notice that the number of multiplications at $k$ is precisely $k$'s Markowitz degree at the time of its elimination, i.e., the cost of the elimination.\looseness=-1

\begin{align*}
    e_{ij} &=
    \begin{cases}
        1 \quad\text{ if $(i, j) \in \bigcup_t E(t)$.}\\
        0 \quad\text{ otherwise.}
    \end{cases}\label{eq:ilp:order:e-vars}
\end{align*}
\begin{align*}
    z_{ijk} &=
    \begin{cases}
        1 \quad\text{ if $(i, j, k)$ is a multiplication at $k$.}\\
        0 \quad\text{ otherwise.}
    \end{cases}\label{eq:ilp:order:z-vars}
\end{align*}
We can now present our new ILP.
\vspace{\abovedisplayskip}
{\abovedisplayskip=0pt \belowdisplayskip=0pt \abovedisplayshortskip=0pt \belowdisplayshortskip=0pt
\begin{alignat}{2}
    &\text{min.}\quad F = \sum_{k \in \internals} \sum_{(i, j) \in E} z_{ijk}. & &\label{eq:ilp:order:objective} \\
    &\text{subject to:} & &\nonumber\\
    &x_{ij} + x_{jk} \leq x_{ik} + 1 & \quad\quad&\forall\ i, j, k \in V. \label{eq:ilp:order:partial-order}\\
    &x_{ij} = 1 - x_{ji} & &\forall\ i, j \in \internals. \label{eq:ilp:order:total-order} \\
    &e_{ij} = 1 & &\forall\ (i, j) \in E. \label{eq:ilp:order:initial-edges}
\end{alignat}
\begin{multline}
    \hspace{.75em}x_{ki} + x_{kj} + e_{ik} + e_{kj} \leq e_{ij} + 3\\
    \hfill \forall\ k \in \internals,\ \forall\ i, j \in V.
    \label{eq:ilp:order:existence}
\end{multline}
\begin{multline}
    \hspace{.75em}x_{ki} + x_{kj} + e_{ik} + e_{kj} \leq z_{ijk} + 3\\
    \hfill\forall\ k \in \internals, \ \forall\ i, j \in V.
     \label{eq:ilp:order:payment}
\end{multline}
\begin{multline}
    \hspace{.75em}x_{ij}, e_{ij}, z_{ijk} \in \{0, 1\} \hfill{\forall\ i, j \in V, \ \forall\ k \in \internals.} \label{eq:ilp:order:integrality}
\end{multline}
}\vspace{\belowdisplayskip}

The objective~(\ref{eq:ilp:order:objective}) is the total number of multiplications, i.e., the cost of the elimination sequence. Constraints~(\ref{eq:ilp:order:partial-order}) ensure that the relation defined by $x_{ij} = 1$ is transitive. Constraints~(\ref{eq:ilp:order:total-order}) extend this to a total order on the internal vertices. Given a feasible solution, we may recover the vertex elimination order by examining the variables $x_{ij}$ for all $i, j \in \internals$.
Constraints~(\ref{eq:ilp:order:initial-edges}) ensure that the initial structure of the graph is properly encoded.
The constraints~(\ref{eq:ilp:order:existence}) and~(\ref{eq:ilp:order:payment}) ensure correctness, and are motivated by the following observation: if the edge $(i, j)$ exists at some point in the elimination process, then it must continue existing until either $i$ or $j$ is eliminated. 
Then if $(i, k)$ and $(k, j)$ exist and $k$ is eliminated before both $i$ and $j$, the edge $(i, j)$ must exist~(\ref{eq:ilp:order:existence}), and furthermore $(i, j, k)$ is a multiplication at $k$~(\ref{eq:ilp:order:payment}).
Finally, constraints~(\ref{eq:ilp:order:integrality}) ensure that all variables take binary values.

\smallskip
\noindent\textbf{Optimizations}. We present several optimizations which break symmetries and encode non-trivial lower bounds.
The first is the simplest. Let $\widetilde{E}$ denote the set of ordered pairs $(i, j) \in V \times V$ such that $j$ is reachable from $i$ in the input DAG $D$, i.e., $\widetilde{E}$ is the edge-set of the transitive closure of $D$.
We observe that pairs $(i, j) \notin \widetilde E$ can never be edges at any point in the elimination procedure. Furthermore, $\widetilde{E}$ is easily precomputable.
Thus, we constrain some of our variables:

{\abovedisplayskip=0pt \belowdisplayskip=0pt \abovedisplayshortskip=0pt \belowdisplayshortskip=0pt
\begin{multline}
    e_{ij} = 0 \hfill \forall \ i, j \in V \text{ such that } (i, j) \notin \widetilde E.\label{eq:ilp:order:reachability-ij}
\end{multline}
\begin{multline}
    z_{ijk} = 0 \hfill \forall \ i, j \in V, k \in \internals \text{ such that}\\
    \hfill (i, k) \notin \widetilde E \text{ or } (k, j) \notin \widetilde E.\label{eq:ilp:order:reachability-ijk}
\end{multline}
}\vspace{\belowdisplayskip}

Our next optimizations provide useful lower bounds. For an internal vertex $i \in \internals$,
we denote by $\cutsize{\sources}{i}$ the minimum size of a vertex set $X \not\ni i$ separating $\sources$ and $i$,
and by $\cutsize{i}{\sinks}$ the analogous quantity for $i$ and $\sinks$.
The following observation is easy to see.

\begin{observation}[\cite{naumann1999efficient,naumannhu2008optimal}]\label{obs:sep-lower-bound}
    For any $k \in \internals$, the Markowitz degree of $k$ at the time of its elimination is at least $\cutsize{\sources}{k}\cdot \cutsize{k}{\sinks}$.
\end{observation}

Note that $\cutsize{\sources}{k}$ and $\cutsize{k}{\sinks}$ can be efficiently computed via a standard reduction to~\pname{Max Flow}.
This motivates the following lower-bounding constraints:

\begin{align}
    \sum_{i \neq k}\sum_{j \not\in \{i, k\}} z_{ijk} \ge \cutsize{\sources}{k}\cdot \cutsize{k}{\sinks} && \hfill\forall \ k \in \internals.\label{eq:ilp:order:lb-separators}
\end{align}

Constraints similar to~(\ref{eq:ilp:order:reachability-ij}) - (\ref{eq:ilp:order:lb-separators}) have been used previously~\cite{chen2012integer,chen2012scarcity}.
Our last lower-bound constraint is based on a new idea. Observe that without loss of generality, we may assume that every edge has at least one internal endpoint. Otherwise, the edge may be deleted as a safe preprocessing step, since neither endpoint can ever be eliminated.
A lower-bound follows. 

\begin{restatable}{lemma}{movertwolowerbound}\label[lemma]{lem:LB:m/2}
    For every DAG $D = (S \uplus I \uplus T, E)$ with every edge having at least one endpoint in $I$, the minimum total elimination cost $\OPT(D)$ is bounded by $\OPT(D) \geq |E|/2$.
\end{restatable}

We give a proof in Appendix~\ref{appendix:proofs:ilps}. Here, we use~\Cref{lem:LB:m/2} to introduce another constraint:

\begin{align}
    \sum_{k \in \internals}\sum_{i \neq k}\sum_{j \not\in\{i, k\}} z_{ijk} \geq |E|/2. &&\label{eq:ilp:order:lb-edges}
\end{align}

Finally, to reduce the size of instances before solving, we use a widely accepted preprocessing rule. Note that while an argument is given in~\cite{naumann1999efficient}, we believe that it would be valuable to obtain a concise formal proof.

\begin{proposition}[\cite{griewank2008evaluating,naumann1999efficient}]\label{prop:subdivision-preprocessing}
    Let $D = (\sources \uplus \internals \uplus \sinks, E)$, and let $w \in \internals$ have Markowitz-degree~$1$. Then there exists a minimum cost total elimination sequence in which $w$ appears first.
\end{proposition}

In~\Cref{sec:experiments}, we will demonstrate the effectiveness of our optimizations by testing multiple variants of our ILP formulation.
These are
\begin{itemize}
    \item \va{}: Obj.~(\ref{eq:ilp:order:objective}) with constraints~(\ref{eq:ilp:order:partial-order})-(\ref{eq:ilp:order:integrality}),
    \item \vb{}: Obj.~(\ref{eq:ilp:order:objective}) with constraints~(\ref{eq:ilp:order:partial-order})-(\ref{eq:ilp:order:reachability-ijk}),
    \item \vc{}: Obj.~(\ref{eq:ilp:order:objective}) with constraints~(\ref{eq:ilp:order:partial-order})-(\ref{eq:ilp:order:lb-edges}),
    \item \vd{}: The program \vc{}, with exhaustive use of~\Cref{prop:subdivision-preprocessing} as a preprocessing routine.
\end{itemize}

\subsection{ILP for \MinEC{}.}\label{sec:ilp-scarcity}

The ILP of~\cite{chen2012integer,chen2012scarcity} has $\Theta(n^3)$ variables and $\Theta(n^4)$ constraints, tracking events such as the creation of edge $(i, j)$ via the elimination of vertex $k$ at timestep $t$.
We again do away with the timestep abstraction, improving the number of constraints to $\Theta(n^3)$. More significantly, we reduce the number of variables to $\Theta(n^2)$.
The key tool is a recent result of~\cite{bentert2025structural}.

\begin{lemma}[\cite{bentert2025structural}]\label{lemma:scarcity-order-doesnt-matter}
    Let $D = (\sources \uplus \internals \uplus \sinks, E)$ be a DAG, $X \subseteq \internals$, and let $\sigma_1$ and $\sigma_2$ be any two permutations of $X$. Then $D_{\sigma_1} = D_{\sigma_2}$.
\end{lemma}

\Cref{lemma:scarcity-order-doesnt-matter} essentially says that~\MinEC{} is a vertex-selection problem, rather than a vertex-ordering problem. Henceforth, we will write $D_X$ for the DAG resulting from eliminating the vertices $X \subseteq I$ in any order; this notation is well-defined by~\Cref{lemma:scarcity-order-doesnt-matter}.
Our first set of variables encodes the set $X$ of eliminated vertices.

\begin{align*}\label{eq:ilp:scarcity:x-vars}
    x_{i} =
    \begin{cases}
        1 \quad\text{ if $i \in X$.}\\
        0 \quad\text{ otherwise.}
    \end{cases}
\end{align*}

It remains to encode the structure of the DAG. We use another fundamental observation.

\begin{restatable}{observation}{pathcontraction}\label{obs:directed-path-with-eliminated-internals}
    Let $D = (V = \sources \uplus \internals \uplus \sinks, E)$, $X \subseteq \internals$, and $i, j \in V \setminus X$. Then $(i, j) \in E(D_X)$ if and only if there exists a directed path from $i$ to $j$ for which every internal vertex is in $X$.
\end{restatable}

We give a proof in Appendix~\ref{appendix:proofs:ilps}. We now use variables $y_{ij}$ and $z_{ij}$ to encode, respectively, the existence of edges in the final DAG $D_X$, and at intermediate stages.

\begin{align*}
    y_{ij} &=
    \begin{cases}
        1 \quad\text{if $(i, j) \in E(D_X)$.}\\
        0 \quad\text{otherwise.}
    \end{cases}\label{eq:ilp:scarcity:y-vars}
\end{align*}
\begin{align*}
    z_{ij} &=
    \begin{cases}
        1 \quad\text{if there exists $X' \subseteq X$ with $(i, j) \in E(D_{X'})$.}\\
        0 \quad\text{otherwise.}
    \end{cases}\label{eq:ilp:scarcity:z-vars}
\end{align*}
We now present our ILP for \MinEC{}.
\begin{alignat}{2}
    &\text{min.}\quad \sum_{i, j \in V} y_{ij} & &\label{eq:ilp:scarcity:objective}\\
    &\text{subject to}&&\nonumber\\
    &z_{ij} = 1 & &\forall\ (i, j) \in E. \label{eq:ilp:scarcity:initial-z-variables}\\
    &z_{ij} \geq z_{kj} + x_k - 1 & \quad\quad&\forall\ i, j \in V, \ \forall\ k \in N^+(i). \label{eq:ilp:scarcity:z-variable-correctness}\\
    &y_{ij} \geq z_{ij} - x_i - x_j & &\forall\ i, j \in V. \label{eq:ilp:scarcity:final-edges}\\
    &x_i, y_{ij}, z_{ij} \in \{0, 1\} & &\forall\ i, j \in V. \label{eq:ilp:scarcity:integrality}
\end{alignat}

The objective~(\ref{eq:ilp:scarcity:objective}) is the number of edges in the final DAG. The constraints~(\ref{eq:ilp:scarcity:initial-z-variables}) encode the initial structure of the input DAG $D$. The constraints~(\ref{eq:ilp:scarcity:z-variable-correctness}) ensure correctness of the $z_{ij}$ variables: if $k$ is an out-neighbor of $i$, $k$ is eliminated, and there exists a path from $k$ to $j$ all of whose internal vertices are eliminated, then there exists such a path from $i$ to~$j$. Constraints~(\ref{eq:ilp:scarcity:final-edges}) ensure correctness of the $y_{ij}$ variables: if there exists a path from $i$ to $j$ for which all internal vertices are eliminated, and if neither $i$ nor $j$ is eliminated, then by~\Cref{obs:directed-path-with-eliminated-internals} the arc $(i, j)$ exists in the final DAG. Finally, constraints~(\ref{eq:ilp:scarcity:integrality}) enforce integrality.\looseness=-1

\subsection{Integrality Gaps.}\label{sec:integrality-gaps}
By relaxing constraints~(\ref{eq:ilp:order:integrality}) or~(\ref{eq:ilp:scarcity:integrality}) to allow variable values in the range $[0, 1]$, we obtain efficiently solvable linear programming (LP) relaxations for \StrOJA{} and \MinEC{}, respectively. It is natural to wonder whether rounding fractional solutions may lead to high-quality approximations.
It turns out that this direction is unlikely to be fruitful without additional assumptions.

\begin{restatable}{proposition}{flopsintegralitygap}\label[proposition]{prop:flops-integrality-gap}
    The fractional LP-relaxation of the ILP defined by objective (\ref{eq:ilp:order:objective}) and constraints (\ref{eq:ilp:order:partial-order})-(\ref{eq:ilp:order:lb-edges}) has an integrality gap of $\Omega(n^{1/4})$.
\end{restatable}

\begin{restatable}{proposition}{scarcityintegralitygap}\label[proposition]{prop:scarcity-integrality-gap}
    The fractional LP-relaxation of the ILP defined by objective (\ref{eq:ilp:scarcity:objective}) and constraints (\ref{eq:ilp:scarcity:initial-z-variables})-(\ref{eq:ilp:scarcity:integrality}) has an integrality gap of $\Omega(n)$.
\end{restatable}

Both of our constructions (see Appendix~\ref{appendix:proofs:ilps}) have vertices with degree linear in $n$, so it may still be interesting to explore rounding schemes in the bounded-degree setting.
Bounded in-degree, in particular, may be practically relevant for computational graphs.
\section{Topological Ordering and Small Separators.}\label{sec:separator-analysis}

It turns out that a second application of~\Cref{lem:LB:m/2} is to analyze the forward (\fo{}) and reverse (\re{}) modes of AD, which eliminate internal vertices in forward (or reverse, respectively) topological order.
We need a straightforward second lemma. Proofs of statements in this section can be found in Appendix~\ref{appendix:proofs:separator-analysis}.

\begin{restatable}{lemma}{costoftopo}\label[lemma]{lem:topo}
    Eliminating $D = (S \uplus I \uplus T, E)$ in forward topological order incurs cost at most $|S|\cdot |E|$. Reverse topological order incurs cost at most $|T|\cdot |E|$.
\end{restatable}

Combining~\Cref{lem:LB:m/2,lem:topo} immediately yields an approximation bound.
\begin{restatable}{theorem}{apxratioforwardreverse}\label{thm:forward-reverse-approximations}
    The forward mode of AD is a $(2\cdot|S|)$-approximation for \StrOJA{}, and reverse mode is a $(2\cdot|T|)$-approximation. Both bounds are tight.
\end{restatable}

\Cref{thm:forward-reverse-approximations} matches the folklore knowledge that topological elimination order (say, \fo{}) is best when $|S|$ is small, as in-neighborhoods remain small throughout the elimination process. Prior work on heuristics has observed that, similary to small sets of sources, small separators could be algorithmically promising for \StrOJA{} (e.g., see~\cite{Hovland1997efficient, Iri1991HoA,naumannhu2008optimal,Tadjouddine2008VoA,volinostrovski85}). Building on \Cref{lem:LB:m/2,lem:topo}, we prove that eliminating ``outwards'' from small separators gives a parameterized approximation similar to \Cref{thm:forward-reverse-approximations}.

We now describe an algorithm, \mo, which takes as input $C$, an $S$-$T$ separator on the underlying undirected graph of $D$ (which we will call $G_D$). Intuitively, we want to find a topological ordering, run \re{} and \fo{} from $C$ towards the sources and sinks, and then eliminate $C\cap I$ in topological order. However, the vertices of $C$ are not necessarily consecutive within a topological ordering. Defining \emph{left} (resp., \emph{right}) internal vertices to be those reachable in $G_D\setminus C$ from $S$ (resp., $T$), it is possible to have left and right vertices interspersed within $C$ in a topological ordering. There may also exist \emph{middle} vertices which are disconnected from both $S$ and $T$ in $G_D\setminus C$. Since these cases can occur even when $C\cap I$ is a minimum separator (see \Cref{fig:topoweird}), we must take more care when formally specifying \mo.

\begin{figure}[h]
		\centering
\begin{tikzpicture}
\node[] (s1) at (0,1) {$s_1$};
\node[] (s2) at (0,-1) {$s_2$};
\node[circle, draw=black] (v1) at (1.5,0) {$v_1$};
\node[] (v2) at (3,0) {$v_2$};
\node[circle, draw=black] (v3) at (4.5,0) {$v_3$};
\node[] (t1) at (6,1) {$t_1$};
\node[] (t2) at (6,-1) {$t_2$};

\draw[-Latex] (s1)--(v1);
\draw[-Latex] (s2)--(v1);
\draw[-Latex] (v1)--(v2);
\draw[-Latex] (v2)--(v3); 
\draw[-Latex] (v3)--(t1); 
\draw[-Latex] (v3)--(t2); 

\draw[-Latex] (s1) to [bend left] (v3);
\draw[dashed, -Latex] (s2) to [bend right] (v2);
\draw[-Latex] (v1) to [bend right] (t2);

\end{tikzpicture}
	\caption{Examples of non-consecutive minimum separators. For $D$ (solid edges only) with $S=\{s_1,s_2\}$ and $T=\{t_1,t_2\}$, $C=\{v_1,v_3\}$ is a minimum $S-T$ separator in $G_D$, but $v_2$ is a middle vertex (see text for definitions). For $D'=D + (s_2,v_2)$ (including the dashed edge) $C$ is a minimum separator in $G_{D'}$ but $v_2$ is a left vertex contained ``within'' $C$ in every topological order.\looseness=-1}
  \label{fig:topoweird}
\end{figure}

Let $V_L$ and $V_R$ be the sets of left and right vertices, and then partition the middle vertices into connected components $V_1,\dots, V_k$ in $G_D\setminus C$.  
\mo{} proceeds by eliminating $V_L$ in reverse topological order, then $V_R$ forwards, then each $V_i$ forwards, and finally $C\cap I$ forwards. Note that we cannot simply eliminate vertices $I\setminus C$ in reverse order starting from the topological ``beginning'' of $C$, as this could miss some left vertices, nor from the topological ``end'' of $C$, as this could eliminate some of the right vertices in reverse order. 
We now state the algorithm, denoting the order $\pi$ restricted to a subset $S$ as $\pi[S]$ and the reverse of $\pi$ as $\pi^r$.
\begin{algorithm}
    \caption{\label{alg:middleout}\mo$(D,C)$}
\begin{algorithmic}[1]
\STATE Partition $I\setminus C$ into $V_L$, $V_R$, and $\{V_i\}_{i\in[k]}$ by finding connected components of $G_D\setminus C$
\STATE Find a topological order $\pi$ on $D$ 
\STATE $\sigma \gets \pi^r[V_L] $
\STATE $\sigma \gets \sigma+ \pi[V_R]$ \COMMENT{We append to the end of $\sigma$}
\FOR{$i\in [k]$}
    \STATE $\sigma \gets \sigma + \pi[V_i]$ 
\ENDFOR
\STATE $\sigma \gets \sigma +\pi[C\cap I]$\\
\STATE \textbf{return} $\sigma$
\end{algorithmic}
\end{algorithm}

Observe that both $S$ and $T$ are $S$-$T$ separators, so \mo{} strictly generalizes~\fo{} and~\re{}. Furthermore, a minimum-size separator may be found in polynomial time via a standard reduction to~\pname{Max Flow}.
To analyze~\mo{}, we prove a new lower bound which states that for any $v\in I$, the Markowitz degree of $v$ after eliminating all other interior vertices is a lower bound on the total cost of all elimination sequences.
\begin{restatable}{lemma}{lowerboundsources}\label[lemma]{lem:LB:sources}
    Let $D = (S \uplus I \uplus T, E)$ be a DAG.
    For each $v\in I$, the minimum total elimination cost $\OPT(D)$ is bounded by
    $\OPT(D) \geq |\innbrs{\gElimSet{D}{I\setminus\{v\}}}{v}|\cdot |\outnbrs{\gElimSet{D}{I\setminus\{v\}}}{v}|$.
\end{restatable} 

By combining~\Cref{lem:LB:m/2,lem:topo,lem:LB:sources}, we prove the following approximation guarantee for~\mo{}, parameterized by the size of the separator $C$.

\begin{restatable}{theorem}{middleoutapproximation}\label{thm:middle-out}
   Let $C$ be an $S$-$T$ separator in the underlying undirected graph of $D$. Then the algorithm $\mo(D,C)$ is an $O(|C|^2)$-approximation for~\StrOJA{}.
\end{restatable}

In Appendix~\ref{appendix:proofs:separator-analysis}, we discuss settings where this approximation ratio can be improved.

\section{Summary of Evaluated Heuristics.}
\label{sec:heuristics}

In addition to \fo{}, \re{}, and separator-based methods,
numerous other heuristics for \StrOJA{} have been proposed, which we summarize here.
A key is presented below.
The original \alg{Greedy Markowitz} (\greedymark{}) rule of Griewank and Reese~\cite{Griewank1991OtC}
greedily minimizes the cost of the current elimination.
Several variations have subsequently appeared. For a vertex $v \in \internals$, let $\sources_v \subseteq \sources$ be the subset of sources from which $v$ is reachable, and let $\sinks_v$ be defined similarly. Naumann~\cite{naumann1999efficient} proposed
to modify \greedymark{} by considering the cost of eliminating $v$ \emph{last}: \alg{Relative Markowitz} (\relmark) greedily minimizes $\markdeg{v} - |S_v|\cdot|T_v|$.
Two more variations, both proposed in~\cite{albrecht2003markowitz}, consider the effect of eliminating $v$ on the Markowitz degrees of the remaining vertices. Let $\totalmark{D} \vcentcolon = \sum_{v \in I} \markdeg{v}$ denote the total Markowitz degree of $D$. 
\alg{Maximum Total Markowitz Reduction} (\mtmr{}) greedily selects $v$ minimizing $M(\gElimVertex{D}{v})$. This idea can be refined as follows.
Define the \emph{damage} $d(v)$ of eliminating $v$ to be $d(v) \vcentcolon = \totalmark{\gElimVertex{D}{v}} - \totalmark{D} + \markdeg{v}$.
\alg{Least Markowitz Minimal Damage} (\lmmd) greedily selects $v$ minimizing $\markdeg{v} + d(v)$.
Finally, another natural variant, which we believe has not appeared in the literature previously, considers the Markowitz degree of $v$ in comparison to its minimum possible elimination cost $\mu_D^*(v) \vcentcolon = \cutsize{\sources}{v}\cdot\cutsize{v}{\sinks}$, as stated in~\Cref{obs:sep-lower-bound}.
\alg{Difference to Minimum Cost} (\diffmincost{}) greedily selects $v$ minimizing $\markdeg{v} - \mu_D^*(v)$.

Another strategy arises from an upper bound on the optimal objective value, shown in~\cite{griewank2008evaluating}.
Let $|| \sources \rightarrow \sinks ||_D$ denote the sum of the lengths of all maximal paths in~$D$.
This quantity is computable via a dynamic program, leading to
\alg{Maximum Pathlength Reduction} (\pathlen{}), which greedily selects $v$ minimizing $|| \sources \rightarrow \sinks ||_{\gElimVertex{D}{v}}$~\cite{griewank2008evaluating,naumann1999efficient}.
Similarly, writing~$\numpaths{v}$ for the number of paths on which $v$ appears,
\alg{Pathcount} (\pc{})~\cite{naumann1999efficient} greedily selects $v$ minimizing the ratio $\markdeg{v} / \numpaths{v}$.

Finally, Naumann~\cite{naumann1999efficient} explored the intuition that reducing the number of edges should lead to low cost.
\alg{Greedy Edge Reduction} (\edgereduction{}) greedily selects $v$ minimizing $|E(\gElimVertex{D}{v})|$,
while a computationally simplified variant (\markanddegree) greedily selects $v$ minimizing $\markdeg{v} - \deg(v)$.

\begin{description}[leftmargin=1cm, labelwidth=1.25cm, itemsep=0pt, parsep=0pt]
    \item[\fo{}:] Eliminate all $v \in I$ in forward topo. order.
    \item[\re{}:] Eliminate all $v \in I$ in reverse topo. order.
    \item[\greedymark{}:] Select $v \in I$ minimizing $\markdeg{v}$.
    \item[\relmark{}:] Select $v \in I$ minimizing $\markdeg{v} - |S_v|\cdot|T_v|$.
    \item[\mtmr{}:] Select $v \in I$ minimizing $\totalmark{\gElimVertex{D}{v}}$.
    \item[\lmmd{}:] Select $v \in I$ minimizing $\markdeg{v} + d(v)$.
    \item[\diffmincost{}:] Select $v \in I$ minimizing $\markdeg{v} - \mu_D^*(v)$.
    \item[\pathlen{}:] Select $v \in I$ minimizing $|| \sources \rightarrow \sinks ||_{\gElimVertex{D}{v}}$.
    \item[\pc{}:] Select $v \in I$ minimizing $\markdeg{v} / \numpaths{v}$. 
    \item[\edgereduction{}:] Select $v \in I$ minimizing $|E(\gElimVertex{D}{v})|$
    \item[\markanddegree{}:] Select $v \in I$ minimizing $\markdeg{v} - \deg(v)$.
\end{description}

For all the above except~\greedymark{}, we take a standard approach to tiebreaking, e.g., as seen in~\cite{albrecht2003markowitz}, by choosing a candidate with minimum Markowitz degree. For~\greedymark{}, we again follow the lead of~\cite{albrecht2003markowitz} by choosing a candidate with highest degree.
We also test a simulated annealing approach (\sa{}) proposed by~\cite{naumann1999efficient,naumann2002SA} and a reinforcement learning approach (\ag{}) proposed by~\cite{lohoff2024optimizing}. Finally, while we proposed \mo{} (\moshort{}~for short) for theoretical reasons, we include it for completeness; our implementation uses \textsc{Max Flow} to find a minimum separator.

\smallskip
\noindent{\textbf{Heuristics for \MinEC{}}}.
A classic strategy for~\MinEC{} is to simply execute one of the greedy~\StrOJA{} heuristics described above,
recording at each step the size of the DAG, and return the smallest DAG encountered.
In this context, the directly greedy approach is represented by \edgereduction{}.

We also consider a dedicated heuristic (\mcmc{}),
which searches for small representations via a Markov Chain Monte Carlo approach.
Our implementation is an adaptation (to the setting of vertex elimination) of an
algorithm~\cite{lyons2012randomized} defined in the setting of
edge elimination~\cite{naumann1999efficient}, which can be viewed as
a family of finer-grained graph edits relative to vertex elimination.

\section{Computational Results.}\label{sec:experiments}
We implemented our new ILPs, the ILPs from~\cite{chen2012integer,chen2012scarcity}, and all algorithms except~\ag{} in C++, and tested on a Intel(R) Xeon(R) Gold 6230 CPU @ 2.10GHz. We used Gurobi~\cite{Gurobi} (version $13$; single-threaded) to solve ILPs.
For~\ag{}, we used the implementation of~\cite{lohoff2024optimizing} and tested on a NVIDIA RTX A6000.
Complete results are in Appendix~\ref{appendix:experiments}.
Implementations and datasets are at~\href{https://github.com/TheoryInPractice/crosscountryAD}{github.com/TheoryInPractice/crosscountryAD}.
\begin{figure*}[t]
    \centering
    \begin{subfigure}[b]{0.32\textwidth}
        \includegraphics[width=\textwidth]{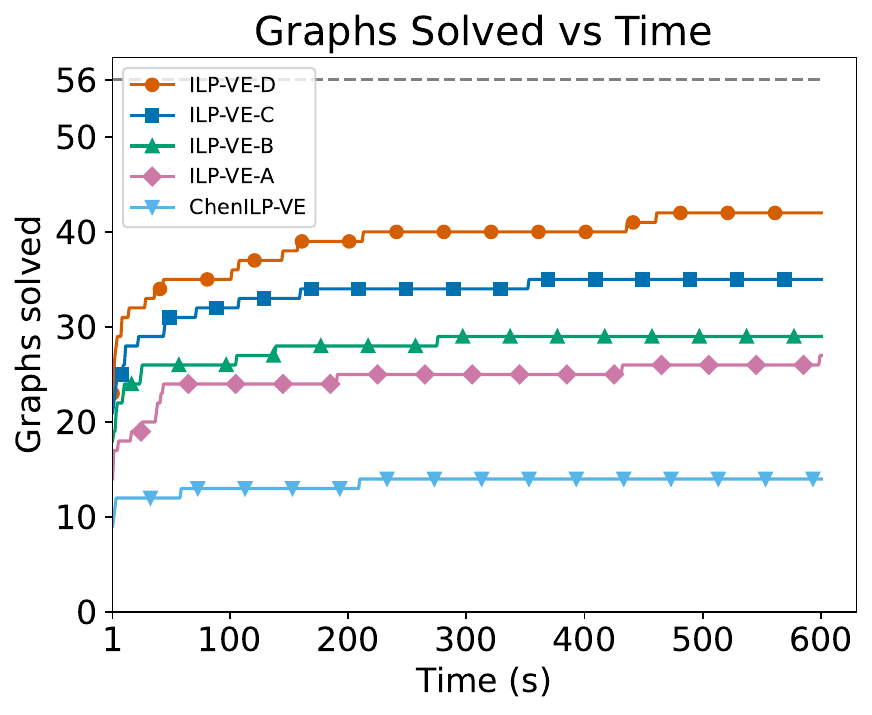}
        \caption{}
        \label{fig:flops-ilp:graphs-vs-time}
    \end{subfigure}
    \hfill
    \begin{subfigure}[b]{0.32\textwidth}
        \includegraphics[width=\textwidth]{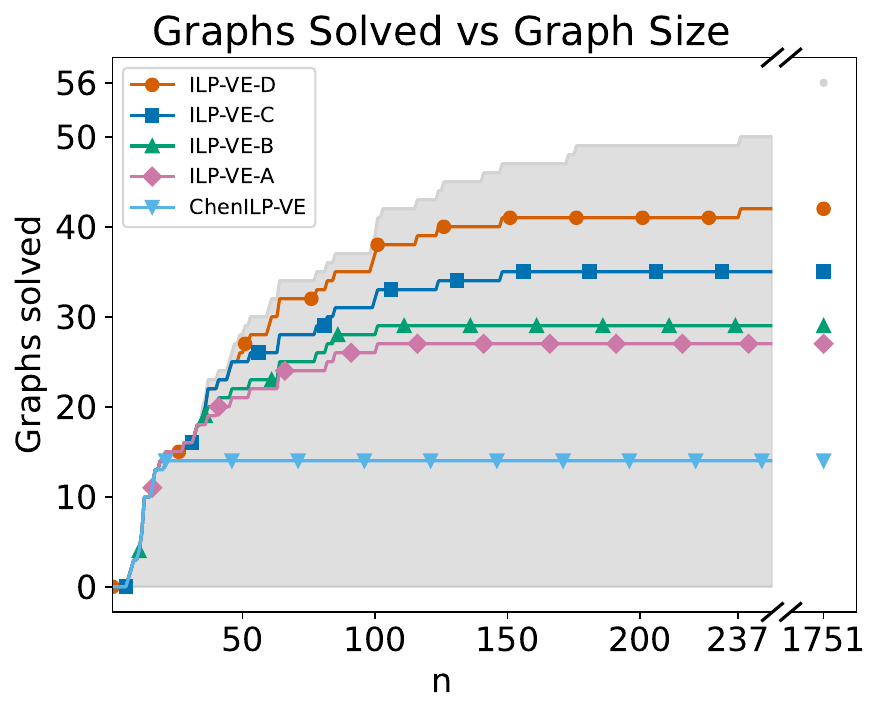}
        \caption{}
        \label{fig:flops-ilp:graphs-vs-n}
    \end{subfigure}
    \hfill
    \begin{subfigure}[b]{0.32\textwidth}
        \includegraphics[width=\textwidth]{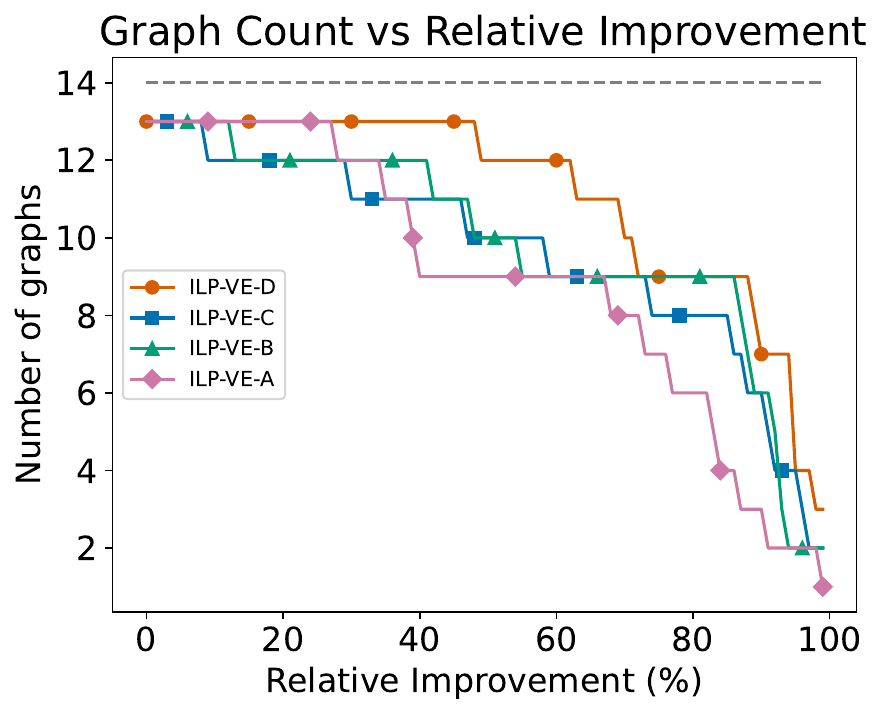}
        \caption{}
        \label{fig:flops-ilp:speedups}
    \end{subfigure}

    \vspace{1em}

    \begin{subfigure}[b]{0.32\textwidth}
        \includegraphics[width=\textwidth]{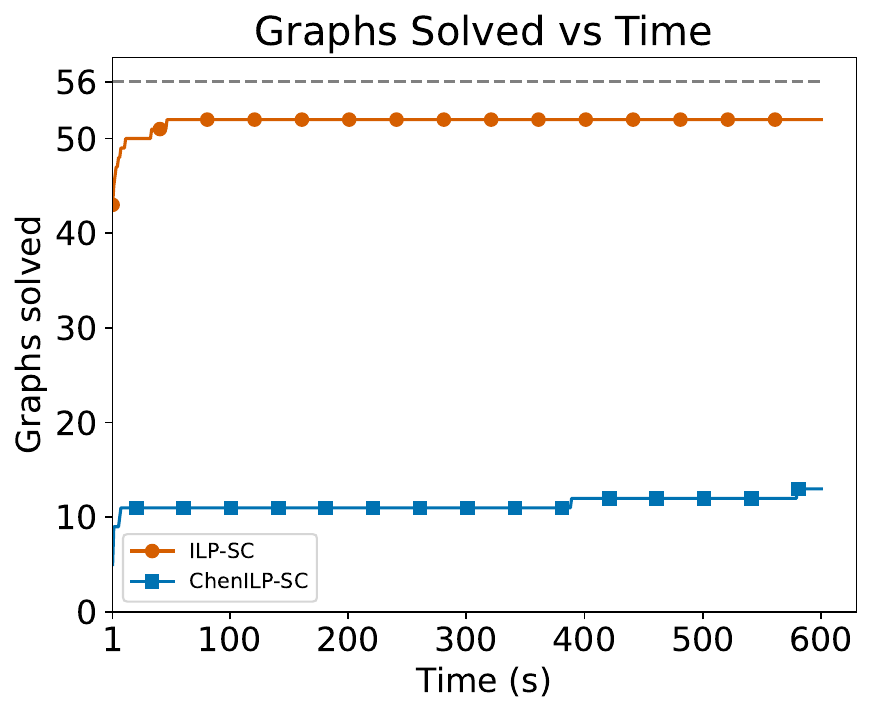}
        \caption{}
        \label{fig:scarcity-ilp:graphs-vs-time}
    \end{subfigure}
    \hfill
    \begin{subfigure}[b]{0.32\textwidth}
        \includegraphics[width=\textwidth]{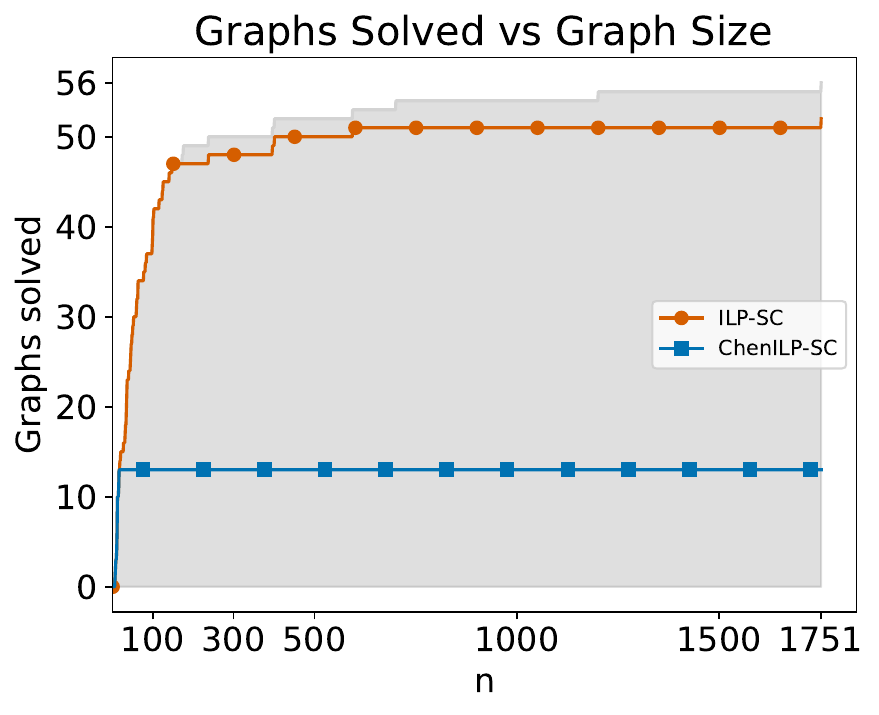}
        \caption{}
        \label{fig:scarcity-ilp:graphs-vs-n}
    \end{subfigure}
    \hfill
    \begin{subfigure}[b]{0.32\textwidth}
        \includegraphics[width=\textwidth]{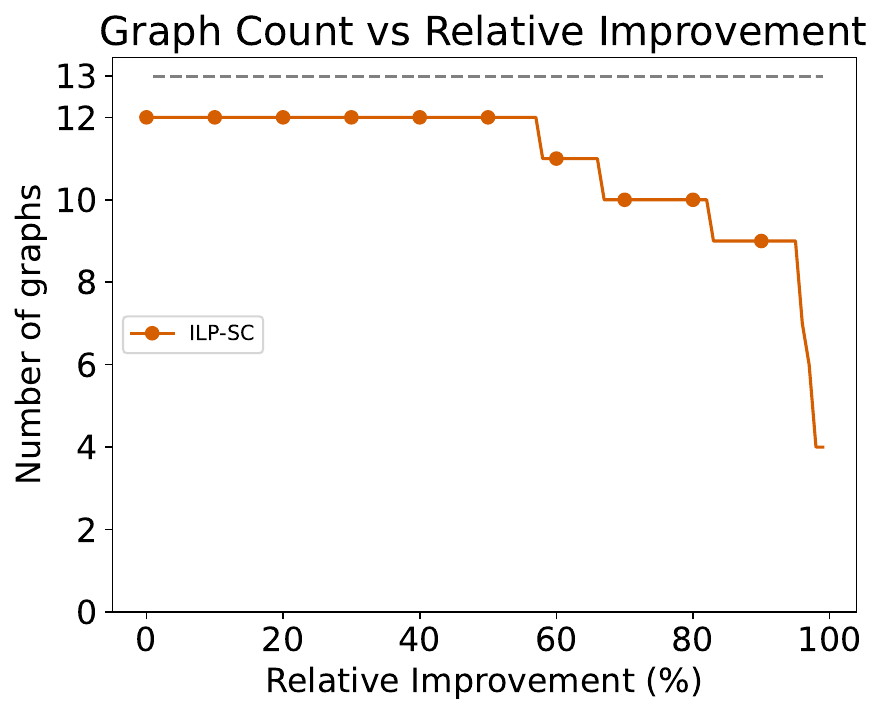}
        \caption{}
        \label{fig:scarcity-ilp:speedups}
    \end{subfigure}
    \caption{Accumulation plots for ILPs for~\StrOJA{} (top) and~\MinEC{} (bottom). Plots (a) and (d) (left) show the number of graphs solved optimally in a given number of seconds; the gray dashed line indicates the corpus size. Plots (b) and (e) (center) record the number of graphs with at most~$n$ vertices solved within a $10$-minute timeout; the shaded region indicates the full corpus. The discontinuity on the $x$-axis of plot (b) is included because none of the tested ILPs for~\StrOJA{} solve any graphs with more than $237$ vertices. Plots (c) and (f) (right) record the number of graphs solved with at least a given percentage improvement in running time, measured relative to the fastest variant from~\cite{chen2012integer,chen2012scarcity} and restricted to the 14 (top) or 13 (bottom) graphs solved within the timeout by~\cite{chen2012integer,chen2012scarcity}. For a fixed percentage, higher is better.} 
    \label{fig:flops-ilp:combined-plots}
\end{figure*}
\begin{table*}
\caption{Combined results for \StrOJA{} ILPs. Each row corresponds to one data source; the bottom row is the entire corpus. The second column records the number of graphs considered. The next five columns record the number of graphs solved in 10 minutes or less. The final four columns record median percentage improvements in running time, computed with respect to \chenVE{}, restricted to those graphs solved by \chenVE{}. The percentage N/A indicates that no graphs were solved by \chenVE{}.}
\label{tab:ilp-timing-flops-combined}
\centering
{\small\setlength{\tabcolsep}{4pt}\begin{tabular}{ll|lllll|llll}
\toprule
\multicolumn{2}{c}{} & \multicolumn{5}{c}{Graphs Solved (\% of total)} & \multicolumn{4}{c}{Median Relative Improvement \%} \\
Dataset & \# & \chenVE{} & \va{} & \vb{} & \vc{} & \vd{} & \va{} & \vb{} & \vc{} & \vd{} \\
\midrule
\CHMUset{} & 5 & 5 (100) & 5 (100) & 5 (100) & 5 (100) & 5 (100) & 77 & 92 & 90 & 94 \\
\GWset{} & 4 & 2 (50) & 3 (75) & 4 (100) & 3 (75) & 3 (75) & -101 & -117 & -89 & 76 \\
\NLSset{} & 26 & 6 (23) & 12 (46) & 13 (50) & 16 (62) & 22 (85) & 83 & 90 & 90 & 94 \\
\EGset{} & 10 & 0 (0) & 5 (50) & 5 (50) & 5 (50) & 5 (50) & N/A & N/A & N/A & N/A \\
\AGset{} & 11 & 1 (9) & 2 (18) & 2 (18) & 6 (55) & 7 (64) & 27 & 55 & 74 & 72 \\
\midrule
All Graphs & 56 & 14 (25) & 27 (48) & 29 (52) & 35 (62) & 42 (75) & 75 & 88 & 87 & 92 \\
\bottomrule
\end{tabular}}
\end{table*}

\subsection{Datasets.}
We evaluated on a corpus assembled from five collections
of computational graphs.
\CHMUset{} consists of the five graphs considered
in~\cite{chen2012integer,chen2012scarcity}, with $|V| \in [10,16]$ and $|E| \in [12,24]$.
\GWset{} consists of all four nontrivial
graphs from chapter 10 of~\cite{griewank2008evaluating} not already included
in \CHMUset, with $|V| \in [6,34]$ and $|E| \in [6,70]$. 
\NLSset{} has 26 graphs generated from the TEST\_NLS nonlinear
least squares test problems~\cite{testnls}. The original Fortran was translated to C++ and the AD tool CoDiPack~\cite{CoDiPack} was used to capture the computational graphs, with $|V| \in [8,1751]$ and $|E| \in [8,2575]$.
\EGset{} is a collection of ten small instances ($|V| \in [36,1200]$ and $|E| \in [135,5500]$) of the evolution graphs that result
from discretizing a time-dependent partial differential equation using
a 5-point stencil on a torus (two-dimensional grid with periodic
boundary conditions). 
Evolution graphs are used
in~\cite{griewank2008evaluating} as examples where \greedymark{} is an
ineffective heuristic. 
\AGset{} is the set of eleven test problems used in~\cite{lohoff2024optimizing}, generated by modifying Graphax~\cite{graphax} to output dot files from functions implemented in JAX.
For these graphs, $|V| \in [7,140]$ and $|E| \in [8,226]$.

\subsection{ILP Performance.}

We evaluate the ILPs presented in this work against those in~\cite{chen2012integer,chen2012scarcity}, which we refer to as~\chenVE{} (for~\StrOJA{}) and~\chenSC{} (for~\MinEC{}). To give as much credit as possible to~\cite{chen2012integer,chen2012scarcity}, we implemented every variant of the ILPs from those works, and in all cases the statistics presented for~\chenVE{} and~\chenSC{} are the most favorable, e.g., the fastest running time across all variants. 

\Cref{fig:flops-ilp:graphs-vs-time} visualizes the number of graphs solved optimally (for~\StrOJA{}) in given running times. \Cref{fig:scarcity-ilp:graphs-vs-time} gives the same visualization for~\MinEC{}. We imposed a~$10$-minute timeout on all ILPs. This choice was motivated both by computational resource considerations and by the diminishing rate of returns as running times extended.
Overall, \chenVE{} solved 14 out of 56 graphs within the time limit, while every variant of~\ilpVE{} solved at least 27, and~\vd{} solved 42. Similarly,~\chenSC{} solved 13 out of 56 graphs within the time limit, while~\ilpSC{} solved~52. These statistics are broken out per-dataset in~\Cref{tab:ilp-timing-flops-combined,tab:ilp-timing-scarcity-combined}. It is also worth noting that each graph solved by~\chenVE{} was solved by every variant of~\ilpVE{}, and every graph solved by~\chenSC{} was solved by~\ilpSC{}. Also, among those graphs not solved by~\vd{}, only one is solved by any other variant of~\ilpVE{}.

\Cref{fig:flops-ilp:graphs-vs-n} visualizes the scalability of~\StrOJA{} ILPs with respect to graph size. The largest graph solved by~\chenVE{} has~$20$ vertices. The smallest graph not solved by any single variant of~\ilpVE{} has~$31$ vertices, while the largest graph solved by~\vd{} has~$237$ vertices.~\Cref{fig:scarcity-ilp:graphs-vs-n} gives the same visualization for~\MinEC{}. The largest graph solved by~\chenSC{} has~$16$ vertices, while the smallest graph not solved by~\ilpSC{} has~$172$ vertices. The largest graph solved by~\ilpSC{} has~1751 vertices; this is the maximum in our corpus.

\Cref{fig:flops-ilp:speedups} depicts relative improvements in running time\footnote{\label{footnote:relative-improvement}For two methods $A$ and $B$ with running times $t_A, t_B$, the \emph{relative improvement} of method $A$ with respect to method $B$ is computed as $100\cdot(t_B - t_A)/t_B$.} for our~\ilpVE{} variants, computed with respect to~\chenVE{} and restricted to those graphs solved by~\chenVE{} within the time limit.~\Cref{fig:scarcity-ilp:speedups} gives the same visualization for~\ilpSC{} computed with respect to~\chenSC{}. Median relative improvements are presented per-datasource in~\Cref{tab:ilp-timing-flops-combined,tab:ilp-timing-scarcity-combined}. On half of the graphs solved by~\chenVE{}, using~\vd{} instead reduces running time by at least $92\%$. For~\MinEC{}, the analogous statistic is~$96\%$.
Each variant of~\ilpVE{} and~\ilpSC{} has one graph on which \chenVE{} or \chenSC{} ran faster, i.e., on which the relative improvement statistic is negative. These graphs have at most $10$ vertices, and all ILPs solve them in less than~$350$ milliseconds.

\begin{table}[t]
\caption{Combined results for \MinEC{} ILPs. See the caption of Table~\ref{tab:ilp-timing-flops-combined}.}
\label{tab:ilp-timing-scarcity-combined}
\centering
{\small\setlength{\tabcolsep}{4pt}\begin{tabular}{ll|ll|@{\hspace{8pt}}l}
\toprule
\multicolumn{2}{c}{} & \multicolumn{2}{c}{Graphs Solved (\%)} & Med. RI \% \\
Dataset & \# & \chenSC{} & \ilpSC{} & \ilpSC{} \\
\midrule
\CHMUset{} & 5 & 5 (100) & 5 (100) & 96 \\
\GWset{} & 4 & 2 (50) & 4 (100) & 82 \\
\NLSset{} & 26 & 5 (19) & 25 (96) & 97 \\
\EGset{} & 10 & 0 (0) & 7 (70) & N/A \\
\AGset{} & 11 & 1 (9) & 11 (100) & 83 \\
\midrule
All Graphs & 56 & 13 (23) & 52 (93) & 96 \\
\bottomrule
\end{tabular}}
\end{table}

\subsection{Heuristic Performance.}\label{sec:heuristic-performance}
\begin{figure*}[t!]
    \centering
    \begin{subfigure}[b]{0.66\textwidth}
        \includegraphics[width=\textwidth]{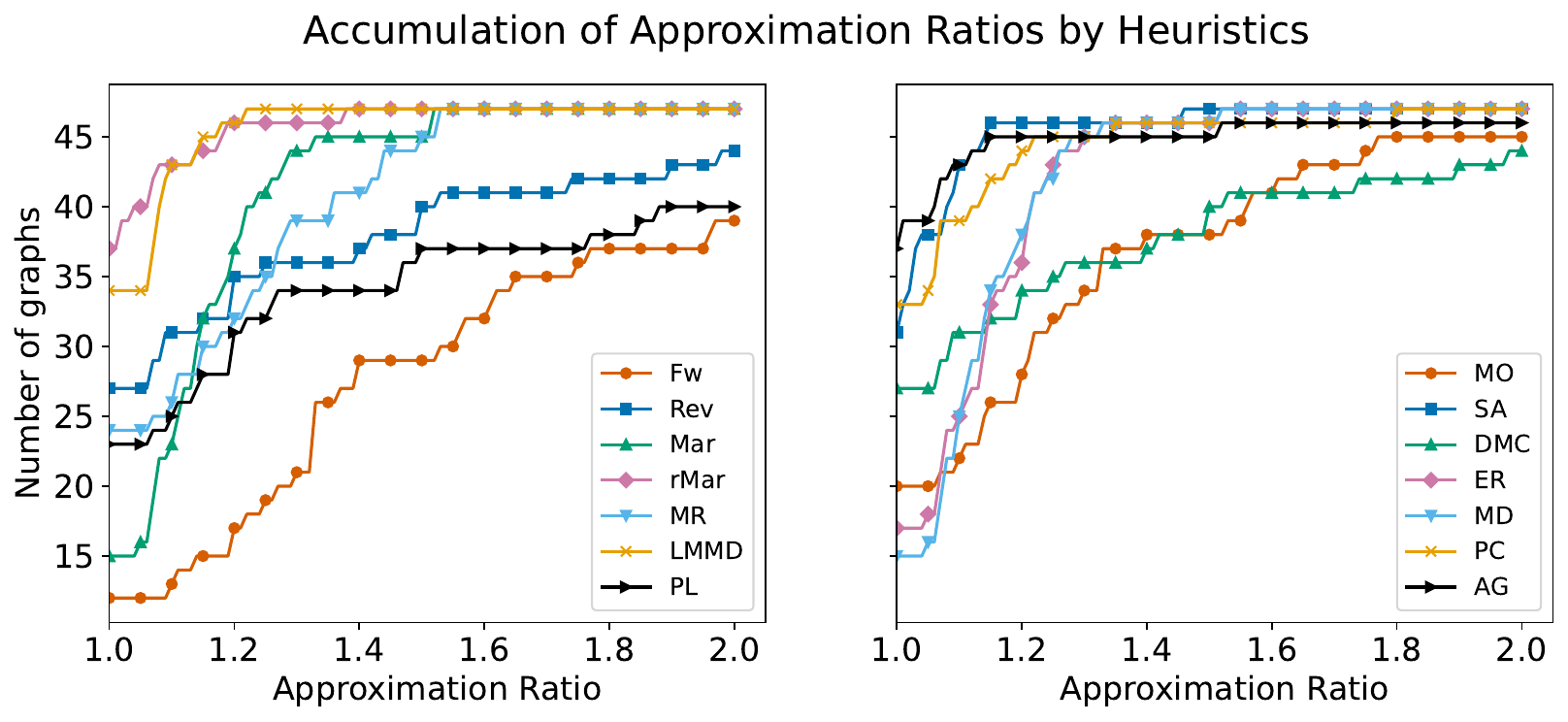}
        \caption{}
        \label{fig:flops-heur:accumulation-plots}
    \end{subfigure}
    \hfill
    \begin{subfigure}[b]{0.33\textwidth}
        \includegraphics[width=\textwidth]{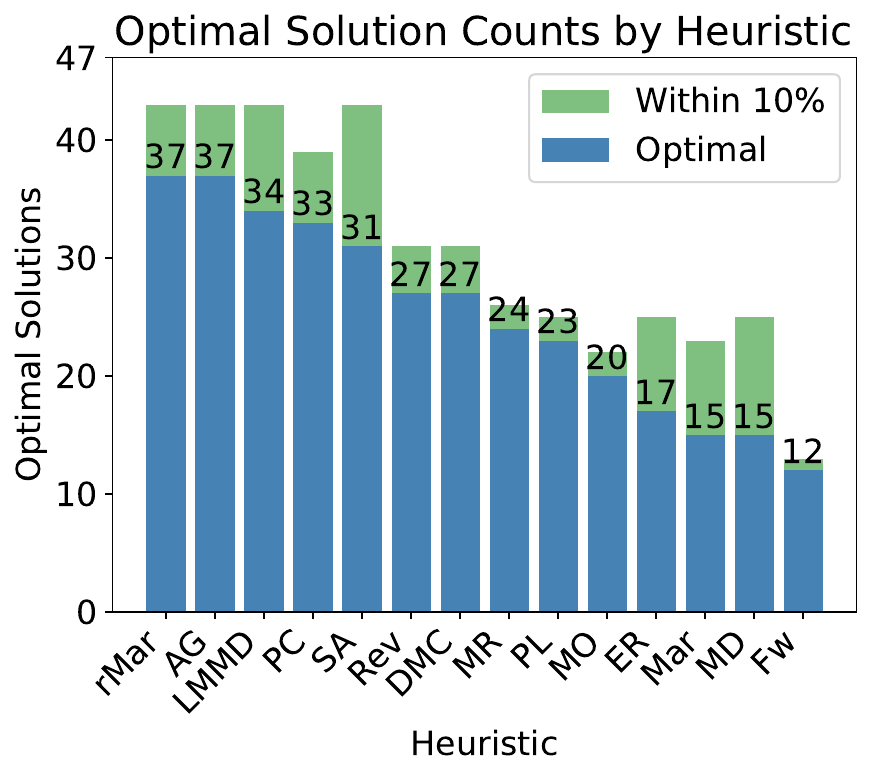}
        \caption{}
        \label{fig:flops-heur:barchart}
    \end{subfigure}

    \vspace{1em}

    \begin{subfigure}[b]{0.66\textwidth}
        \includegraphics[width=\textwidth]{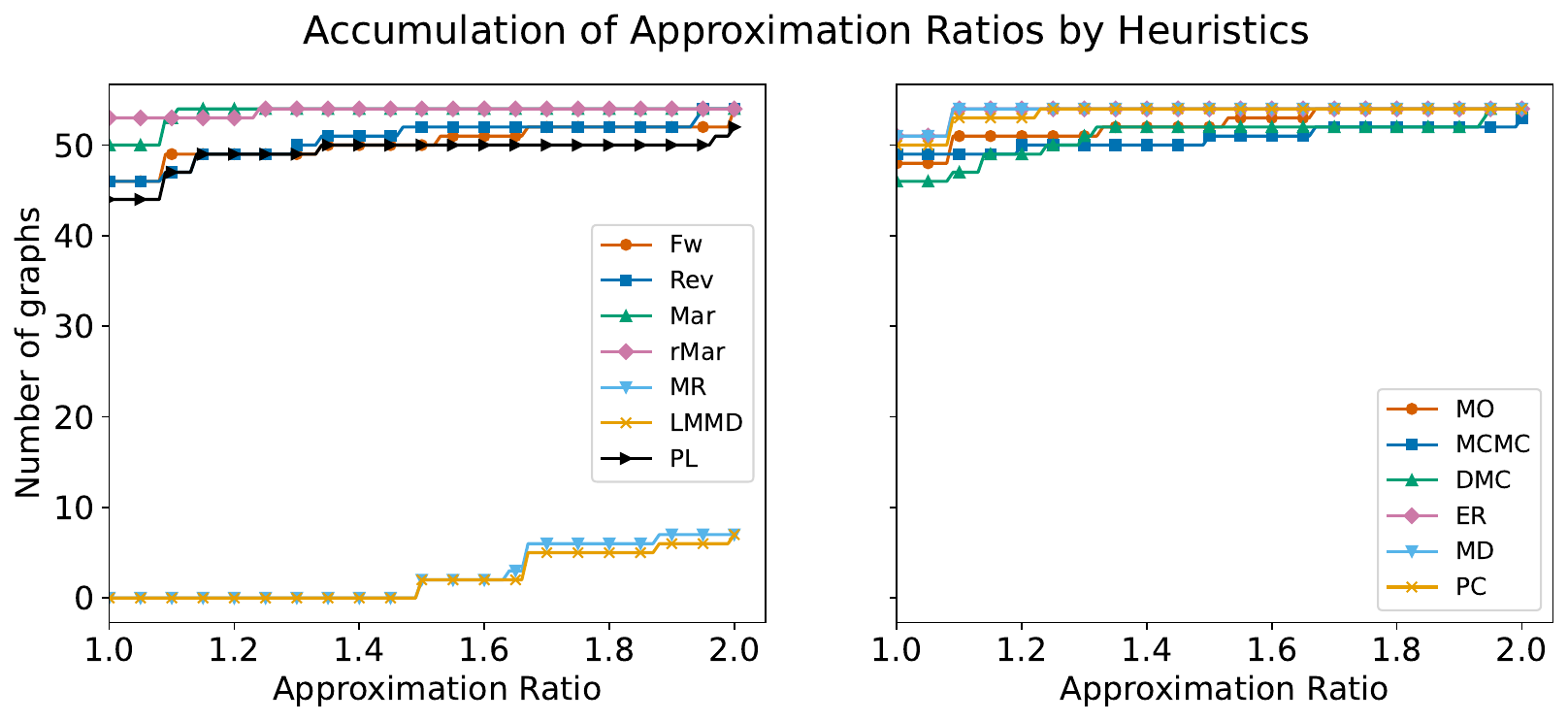}
        \caption{}
        \label{fig:scarcity-heur:accumulation-plots}
    \end{subfigure}
    \hfill
    \begin{subfigure}[b]{0.33\textwidth}
        \includegraphics[width=\textwidth]{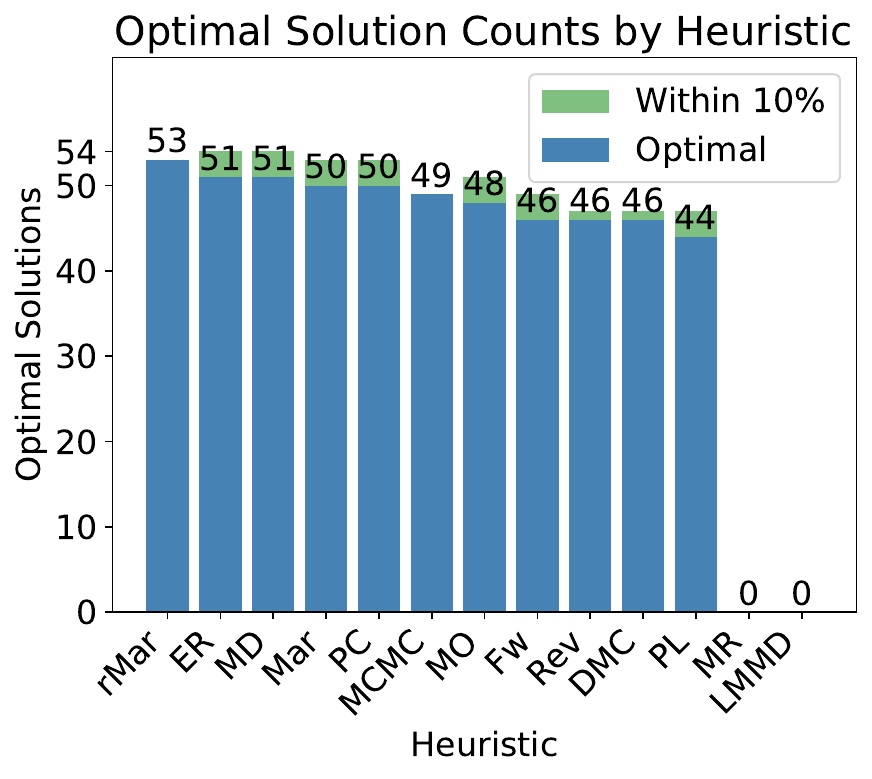}
        \caption{}
        \label{fig:scarcity-heur:barchart}
    \end{subfigure}
    \caption{Heuristic performances for~\StrOJA{} (top) and~\MinEC{} (bottom). Plots (a) and (c) (left) are accumulation plots showing, for each heuristic, the number of graphs on which a given approximation ratio is achieved. Plots (b) and (d) (right) present the evaluated heuristics in sorted order according to the number of graphs for which they achieve an optimal solution. Bars are extended to indicate additional graphs on which the heuristics achieve an approximation ratio at most~$1.10$. All plots are restricted to those graphs for which we were able to compute optimal solutions via integer programming. There are $47$ such graphs for~\StrOJA{} (top), and $54$ for~\MinEC{} (bottom).}
    \label{fig:heuristics:plots}
\end{figure*}

By extending the timeouts and utilizing multithreading, we were able to compute optimal solutions for an additional $4$ graphs for~\StrOJA{}, and an additional $2$ graphs for~\MinEC{}. See Appendix~\ref{appendix:experiments} for details. This gives us a total of $47$ and $54$ graphs, respectively, for which we know optimal solutions. In this section we restrict our attention to these graphs.

Throughout this evaluation, we will primarily be concerned with the observed approximation ratios. For completeness, running times for every evaluated heuristic on every graph in the corpus are included in Appendix~\ref{appendix:experiments}. The running time asymptotics of the evaluated heuristics are generally well-understood, and we do not claim that our reimplementations are optimized for speed. Still, this data may be useful for some coarse qualitative comparisons; we discuss this point below.

\Cref{fig:flops-heur:accumulation-plots} plots the observed approximation ratios for~\StrOJA{} heuristics, while~\Cref{fig:flops-heur:barchart} ranks the heuristics by the number of graphs on which they achieve an optimal solution. \Cref{fig:scarcity-heur:accumulation-plots,fig:scarcity-heur:barchart} give the same visualizations for~\MinEC{}. 
Encouragingly, most heuristics achieve optimal solutions on a sizable portion of our corpus; $8$ of the $14$~\StrOJA{} heuristics achieve optimal solutions on more than half of the corpus, while $11$ out of $13$~\MinEC{} heuristics achieve the same.

\begin{figure*}[t!]
    \centering
    \begin{subfigure}[b]{0.49\textwidth}
        \includegraphics[width=\textwidth]{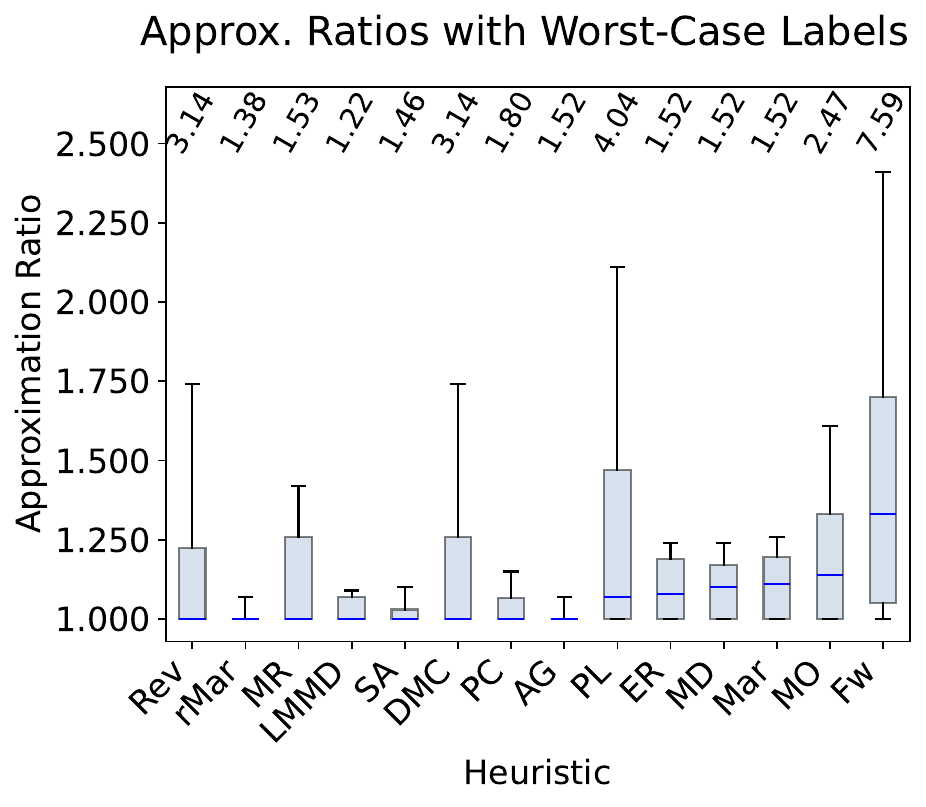}
        \label{fig:flops-heur:boxplot}
    \end{subfigure}
    \hfill
    \begin{subfigure}[b]{0.49\textwidth}
        \includegraphics[width=\textwidth]{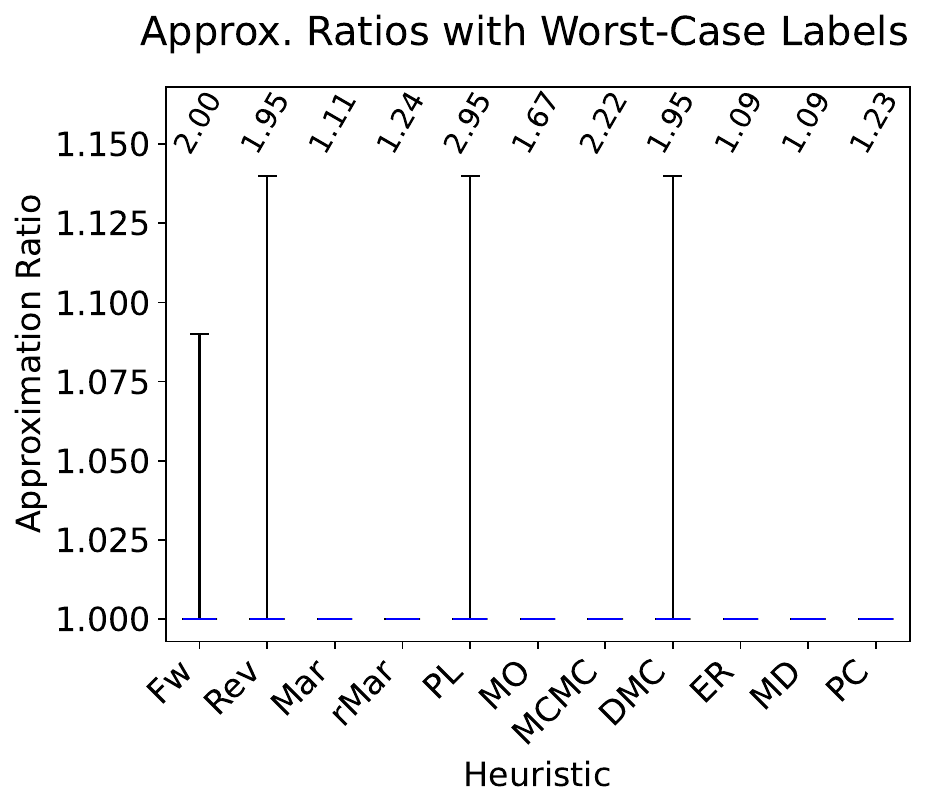}
        \label{fig:scarcity-heur:boxplot}
    \end{subfigure}
    \caption{Box-and-whisker plots for observed approximation ratios for~\StrOJA{} (left) and~\MinEC{} (right). The heuristics are sorted by median observed approximation ratio. The boxes extend between the 25th and 75th percentiles, while the whiskers extend between the 10th and 90th percentiles. To record outliers without compressing the visualization, each heuristic is labeled by its worst observed approximation ratio, above the upper whisker. To improve readability, the poorly performing heuristics~\lmmd{} and~\mtmr{} are omitted from the~\MinEC{} plot (right); this point is discussed in the text.}
    \label{fig:boxplots}
\end{figure*}
Another measure of success is the number of graphs on which a \emph{near-optimal} solution is achieved. See~\Cref{fig:boxplots} for box-and-whisker plots, and~\Cref{table:flops-heuristics-combined,table:scarcity-heuristics-combined} for arithmetic mean approximation ratios, presented per-dataset.
For~\StrOJA{}, every heuristic achieves a $2$-approximation or better on at least 39/47 (82\%) of graphs, and more than half of the heuristics (9/14) have a \emph{worst} observed approximation ratio less than~$2$.

For~\MinEC{}, 11 out of $13$ heuristics actually achieve an optimal solution on at least $44/54$ (81\%) of graphs. Among these~$11$ heuristics, the worst-case observed approximation ratio is $2.95$. 
There are two outliers, \lmmd{} and \mtmr{}, which achieve median approximation ratios of $4.19$ and $5.04$, respectively, while their worst-case observed approximation ratios are $30.83$ and $31$, respectively. A concrete recommendation from our evaluation is therefore to avoid using these heuristics, both originally designed for~\StrOJA{}, for~\MinEC{}.

For~\StrOJA{}, the performance of the deep-learning method~\ag{}~\cite{lohoff2024optimizing} is quite strong, but there are two caveats worth mentioning.
The first is that the method in fact uses~\fo{},~\re{}, and~\greedymark{} as an ensemble before beginning its learning iterations. While~\ag{} reports~$37$ optimal solutions, this simple ensemble is already responsible for~$31$. Similarly, the median, 75th percentile, and worst-case statistics for~\ag{} are $1.0, 1.0$, and $1.52$, while the same statistics for the simple ensemble are $1.0, 1.07$, and $1.52$.
The second caveat is related to running time.~\ag{} is significantly slower than the other heuristics, often requiring several hours of GPU time even on the small graphs in our corpus (see Appendix~\ref{appendix:experiments}). We were also unable to obtain a solution (using the software of~\cite{lohoff2024optimizing}) from~\ag{} for one of the larger graphs in our corpus (395 vertices, or more than twice as many as the largest graph from~\cite{lohoff2024optimizing}); this graph is omitted from the statistics presented for~\ag{}. 

However, the~\ag{} method was designed for the scenario in which the time spent solving~\StrOJA{} may be amortized over numerous derivative evaluations using the computed elimination ordering. We also acknowledge that~\ag{} was initially designed for a more general setting than the one considered in this paper, where the edge-weights on a computational DAG may be non-scalar. The results of our evaluation do not have direct implications for that setting.

\smallskip
\noindent\textbf{Evolution Graphs}. We note that the performance of \greedymark{} on evolution graphs is somewhat deceptive, because we can only compute optimal solutions on small evolutions. It is known that the performance of \greedymark{} on evolutions degrades as the graph increases in size. For example, though we cannot compute a minimum-cost elimination sequence on a two-dimensional $20\times20\times20$ evolution ($8800$ vertices), we can observe that the solution \greedymark{} computes is $2.69$ times as expensive as \fo{}. Similarly, for a planar evolution derived from a $9$-point stencil (see~\cite[Chapter 10]{griewank2008evaluating}), \greedymark{} produces a solution $2.27$ times as expensive as that produced by \fo{}. On the latter graph, \relmark{} also fares poorly ($2.14$ times worse than \fo{}), and we acknowledge that the performance of \relmark{} on the evolution graphs in this paper is buoyed by use of the $5$-point stencil as a template; \relmark{} is known to perform poorly for other stencil choices~\cite{griewank2008evaluating}.
It has been observed previously that \fo{} and \re{} perform very well on evolutions (of multiple stencils)~\cite{griewank2008evaluating}; we have found no evidence to the contrary, and it would be interesting to pursue a theoretical justification.

Finally, while the presence of evolutions in our corpus may inflate the apparent performance of~\greedymark{} and~\relmark{}, we observe that they may partially obscure the performance of~\mcmc{} for~\MinEC{}. This heuristic solves every graph in our corpus optimally, with the exception of evolutions.
\begin{table*}[t]
\caption{Mean approximation ratios for \StrOJA{} heuristics, arranged per-dataset. Entries marked with an asterisk are discussed further in the text; see Section~\ref{sec:heuristic-performance}.}
\label{table:flops-heuristics-combined}
\centering
{\small\setlength{\tabcolsep}{3pt}\begin{tabular}{ll|llllllllllllll}
\toprule
Dataset & \# & \fo{} & \re{} & \greedymark{} & r\greedymark{} & \mtmr{} & \lmmd{} & \pathlen{} & \moshort{} & \sa{} & \diffmincost{} & \edgereduction{} & \markanddegree{} & PC & \ag{} \\
\midrule
\CHMUset{} & 5 & 1.16 & 1.24 & 1.09 & 1.00 & 1.16 & 1.09 & 1.19 & 1.10 & 1.02 & 1.25 & 1.09 & 1.09 & 1.25 & 1.00 \\
\GWset{} & 4 & 1.54 & 1.28 & 1.08 & 1.10 & 1.40 & 1.08 & 1.65 & 1.10 & 1.02 & 1.28 & 1.05 & 1.08 & 1.05 & 1.03 \\
\NLSset{} & 24 & 1.52 & 1.35 & 1.13 & 1.02 & 1.10 & 1.02 & 1.35 & 1.40 & 1.04 & 1.35 & 1.11 & 1.11 & 1.02 & 1.03 \\
\EGset{} & 6 & 1.00 & 1.00 & 1.03* & 1.00* & 1.22 & 1.03 & 1.03 & 1.00 & 1.00 & 1.00 & 1.03 & 1.03 & 1.03 & 1.00 \\
\AGset{} & 8 & 2.77 & 1.05 & 1.20 & 1.04 & 1.04 & 1.00 & 1.73 & 1.08 & 1.04 & 1.05 & 1.20 & 1.19 & 1.04 & 1.03 \\
\midrule
All & 47 & 1.63 & 1.24 & 1.12 & 1.03 & 1.14 & 1.03 & 1.38 & 1.24 & 1.03 & 1.24 & 1.11 & 1.11 & 1.05 & 1.02 \\
\bottomrule
\end{tabular}}
\end{table*}

\begin{table*}[t]
\caption{Mean approximation ratios for \MinEC{} heuristics, arranged per-dataset.}
\label{table:scarcity-heuristics-combined}
\centering
{\small\setlength{\tabcolsep}{4pt}\begin{tabular}{ll|llllllllllllll}
\toprule
Dataset & \# & \fo{} & \re{} & \greedymark{} & r\greedymark{} & \mtmr{} & \lmmd{} & \pathlen{} & \moshort{} & \mcmc{} & \diffmincost{} & \edgereduction{} & \markanddegree{} & PC \\
\midrule
\CHMUset{} & 5 & 1.04 & 1.03 & 1.04 & 1.00 & 4.10 & 3.93 & 1.06 & 1.04 & 1.00 & 1.03 & 1.04 & 1.04 & 1.04 \\
\GWset{} & 4 & 1.02 & 1.06 & 1.02 & 1.00 & 5.36 & 4.07 & 1.06 & 1.02 & 1.00 & 1.06 & 1.02 & 1.02 & 1.02 \\
\NLSset{} & 26 & 1.12 & 1.10 & 1.00 & 1.00 & 3.68 & 3.28 & 1.23 & 1.05 & 1.00 & 1.09 & 1.00 & 1.00 & 1.00 \\
\EGset{} & 8 & 1.00 & 1.00 & 1.00 & 1.00 & 12.13 & 10.72 & 1.00 & 1.00 & 1.45 & 1.00 & 1.00 & 1.00 & 1.00 \\
\AGset{} & 11 & 1.03 & 1.03 & 1.01 & 1.02 & 10.60 & 9.97 & 1.03 & 1.03 & 1.00 & 1.03 & 1.00 & 1.00 & 1.02 \\
\midrule
All & 54 & 1.07 & 1.06 & 1.01 & 1.00 & 6.50 & 5.86 & 1.12 & 1.03 & 1.07 & 1.06 & 1.00 & 1.00 & 1.01 \\
\bottomrule
\end{tabular}}
\end{table*}

\section{Approximation Lower Bounds.}\label{sec:approximation-hardness}

The decision problems underlying \StrOJA{} and \MinEC{} have only recently been shown to be \cclass{NP}-complete~\cite{bentert2025structural}, and little else is known about their complexity.
We now establish the first approximation lower bounds for \StrOJA{}, \MinEC{}, and \MaxER{}.
For all three problems, we exclude a polynomial-time approximation scheme (PTAS) unless $\cclass{P} = \cclass{NP}$. In each case, the hardness holds even in bounded-degree DAGs; for~\MinEC{}, \cclass{NP}-hardness was not previously known in bounded-degree DAGs. Finally, we exclude a sub-polynomial approximation for~\MaxER{}, unless $\cclass{NP} \subseteq \cclass{BPP}$.
See Appendix~\ref{appendix:proofs:approximation-hardness} for complete proofs.

The existing \cclass{NP}-completeness reduction for~\StrOJA{} is from~\VC{}~\cite{bentert2025structural}. 
We extend the analysis, showing that approximations are preserved when the input graph is cubic.\footnote{A graph is \emph{cubic} if every vertex has degree~$3$. \VC{} remains \APX{}-hard in cubic graphs~\cite{alimonti1997hardness}.}\looseness=-1

\begin{restatable}{theorem}{minflopsapxhard}\label{thm:oja-apx-hard}
  \StrOJA{} is \APX-hard, even when restricted to DAGs with max. in-degree 7, max. out-degree 6, and max. total-degree 9.
\end{restatable}

Henceforth, we focus on \MinEC{} and \MaxER{}, for which our constructions deviate from those of~\cite{bentert2025structural}.
First, we need an intermediate result.

\begin{restatable}{theorem}{maxisapxhardrestricted}\label{thm:maxIS-apx-hard-large-girth}
  $\MaxIS{}$ is \APX-hard, even when restricted to graphs with max. degree $3$, min. degree $2$, girth at least $g$ for any constant $g \geq 3$, and no pair of degree-$3$ vertices adjacent.
\end{restatable}

Our reduction makes use of a classic idea, e.g., as seen in~\cite{komusiewicz2018tight}. Given an instance of \MaxIS{}, we subdivide edges repeatedly to ensure large girth in the resulting graph.
Our contribution is to show that the reduction is approximation-preserving when the input graph is cubic.

We are now ready to sketch our proof that both \MinEC{} and \MaxER{} are \APX-hard. 
Our construction is inspired by that of~\cite{bentert2025structural}, but we have made some modifications necessary to prove hardness even in bounded-degree DAGs.

\begin{figure}[t]
  \centering
  \resizebox{\columnwidth}{!}{%
  \begin{tikzpicture}
    \node[circle, draw] (v) at (0, 0) {$v'$};
    \node[circle, draw, minimum size=1cm, fill=blue!20, label=above:\large$4$] (iv) at (-2, 2) {\large$I_v$};
    \node[circle, draw, minimum size=1cm, fill=blue!20, label=above:\large$6$] (ov) at (2, 2) {\large$O_v$};
    \node[] (n1) at (2.25, 1) {};
    \node[] (n2) at (2.25, 0) {};
    \node[] (n3) at (2.25, -1) {};
    \draw[->,line width=1.5pt] (iv)--node[above, yshift=2pt]{\large$4$}(v);
    \draw[->,line width=1.5pt] (iv)--node[above, yshift=2pt]{\large$24$}(ov);
    \draw[->,line width=1.5pt] (v)--node[above, xshift=-2pt]{\large$6$}(ov);
    \draw[->,dashed] (v)--(n1);
    \draw[->,dashed] (v)--(n2);
    \draw[->,dashed] (v)--(n3);
    \node[shift={(6,0)}, circle, draw] (v) at (0, 0) {$v'$};
    \node[shift={(6,0)}, circle, draw, minimum size=1cm, fill=blue!20, label=above:\large$4$] (iv) at (-2, 2) {\large$I_v$};
    \node[shift={(6,0)}, circle, draw, minimum size=1cm, fill=blue!20, label=above:\large$4$] (ov) at (2, 2) {\large$O_v$};
    \node[shift={(6,0)}] (n1) at (-2.25, .75) {};
    \node[shift={(6,0)}] (n3) at (-2.25, -.75) {};
    \draw[->,shift={(6,0)}, line width=1.5pt] (iv)--node[above, yshift=2pt]{\large$4$}(v);
    \draw[->,shift={(6,0)}, line width=1.5pt] (iv)--node[above, yshift=2pt]{\large$15$}(ov);
    \draw[->,shift={(6,0)}, line width=1.5pt] (v)--node[above, xshift=-2pt]{\large$4$}(ov);
    \draw[->,shift={(6,0)}, dashed] (n1)--(v);
    \draw[->,shift={(6,0)}, dashed] (n3)--(v);
    \node[shift={(12,0)}, circle, draw] (v) at (0, 0) {$v'$};
    \node[shift={(12,0)}, circle, draw, minimum size=1cm, fill=blue!20, label=above:\large$4$] (iv) at (-2, 2) {\large$I_v$};
    \node[shift={(12,0)}, circle, draw, minimum size=1cm, fill=blue!20, label=above:\large$4$] (ov) at (2, 2) {\large$O_v$};
    \node[shift={(12,0)}] (n1) at (2.25, 0) {};
    \node[shift={(12,0)}] (n2) at (-2.25, 0) {};
    \draw[->,shift={(12,0)}, line width=1.5pt] (iv)--node[above, yshift=2pt]{\large$4$}(v);
    \draw[->,shift={(12,0)}, line width=1.5pt] (iv)--node[above, yshift=2pt]{\large$16$}(ov);
    \draw[->,shift={(12,0)}, line width=1.5pt] (v)--node[above, xshift=-2pt]{\large$4$}(ov);
    \draw[->,shift={(12,0)}, dashed] (v)--(n1);
    \draw[->,shift={(12,0)}, dashed] (n2)--(v);
  \end{tikzpicture}%
  }
  \caption{\label{fig:scarcity-apx-hard-gadget}Gadgets (see the construction in~\Cref{sec:approximation-hardness}) for vertices of type~$1$ (left), 2a (center), and 3 (right). The vertex $v'$ is a replica vertex, and the sets $I_v$ and $O_v$ contain auxiliary vertices; the cardinalities are indicated above. Bolded arrows represent multiple edges. Dashed arrows represent edges to or from the replica vertices in other gadgets. The gadget for vertices of type 2b is similar to that for type 2a, except that the replica vertex has two edges to, rather than from, other gadgets.}
\end{figure}
\smallskip
\noindent\textbf{Construction:}
See~\Cref{fig:scarcity-apx-hard-gadget}.
Let $G = (V, E)$ be a graph with maximum degree $3$, minimum degree $2$, girth at least $5$, and no pair of degree-$3$ vertices adjacent. 
Let $\pi$ be a total order on $V$ such that every vertex of degree $3$ precedes every vertex of degree $2$. For two vertices $u, v \in V$, we write $u <_\pi v$ if $u$ precedes $v$ according to $\pi$.
We define $\deg^-_G(v)$ as the number of neighbors $u$ of $v$ with $u <_\pi v$, and define $\deg^+_G(v)$ similarly. 
We partition the vertices $V$ into~$4$ types. Type 1 vertices have degree $3$, so all neighbors appear later in $\pi$ (recall that no pair of degree-$3$ vertices is adjacent).
Type 2a and 2b vertices have degree $2$, with both neighbors appearing earlier (type 2a) or later (type 2b) in $\pi$.
Type 3 vertices have degree $2$, with one neighbor appearing earlier and one neighbor appearing later in $\pi$.

We construct a DAG $D$ as follows. For each $v \in V$, we create a \emph{replica} vertex $v'$, as well as two sets $I_v$ and $O_v$ of \emph{auxiliary} vertices.
The cardinalities of $I_v$ and $O_v$ depend on the type of $v$; for vertices of type $1$, $|I_v| = 4$ and $|O_v| = 6$; for vertices of types 2a, 2b, and 3, $|I_v| = |O_v| = 4$.
We add arcs from every vertex in $I_v$ to $v'$, and from $v'$ to every vertex in $O_v$. Finally, we add a type-specific number of arcs from $I_v$ to $O_v$. Only the number matters; it does not matter which arcs we choose. The numbers of arcs from $I_v$ to $O_v$ are $24, 15, 15$, and $16$ for vertices of types 1, 2a, 2b, and 3, respectively.
Finally, for each edge $uv \in E$ we add the arc $(u', v')$ if $u <_\pi v$, or $(v', u')$ if $v <_\pi u$.
This completes the construction of $D$.

\smallskip
Observe that $D$ has maximum in-degree $6$, maximum out-degree $9$, and maximum total degree $13$. Also, all auxiliary vertices are either sources or sinks, and every replica vertex is an internal vertex.
To analyze the reduction, we will need several supporting results. 
Similarly to~\cite{bentert2025structural}, for \MinEC{} or \MaxER{} on input $D$ we say that a \emph{nice solution} is a set $X$ of replica vertices such that for every $v' \in X$, $D_{X \setminus \{v'\}}$ has at least as many arcs as $D_X$.
The key step in the analysis is to show a connection between independent sets and nice solutions. For the following three claims, let $G$ be a graph with maximum degree $3$, minimum degree $2$, girth at least $5$, and no pair of degree-$3$ vertices adjacent, and let $D$ be the DAG generated by the construction described above.

\begin{restatable}{claim}{nicesolutionclaim}\label{claim:scarcity-reduction-nice-solutions}
  Let $X$ be a subset of the internal vertices of $D$. If $X$ is a nice solution, then it is an independent set in $D$ and has cardinality $|E(D)| - |E(D_X)|$. Further, every independent subset of the internal vertices of $D$ is a nice solution.
\end{restatable}

Claim~\ref{claim:scarcity-reduction-nice-solutions} follows from an analysis of our construction: we show that no nice solution contains two adjacent replica vertices, and that if the set of eliminated vertices is independent, then every elimination reduces the total number of edges by exactly~$1$.

Notice that given any solution $X$, a nice solution $X' \subseteq X$ may be obtained in polynomial time by greedily identifying and removing vertices which cause $X$ to violate the definition of nice solution. 
Combining this procedure with Claim~\ref{claim:scarcity-reduction-nice-solutions}, we produce independent sets in $G$ given vertex elimination sets in $D$.
\begin{restatable}{claim}{produceisclaim}\label{claim:scarcity-apx-produce-is}
  Let $X$ be a subset of the internal vertices of $D$, and let $k = |E(D)| - |E(D_{X})|$. Then it is possible, in polynomial time, to produce an independent set in $G$ of size at least $k$.
\end{restatable}

We write $\optGmin{D}$ and $\optGmax{D}$, respectively, for the optimum objective values for $\MinEC{}$ and $\MaxER{}$ on input $D$.
Again using~\Cref{claim:scarcity-reduction-nice-solutions}, we now relate these quantities to the independence number $\alpha(G)$ of $G$.\looseness=-1
\begin{restatable}{claim}{relateoptsclaim}\label{claim:scarcity-apx-relate-opts}
  $\optGmax{D} = \alpha(G)$, and $\optGmin{D} = |E(D)| - \alpha(G)$.
\end{restatable}

The results follow from Claims~\ref{claim:scarcity-reduction-nice-solutions}, \ref{claim:scarcity-apx-produce-is}, and~\ref{claim:scarcity-apx-relate-opts}.
\begin{restatable}{theorem}{maxedgereductionapxhard}\label{thm:scarcity-maximization-apx-hard}
  \MaxER{} is \APX-hard, even when restricted to DAGs with max. in-degree~$6$, max. out-degree~$9$, and max. total-degree $13$.
\end{restatable}

\begin{restatable}{theorem}{minedgecountapxhard}\label{thm:scarcity-minimization-apx-hard}
  \MinEC{} is \APX-hard, even when restricted to DAGs with max. in-degree~$6$, max. out-degree~$9$, and max. total-degree $13$.
\end{restatable}

Our last result is a much stronger lower bound for~\MaxER{}.

\begin{restatable}{theorem}{maxedgereductionpolyapxbound}\label{thm:scarcity-poly-apx-hardness}
  There exists a constant $c > 0$ such that, unless $\NP{} \subseteq \BPP$, there is no polynomial-time $n^{c}$-approximation for \MaxER{}.
\end{restatable}

We show how to generalize certain aspects of the construction in~\cite{bentert2025structural} so that we may reduce from \MaxIS{} in graphs of girth at least $5$, i.e., we eliminate the restrictions on vertex degrees.
By~\cite{bonnet2020algorithmic}, a suitable approximation barrier exists for \MaxIS{}, even when restricted to graphs of girth $5$.

\section{Conclusion}
We have introduced and implemented novel ILP formulations for~\StrOJA{} and~\MinEC{}, enabling the assembly of a diverse corpus of graphs on up to $\sim 1000$ vertices for which optimal solutions are known.
We used this corpus to evaluate a wide range of heuristics, presenting observed approximation ratios, which we hope will help guide the selection of heuristics used in practice.
On the theoretical side, we have given a tight analysis of the forward and reverse modes of AD, and have extended this analysis to provide a simple parameterized approximation algorithm. 
We also established the first approximation lower bounds for vertex elimination problems arising in AD.

Several interesting questions remain open, including the design of rounding schemes for our ILPs in restricted settings (recall Section~\ref{sec:integrality-gaps}), the tightness of the approximation ratio for~\mo{} (recall~\Cref{thm:middle-out}), and a theoretical analysis of the forward and reverse modes of AD on evolution graphs.
Indeed, we find it surprising that these highly structured graphs are not better understood.
Finally, while our results exclude polynomial-time approximation schemes, the (non) existence of constant-factor approximations remains open.

\section*{Acknowledgments}
We thank Jakob Rødal Skaar who contributed some of the code for the experiments.

This work was supported in part by the U.S. Department of Energy, Office of Science, Advanced Scientific Computing Research, under contract number DE-AC02-06CH11357.

Blair D. Sullivan gratefully acknowledges partial financial support for this research by the Fulbright Program, which is sponsored by the U.S. Department of State and the Franco-American Commission -- Fulbright France. The contents of this work are solely the responsibility of the authors and do not necessarily represent the official views of the Fulbright Program, the Government of the United States, or the Franco-American Commission.

\bibliographystyle{siamplain}
\bibliography{refs}

\newpage
\appendix
\section{Deferred Proofs}\label{appendix:proofs}

\subsection{Deferred Proofs from~\Cref{sec:ilp}}\label{appendix:proofs:ilps}

\pathcontraction*
\begin{proof}
    If $|X| = 1$, then the claim is true by definition of vertex elimination, so suppose $|X| > 1$. Further, we may assume that $(i, j) \not\in E$, as otherwise the claim is vacuously true.
    If there is a directed path $P$ from $i$ to $j$ for which every internal vertex is in $X$, then let $v \in X$ be $i$'s out-neighbor along $P$. By induction, the arc $(x, j)$ exists in $D_{X \setminus v}$, and thus by definition of vertex elimination the arc $(i, j)$ exists in $D_X$.
    For the other direction, suppose that every directed path (if any exist) from $i$ to $j$ contains at least one internal vertex not in $X$, and let $v$ be an arbitrary vertex in $X$. Note that either there is no directed path from $i$ to $v$ with every internal vertex in $X$, or there is no directed path from $v$ to $j$ with every internal vertex in $X$; otherwise, concatenating these paths produces one from $i$ to $j$.
    By induction, $(i, j) \not\in E(D_{X\setminus v})$, and furthermore either $i$ is not an in-neighbor of $v$ or $j$ is not an out-neighbor of $v$ in $D_{X\setminus v}$. The result follows by definition of vertex elimination.
\end{proof}

\movertwolowerbound*
\begin{proof}
    We will prove that $\OPT(D)\geq |E(D)|/2$ by induction on the number $k$ of internal nodes. For $k=0$, there are no internal nodes and so $E(D)=0$ by our stipulation that $E(S,T)=\emptyset$. Thus, $\OPT(D)\geq E(D)/2$ trivially. 
    
    Now, assume the claim holds for every graph with $k$ internal nodes. Let $D$ be a DAG with $k+1$ internal nodes, and let $\sigma=(v_1,v_2,\dots,v_{k+1})$ be an optimal vertex elimination sequence of $D$. We observe that $\sigma'=(v_2,\dots,v_{k+1})$ must be an optimal elimination sequence of $D'=\gElimVertex{D}{v_1}$. First, we re-express $\OPT(D)$ using the optimality of $\sigma'$.
    \begin{align*}
        \OPT(D)&= \cost_D(\sigma)=\mu_D(v_1)+\cost_{D'}(\sigma')\\
        &= |\innbrs{D}{v}|\cdot|\outnbrs{D}{v}|+\OPT(D')
    \end{align*}
    Since $D'$ has $k$ internal nodes, our inductive hypothesis gives $\OPT(D)\geq E(D')/2$. Now, observe that for  $a,b\in \mathbb{N}^+$, we have $ab\geq\frac{a+b}{2}$ by $a+b\leq 2a+2b-2\leq 2a+2b-2 + 2(a-1)(b-1)=2ab$. Thus,
    \begin{align*}        
        \OPT(D)&\geq \frac{|\innbrs{D}{v}|+|\outnbrs{D}{v}|}{2}+E(D')/2\\
        &\geq \frac{1}{2}\cdot \left(|\innbrs{D}{v}|+|\outnbrs{D}{v}| + |E(D)\cap E(D')|\right)\\
        &=E(D)/2
    \end{align*}
    The final equality holds by the fact that $E(D)$ can be partitioned into in-edges to $v$, out-edges from $v$, and edges preserved in $D'$. This concludes the proof.
 \end{proof}

To show integrality gaps for our linear programs, we provide explicit constructions. First, we present a useful lemma.
In a DAG $D$, two vertices $u$ and $v$ are \emph{false twins} if their in-neighborhoods are identical and their out-neighborhoods are identical.
We extend this term to sets by saying that a set $X$ of vertices are \emph{false twins} if every pair of vertices in $X$ are false twins.

\begin{lemma}[\cite{bentert2025structural}]\label{lemma:false-twins}
  Let $D = (S \uplus I \uplus T, E)$ be a DAG and let $X \subseteq I$ be a set of false twins. Then there exists an optimal elimination sequence (for \StrOJA{}) in which the vertices of $X$ are eliminated consecutively.
\end{lemma}

\flopsintegralitygap*
\begin{figure}
  \centering
  \resizebox{\columnwidth}{!}{%
  \begin{tikzpicture}
    \draw[rounded corners=5pt] (0,0) rectangle (1,4);
    \draw[rounded corners=5pt, shift={(3,.5)}] (0,0) rectangle (1, 3);
    \draw[rounded corners=5pt, shift={(6,.5)}] (0,0) rectangle (1, 3);
    \draw[rounded corners=5pt, shift={(9,1)}] (0,0) rectangle (1, 2);
    \draw[rounded corners=5pt, shift={(-3,.5)}] (0,0) rectangle (1, 3);
    \draw[rounded corners=5pt, shift={(-6,.5)}] (0,0) rectangle (1, 3);
    \draw[rounded corners=5pt, shift={(-9,1)}] (0,0) rectangle (1, 2);
    \node[shift={(.5, 2)}] (m) at (0,0) {\LARGE$M$};
    \node[shift={(3.5, 2)}] (u) at (0,0) {\LARGE$U$};
    \node[shift={(6.5, 2)}] (w) at (0,0) {\LARGE$W$};
    \node[shift={(9.5, 2)}] (t) at (0,0) {\LARGE$T$};
    \node[shift={(-2.5, 2)}] (u) at (0,0) {\LARGE$B$};
    \node[shift={(-5.5, 2)}] (w) at (0,0) {\LARGE$A$};
    \node[shift={(-8.5, 2)}] (t) at (0,0) {\LARGE$S$};

    \node[shift={(1, 3.75)}] (mtr) at (0, 0) {};
    \node[shift={(1, .25)}] (mbr) at (0, 0) {};
    \node[shift={(3, 3.25)}] (utl) at (0, 0) {};
    \node[shift={(3, .75)}] (ubl) at (0, 0) {};

    \draw[-Latex] (mtr)--(utl);
    \draw[-Latex] (mtr)--(ubl);
    \draw[-Latex] (mbr)--(utl);
    \draw[-Latex] (mbr)--(ubl);

    \node[shift={(4, 3.25)}] (utr) at (0, 0) {};
    \node[shift={(4, .75)}] (ubr) at (0, 0) {};
    \node[shift={(6, 3.25)}] (wtl) at (0, 0) {};
    \node[shift={(6, .75)}] (wbl) at (0, 0) {};

    \draw[-Latex] (utr)--(wtl);
    \draw[-Latex] (utr)--(wbl);
    \draw[-Latex] (ubr)--(wtl);
    \draw[-Latex] (ubr)--(wbl);

    \node[shift={(7, 3.25)}] (wtr) at (0, 0) {};
    \node[shift={(7, .75)}] (wbr) at (0, 0) {};
    \node[shift={(9, 2.75)}] (ttl) at (0, 0) {};
    \node[shift={(9, 1.25)}] (tbl) at (0, 0) {};

    \draw[-Latex] (wtr)--(ttl);
    \draw[-Latex] (wtr)--(tbl);
    \draw[-Latex] (wbr)--(ttl);
    \draw[-Latex] (wbr)--(tbl);
    \node[shift={(0, 3.75)}] (mtl) at (0, 0) {};
    \node[shift={(0, .25)}] (mbl) at (0, 0) {};
    \node[shift={(-2, 3.25)}] (btr) at (0, 0) {};
    \node[shift={(-2, .75)}] (bbr) at (0, 0) {};

    \draw[-Latex] (btr)--(mtl);
    \draw[-Latex] (btr)--(mbl);
    \draw[-Latex] (bbr)--(mtl);
    \draw[-Latex] (bbr)--(mbl);

    \node[shift={(-3, 3.25)}] (btl) at (0, 0) {};
    \node[shift={(-3, .75)}] (bbl) at (0, 0) {};
    \node[shift={(-5, 3.25)}] (atr) at (0, 0) {};
    \node[shift={(-5, .75)}] (abr) at (0, 0) {};

    \draw[-Latex] (atr)--(btl);
    \draw[-Latex] (atr)--(bbl);
    \draw[-Latex] (abr)--(btl);
    \draw[-Latex] (abr)--(bbl);

    \node[shift={(-6, 3.25)}] (atl) at (0, 0) {};
    \node[shift={(-6, .75)}] (abl) at (0, 0) {};
    \node[shift={(-8, 2.75)}] (str) at (0, 0) {};
    \node[shift={(-8, 1.25)}] (sbr) at (0, 0) {};

    \draw[-Latex] (str)--(atl);
    \draw[-Latex] (str)--(abl);
    \draw[-Latex] (sbr)--(atl);
    \draw[-Latex] (sbr)--(abl);
  \end{tikzpicture}%
  }
  \caption{\label{fig:flops-integrality-gap}The construction used in the proof of~\Cref{prop:flops-integrality-gap}. The edges between horizontally adjacent sets indicate complete bipartite connections. The sets $S$ and $T$ have cardinality $n^{1/4}$. The sets $A, B, U$, and $W$ have cardinality $n^{1/2}$. The set $M$, containing all remaining vertices, has cardinality $\Omega(n)$.}
\end{figure}
\begin{proof}
    Consider a DAG $D$ with $n$ vertices partitioned into~$7$ sets $S, A, B, M, U, W, T$ with cardinalities $|S| = |T| = n^{1/4}$, $|A| = |B| = |U| = |W| = n^{1/2}$, and $|M| = n - 2n^{1/4} - 4n^{1/2}=\Omega(n)$. We form complete bipartite graphs between each consecutive set, adding arcs from every vertex in $S$ to every vertex in $A$, from every vertex in $A$ to every vertex in $B$, and so on. Notice that $D$ has sources $S$, sinks $T$, and interior vertices $I=A\cup B\cup M\cup U\cup W$. See~\Cref{fig:flops-integrality-gap} for a visual depiction.

    We argue first that every total elimination sequence for $D$ has cost $\Omega(n^{7/4})$.
    By~\Cref{lemma:false-twins}, we may assume that the vertices of $M$ are eliminated consecutively.
    We now split into cases. For the first case, assume that the vertices of $M$ are eliminated before at least half of the vertices of $B$ and before at least half of the vertices of $U$.
    Then the elimination of the vertices in $M$ has cost at least
    \[
      |M|\cdot\left(\frac{1}{2}|B|\cdot\frac{1}{2}|U|\right) = |M| \cdot \frac{1}{4}\cdot n = \Omega(n^2).
    \]
    In the second case, there is a subset $B' \subseteq B$ of $B$, with $|B'| \geq \frac{1}{2}|B|$, such that every vertex in $B'$ is eliminated before any vertex in $M$. The case in which half of $U$ is eliminated before $M$ is symmetric.
    Notice that every vertex in $B'$ has in-degree at least $|S| = n^{1/4}$ at the time of its elimination. Then the elimination of the vertices of $B'$ has cost at least
    \[
      |B'|\cdot(n^{1/4}\cdot|M|) \geq \frac{1}{2}\cdot n^{3/4}\cdot |M| = \Omega(n^{7/4}).
    \]

    To complete the proof, we show that there is a fractional solution with cost $O(n^{3/2})$, and hence, that this instance has an integrality gap of $\Omega(n^{1/4})$.

    For each $i \in I$ and $j\in S\cup T$, set $x_{ij} = 1$ and $x_{ji} = 0$, and for each $i, j \in S \cup T$, set $x_{ij} = x_{ji} = 0$. 
    Then, for each $i, j \in I$, set $x_{ij} = x_{ji} = \frac{1}{2}$. Observe that constraints~(\ref{eq:ilp:order:partial-order}) and~(\ref{eq:ilp:order:total-order}) are satisfied.

    Next, we assign the $e$ variables. For each $(i, j) \in E$, set $e_{ij} = 1$, as required by constraint~(\ref{eq:ilp:order:initial-edges}).
    For each $s \in S$ and $b \in B$, set $e_{sb} = \frac{1}{2}$, satisfying constraint~(\ref{eq:ilp:order:existence}) because $x_{as}+x_{ab} = \frac{1}{2}$ for every $a \in A$. Similarly, set $e_{ut} = \frac{1}{2}$ for every $u \in U$ and $t \in T$.
    Next, for each $s \in S$ and $t \in T$, set $e_{st} = 1$. 
    Finally, set all remaining $e_{ij}$ values to $0$. This is safe with respect to constraint $(\ref{eq:ilp:order:existence})$ because of our fractional assignments of the $x_{ij}$ variables.
    Notice that constraint~(\ref{eq:ilp:order:reachability-ij}) is also satisfied.

    It remains to assign the $y$ variables. To satisfy constraint~(\ref{eq:ilp:order:payment}), we set $z_{sba} = 1$ for each $s \in S, b \in B$, and $a \in A$. Similarly, we set $z_{utw} = 1$ for each $u \in U, t \in T$, and $w \in W$. The sum of these $z$ variables is exactly $2\cdot n^{5/4}$.

    Constraint~(\ref{eq:ilp:order:payment}) is now satisfied for any values of the remaining $z$ variables, due to our fractional values on the $x$ and $e$ variables.
    However, we must also consider constraints~(\ref{eq:ilp:order:lb-separators}) and~(\ref{eq:ilp:order:lb-edges}).
    Notice that for every internal vertex $k$, $|S - k| = |S| = n^{1/4}$ and $|k - T| = |T| = n^{1/4}$. Thus, constraint~(\ref{eq:ilp:order:lb-separators}) is already satisfied for every $k \in A \cup D$. 
    Now, for every $k \in B \cup M \cup U$, set $z_{stk} = 1$ for every $s \in S$ and for every $t \in T$, and set $z_{ijk} = 0$ whenever $i \notin S$ or $j \notin T$.

    For every $k \in B \cup M \cup U$, we have $\sum_{i, j} z_{ijk} = n^{1/2}$. Thus, constraint~(\ref{eq:ilp:order:lb-separators}) is satisfied. Furthermore, every internal vertex $k$, $\sum_{i, j} z_{ijk} \geq \frac{1}{2}\deg(k)$, so constraint~(\ref{eq:ilp:order:lb-edges}) is also satisfied.

    Finally, the total sum of the $z$ variables is $2\cdot n^{5/4} + n^{1/2}\cdot|B \cup M \cup U|$, so our fractional solution has cost $O(n^{3/2})$, as desired.
\end{proof}

\scarcityintegralitygap*

\begin{figure}[t]
		\centering

\tikzmath{\xs = 0; \xa =2.25; \xb = 5.25; \xt = 7.5;
\y1 = 2; \y2 =1.3; \ydots = 1.05; \y3 = 0.7; \yn = 0;}
\begin{tikzpicture}[
    faint/.style={draw = gray, fill = gray, draw opacity=0.5, fill opacity = 0.5},
     v/.style={minimum size = 0.75cm}
    ]

\node[v] (s1) at (\xs, \y1) {$s_1$};
\node[v] (s2) at (\xs,\y2) {};
\node[v] (sdots) at (\xs,\ydots) {$\vdots$};
\node[v] (s3) at (\xs,\y3) {};
\node[v] (sn) at (\xs,\yn) {};
\node[v] (snlabel) at (\xs,\yn) {$s_{n/4}$};

\node[v] (a1) at (\xa, \y1) {$a_1$};
\node[v] (a2) at (\xa,\y2) {};
\node[v] (adots) at (\xa,\ydots) {$\vdots$};
\node[v] (a3) at (\xa,\y3) {};
\node[v] (an) at (\xa,\yn) {};
\node[v] (anlabel) at (\xa,\yn) {$a_{n/4}$};

\node[v] (b1) at (\xb, \y1) {$b_1$};
\node[v] (b2) at (\xb,\y2) {};
\node[v] (bdots) at (\xb,\ydots) {$\vdots$};
\node[v] (b3) at (\xb,\y3) {};
\node[v] (bn) at (\xb,\yn) {};
\node[v] (bnlabel) at (\xb,\yn) {$b_{n/4}$};

\node[v] (t1) at (\xt, \y1) {$t_1$};
\node[v] (t2) at (\xt,\y2) {};
\node[v] (tdots) at (\xt,\ydots) {$\vdots$};
\node[v] (t3) at (\xt,\y3) {};
\node[v] (tn) at (\xt,\yn) {};
\node[v] (tnlabel) at (\xt,\yn) {$t_{n/4}$};

\draw[-Latex] (s1)--(a1);
\draw[-Latex] (sn)--(an);
\draw[-Latex] (b1)--(t1);
\draw[-Latex] (bn)--(tn);

\draw[-Latex, faint] (s2)--(a2);
\draw[-Latex, faint] (s3)--(a3);
\draw[-Latex, faint] (b2)--(t2);
\draw[-Latex, faint] (b3)--(t3);

\draw[-Latex] (a1)--(b1);
\draw[-Latex] (a1)--(bdots);
\draw[-Latex] (a1)--(bn);

\draw[-Latex, faint] (an)--(b1);
\draw[-Latex, faint] (an)--(bdots);
\draw[-Latex, faint] (an)--(bn);

\draw[-Latex, faint] (adots)--(b1);
\draw[-Latex, faint] (adots)--(bdots);
\draw[-Latex, faint] (adots)--(bn);

\end{tikzpicture}
	\caption{\label{fig:scarcity-integrality-gap}The construction used in the proof of~\Cref{prop:scarcity-integrality-gap}. The sets $S,A,B$, and $T$ each have cardinality $n/4$. There are edges between $s_i$ and $a_i$ as well as $b_i$ and $t_i$ for each $1\leq i \leq n/4$. There is a complete bipartite graph between $A$ and $B$.}
\end{figure}
\begin{proof}
    Consider a DAG $D$ with $n$ vertices partitioned into sets $S, A, B$, and $T$, each of cardinality $n/4$. Impose an arbitrary order on the vertices of $S = \{s_1, s_2 \ldots, s_{n/4}\}$. Impose orders on $A, B,$ and $T$ similarly.
    We add the arcs $(s_i, a_i)$ and $(b_i, t_i)$ for each $1 \leq i \leq n/4$, and the arcs $(a_i, b_j)$ for every $i, j$.
    This completes the construction.

    Notice that each $s_i$ has only one out-neighbor $a_i$. Thus, the elimination of $a_i$ (regardless of previous eliminations) creates new arcs from $s_i$ to each out-neighbor of $a_i$ at the time of its elimination. Then this elimination reduces the total number of edges by exactly~$1$. Similarly, the elimination of any $b_i$ reduces the total number of edges by exactly~$1$. It follows that the minimum edge count is achieved by eliminating every internal vertex, yielding $|S|\cdot|T| = n^2/16$ edges.

    Now, we show that there is a fractional solution with objective value $O(n)$, and hence, that this instance has an integrality gap of $\Omega(n)$.

    For each vertex $k \in A \cup B$, we set $x_k = \frac{1}{2}$. Next, for each arc $(u, v)\in E$, we set $z_{uv} = 1$, satisfying constraint~(\ref{eq:ilp:scarcity:initial-z-variables}).
    Then, for each $a \in A$ and $t \in T$, we set $z_{at} = \frac{1}{2}$. This is safe with respect to constraint~(\ref{eq:ilp:scarcity:z-variable-correctness}) because $x_b = \frac{1}{2}$ for every $b \in B$.
    Similarly, for each $s \in S$ and $b \in B$, we set $z_{sb} = \frac{1}{2}$.
    Finally, for each $s \in S$ and $t \in T$, we set $z_{st} = 0$, which is safe with respect to (\ref{eq:ilp:scarcity:z-variable-correctness}) because for each $a \in A$, $x_a = \frac{1}{2}$ and $z_{at} = \frac{1}{2}$.

    It remains to assign the $y$ variables. For each $s \in S$ and $a \in A$, we set $y_{sa} = \frac{1}{2}$ if $(s, a) \in E$ and $y_{sa} = 0$ otherwise. Similarly, for each $b \in B$ and $t \in T$, we set $y_{bt} = \frac{1}{2}$ if $(b,t) \in E$ and $y_{bt} = 0$ otherwise. Both assignments are safe with respect to constraint~(\ref{eq:ilp:scarcity:final-edges}) because $x_k = \frac{1}{2}$ for every internal vertex $k$. 
    The sum of these $y$ variables is exactly $n/4$.

    Next, for each $s \in S$ and $b \in B$, we set $y_{sb} = 0$. This satisfies constraint~(\ref{eq:ilp:scarcity:final-edges}) because $z_{sb} = \frac{1}{2}$ and $x_k = \frac{1}{2}$ for every internal vertex $k$.
    By a symmetric argument, we may set $y_{at} = 0$ for every $a \in A$ and $t \in T$.
    Then, for each $s \in S$ and $t \in T$, we set $y_{st} = 0$. Again, this satisfies constraint~(\ref{eq:ilp:scarcity:final-edges}) because $z_{st} = 0$.
    Finally, for each $a \in A$ and $b \in B$, we set $y_{ab} = 0$. This is safe with respect to constraint~(\ref{eq:ilp:scarcity:final-edges}) because $x_a = x_b = \frac{1}{2}$.

    This completes the variable assignments, and the sum of the $y$ variables is exactly $n/4$. This completes the proof.
\end{proof}

\subsection{Deferred Proofs from~\Cref{sec:separator-analysis}}\label{appendix:proofs:separator-analysis}

\costoftopo*
Indeed, it is straightforward to express the exact cost of the forward mode of AD. For an internal vertex $v \in \internals$, let $\sources_v \subseteq \sources$ be defined as the set of sources from which $v$ is reachable. Then the cost of eliminating vertices in forward topological order is the sum accross $v \in I$ of $|S_v|\cdot|\outnbrs{D}{v}|$. The cost of reverse mode can be expressed similarly. The proof of~\Cref{lem:topo} relates this quantity to the number of edges with at least one internal endpoint. 
\begin{proof}
    Consider elimination order $\pi$ of $\fo(D)$. When an internal vertex $v\in I$ is eliminated, it is by definition the first internal vertex in $\gElimSeq{D}{\pi(v)}$. Thus, $\innbrs{\gElimSeq{D}{\pi(v)}}{v}\subseteq S$, i.e., all of $v$'s in-neighbors are sources. Hence:
    \begin{align*}
    \cost_D(\pi)&=\sum_{v\in \pi(v)}\mu_{\gElimSeq{D}{\pi(v)}}(v)\\
    &= \sum_{v\in \pi(v)}|\innbrs{\gElimSeq{D}{\pi(v)}}{v}| \cdot |\outnbrs{\gElimSeq{D}{\pi(v)}}{v}|\\
    & \leq  |S|\cdot \sum_{v\in \pi(v)} |\outnbrs{\gElimSeq{D}{\pi(v)}}{v}|
    \end{align*}
    As $\innbrs{\gElimSeq{D}{\pi(v)}}{v}\subseteq S$, eliminating $v$ will only create source-internal edges, and so the out-neighborhood of internal nodes are unaffected by the elimination of $v$. It follows that $\outnbrs{\gElimSeq{D}{\pi(v)}}{v}=\outnbrs{D}{v}$, giving:
    \begin{align*}
    \cost_D(\pi)&\leq |S|\cdot \sum_{v\in \pi(v)} |\outnbrs{D}{v}|\\
    &=|S|\cdot |E(I\cup T)|\leq |S|\cdot |E(D)|\ .
    \end{align*}

    Thus, eliminating $D$ in topological order incurs cost at most $|S|\cdot |E|$. By symmetry, a bound of $|T|\cdot|E|$ holds for reverse topological order.
\end{proof}

\begin{figure}[t]
		\centering
\begin{tikzpicture}
\node[] (s0) at (0,0) {$s_1$};
\node[] (v) at (4,0) {$v_k$};
\node[] (s1) at (6,2) {$t_1$};
\node[] (st) at (6,0) {$t_{\ell}$};

\node[] (u1) at (0.2,3) {$v_1$};
\node[] (u2) at (1.4,2.3) {$v_2$};
\node[] (uk) at (3.8,0.9) {$v_{k-1}$};

\draw[-Latex] (s0)--(v);
\draw[-Latex] (v)--(s1);
\draw[-Latex] (v)--(st); 

\draw[-Latex] (s0)--(u1);
\draw[-Latex] (u1)--(u2);
\draw[-Latex] (s0)--(u2);
\draw[-Latex] (s0)--(uk);
\draw[-Latex] (uk)--(v);
\draw[-Latex] (u1)--(s1);

\node[] () at (6,1) {$\vdots$};
\node[] () at (2.8,1.8) {$\cdots$};
\end{tikzpicture}
	\caption{\label{fig:2t-bound graph}The construction used to show that the bound of~\Cref{thm:forward-reverse-approximations} is tight.}
\end{figure}

\noindent{\textbf{Tightness of}}~\Cref{thm:forward-reverse-approximations}.  We pause here to remark that the bound of the result is tight.
Consider a DAG (see~\Cref{fig:2t-bound graph}) with one source $s_1$, $k$ internal vertices $v_1, v_2, \ldots, v_k$, and $\ell = |T|$ sinks $t_1, t_2, \ldots, t_\ell$. We add edges from $s_1$ to every internal vertex. Next, for each $i \in [k-1]$, we add the edge $(v_i, v_{i+1})$. Finally, we add the edges $(v_1, t_1)$ and $(v_k, t_i)$ for each $i \in [\ell]$. Note that the edge $(v_1, t_1)$ is included only to ensure that~\Cref{prop:subdivision-preprocessing} cannot be applied.
Eliminating the internal vertices in forward topological order $(v_1, v_2, \ldots, v_k)$ incurs cost $2 + (k - 2) + \ell$.
Eliminating the vertices in reverse topological order $(v_k, v_{k-1}, \ldots, v_1)$ incurs cost $2\cdot\ell\cdot(k - 1) + \ell =  2\ell(k - 1/2)$.
Then the approximation ratio achieved by reverse mode is $2\ell(k-1/2)/(k + \ell)$, which, for every fixed $\ell$, converges to $2\ell = 2\cdot|T|$ as $k\rightarrow \infty$.
A symmetric construction shows that the bound for forward mode is tight.

\lowerboundsources*
\begin{proof}
  Fix an optimal elimination sequence $\pi$ for $D$
  and observe that
  \[
    |\innbrs{\gElimSet{D}{I\setminus\{v\}}}{v}|\cdot |\outnbrs{\gElimSet{D}{I\setminus\{v\}}}{v}|
   \leq |E(D_I)|
   \]
  due to the fact that $\innbrs{\gElimSet{D}{I\setminus\{v\}}}{v} \times \outnbrs{\gElimSet{D}{I\setminus\{v\}}}{v} \subseteq  E(D_I)$. Now, we show that $|E(D_I)|\leq \cost_D(\pi)$ by a charging argument. We will charge each edge $(s,t)\in E(D_I)$ to the unique vertex of~$I$ whose elimination in $\pi$ first introduces edge $(s,t)$, which by assumption is not in $D$. Let $c(u)$ be the number of times vertex $u\in I$ was charged, and let $\pi(v)$ denote the set of vertices preceding $v$ according to $\pi$. Then by the observation that $u$ can only be charged once for each pair of in- and out-neighbors, we have $c(u)\leq \mu_{\gElimSeq{D}{\pi(v)}}(v)$. To complete the proof,
  \begin{align*}
  |E(D_I)|&=\sum_{u\in I}c(u)\leq \sum_{u\in I}\mu_{\gElimSeq{D}{\pi(v)}}(v)\\
  &=\cost_D(\pi)=\OPT(D)\ .
  \end{align*}
\end{proof}

\middleoutapproximation*
\begin{proof}
    Define our elimination order produced by \Cref{alg:middleout} to be $\sigma=(\sigma_L,\sigma_R,\sigma_1,\dots,\sigma_k,\sigma_C)$, where $\sigma_L=\pi^r[V_L]$, $\sigma_R=\pi[V_r]$, $\sigma_i=\pi[V_i]$ for each $i\in [k]$, and $\sigma_C =\pi[C\cap I]$. We begin by bounding the cost of the prefix $\sigma\setminus \sigma_C$. First observe that $\cost_{D_{\sigma_L}}(\sigma_R)=\cost_D(\sigma_R)$, as $C$ separates $V_L$ from $V_R$ and so the elimination of $V_L$ does not affect the neighborhood of any vertices in $V_R$ and hence the cost of eliminating $V_R$. The same conclusion holds for each component $V_i$, i.e., $\cost_{\gElimSeq{D}{(\sigma_L,\sigma_R,\sigma_1,\dots,\sigma_{i-1})}}(\sigma_i)=\cost_D(\sigma_i)$. Thus, using the notation $\sigma(v)$ to be the prefix of $\sigma$ up to and excluding $v$, we have
    \begin{align*}
        &\cost_D(\sigma\setminus \sigma_C)\\
        &=\cost_D(\sigma_L)+\cost_D(\sigma_R)+\sum_{i\in [k]}\cost_D(\sigma_i)\\
    \end{align*}
    Now, we give a family of edge-disjoint subgraphs of $D$. We define the left DAG $D^L=D[S\uplus V_L\uplus C]\setminus (E(C)\cup E(S,C))$, whose sources are $S$ and sinks are the vertices of $C$ with in-neighbors but no out-neighbors in $V_L$. We similarly define DAGs $D^R=D[C\uplus V_R\uplus T]\setminus (E(C)\cup E(C,T))$ as well as $D^i=D[V_i\uplus C]\setminus E(C)$ for each $i\in [k]$, such that the sources and sinks of $D^i$ are (mutually disjoint) subsets of $C$.  Since we eliminate each DAG $D^x$ in topological order, we can apply \Cref{lem:topo} to bound the elimination costs $\sigma_x$:
    \begin{align*}
        &\cost_D(\sigma\setminus \sigma_C)\\
        &\leq |C|\cdot |E(D^L)|+|C|\cdot |E(D^R)|+\sum_{i\in[k]}|C|\cdot |E(D^i)|\\
        &\leq |C|\cdot |E(D)|\leq 2|C|\cdot \OPT(D)\ .
    \end{align*}
    The second inequality follows from the fact that the the subgraphs have disjoint edges, and the third inequality uses \Cref{lem:LB:m/2} to upper bound $E(D)$. 
    Now, let $D'=\gElimSet{D}{I\setminus C}$ be the remaining graph after eliminating $\sigma\setminus \sigma_C$. Thus,
    \begin{equation}\label{eq:halfwaypoint}
      \begin{aligned}
      \cost_D&(\mo(D,C))\\
      &\leq 2|C|\cdot\OPT(D)+\cost_{D'}(\sigma_C)        
      \end{aligned}
    \end{equation}   
    Note,
    \begin{align*}
        \cost_{D'}(\sigma_C)=\sum_{v\in C\cap I}|\innbrs{\gElimSeq{D}{\sigma_{C}(v)}'}{v}|\cdot |\outnbrs{\gElimSeq{D}{\sigma_{C}(v)}'}{v}|\ .
    \end{align*}
    
    Define $S_v:=\innbrs{\gElimSet{D}{I\setminus \{v\}}}{v}$ to be the set of $v$'s in-neighbors after eliminating all vertices other than $v$, and similarly define $T_v:=\outnbrs{\gElimSet{D}{I\setminus \{v\}}}{v}$. We observe that $\innbrs{D_{\sigma_{C}(v)}'}{v}=S_v$, as all vertices topologically before $v$ have been eliminated. Thus,
    \begin{align}
        &\cost_{D'}(\sigma_C)\notag\\
        &=\sum_{v\in C\cap I}|S_v|\cdot \left(|\outnbrs{\gElimSeq{D}{\sigma_{C}(v)}'}{v}\cap C|+ |\outnbrs{\gElimSeq{D}{\sigma_{C}(v)}'}{v}\cap T|\right)\label{eq:pickup}\\
        &\leq\sum_{v\in C\cap I} |S_v|\cdot\left(|C\setminus \sigma_{C}(v)|+|T_v|\right)\notag\\
        &\leq\sum_{v\in C\cap I} \OPT(D)\cdot |C\setminus \sigma_{C}(v)|+\OPT(D)\notag\\
        &= \OPT(D)\cdot\sum_{v\in C} \left(|C\setminus \sigma_{C}(v)| +1\right) \notag\\
        &=\left(\frac{|C|(|C|+1)}{2}+|C|\right)\cdot \OPT(D)\notag\ .
    \end{align}
    The first inequality makes use of the fact that $\outnbrs{\gElimSeq{D'}{\sigma_{C}(v)}}{v}\cap C\subseteq C\setminus \sigma_{C}(v)$ (trivially) and $\outnbrs{\gElimSeq{D'}{\sigma_{C}(v)}}{v}\cap T\subseteq T_v$, since no edges between $v$ and $T$ can be removed by eliminating other vertices. The second inequality uses \Cref{lem:LB:sources} twice, with a liberal application by the fact that $|T_v|\geq 1$ giving $|S_v|\leq \OPT(D)$.

    Combining both costs, we have a total cost of at most:
    \begin{align*}
    \cost_D&(\mo(D,C))\\
    &\leq \left(2|C|+\frac{|C|(|C|+1)}{2}+|C|\right)\cdot \OPT(D)\\
    &= \frac{|C|^2+7|C|}{2}\cdot \OPT(D)\ .
    \end{align*}
\end{proof}
\Cref{fig:middleout-approx} shows that there are graphs where the quadratic approximation ratio is asymptotically tight. However, we are able to refine our analysis when our separator is well-behaved.
\begin{definition}
  Let $D=(V,E)$ be a DAG. A set of vertices $C\subseteq V$ is \emph{convex} if for every pair of vertices $u,w\in C$ such that vertex $v$ is on a $u$ to $w$ path in $D$, we also have $v\in C$.
\end{definition}
Assuming that our separator is convex, we are able to prove a stronger bound.
\begin{corollary}\label{cor:middleout}
    Let $C$ be a convex $S$-$T$ separator in the underlying undirected graph of $D$. Then $\mo(D,C)$ is an $O(|C|+|E(C)|)$-approx. for~\StrOJA{}.
\end{corollary}
\begin{proof}
  Recall the definitions $D'=\gElimSeq{D}{I\setminus C}, S_v=\innbrs{\gElimSet{D}{I\setminus \{v\}}}{v}$, and $T_v=\outnbrs{\gElimSet{D}{I\setminus \{v\}}}{v}$. By assuming that $C$ is convex, we make the key observation that $E_{D'}(C)=E_D(C)$, and hence $\outnbrs{\gElimSeq{D}{\sigma_{C}(v)}'}{v}\cap C=\outnbrs{D}{v}\cap C$. Now, we pick up from \Cref{eq:pickup}, once more bounding $|\outnbrs{\gElimSeq{D}{\sigma_{C}(v)}'}{v}\cap T|$ by $|T_v|$:
    \begin{align*}
        \text{cost}_{D'}&(\fo(D'))\\
        &\leq \sum_{v\in C\cap I}|S_v|\cdot \left(|\outnbrs{D}{v}\cap C|+ |T_v|\right)\\
        &\leq \sum_{v\in C\cap I}\OPT(D)\cdot |\outnbrs{D}{v}\cap C|+ \OPT(D)\\
        &=\OPT(D)\cdot \sum_{v\in C\cap I}\cdot \left(|\outnbrs{D}{v}\cap C|+ 1\right)\\
        &\leq \left(|E_D(C)|+|C|\right)\OPT(D)
    \end{align*}
    We again use \Cref{lem:LB:sources} twice for the second inequality, and the final inequality follows by observing that the sum of out-edges for all vertices in $C$ counts each edge exactly once. Combining this with \Cref{eq:halfwaypoint}, we get:
  \begin{align*}
    \text{cost}_D(\mo(D,C))\leq \left(3|C|+|E(C)|\right)\cdot \OPT(D)\ .
    \end{align*}
\end{proof}
Of course, $|E(C)|$ can also be possibly quadratic in $|C|$. However, \Cref{cor:middleout} gives an approximation linear in $|C|$ so long as $C$ is convex \emph{and} independent in $D$; this captures the case when $C=S$ (resp. $C=T$) and \mo{} reduces to \fo{} (resp. \re{}), matching the approximation guarantees from \Cref{thm:forward-reverse-approximations}, up to a constant factor.
\begin{observation}
     Let $C$ be a convex, independent $S$-$T$ separator in the underlying undirected graph of $D$. Then $\mo(D,C)$ is a $3|C|$-approx. for~\StrOJA{}.
\end{observation}
It is an interesting and open question whether the analysis of $\mo$ can be improved when $C$ is a minimum separator.

\begin{figure}[t]
		\centering

\tikzmath{\xs = 0; \xa =1.5; \xb = 3; \xc = 4.5; \xd = 6; \xt = 7.5;
\y1 = 2; \y2 =1.3; \ydots = 1.05; \y3 = 0.7; \yn = 0; \ylabel = 3;}
  \resizebox{\columnwidth}{!}{%
\begin{tikzpicture}[
    faint/.style={draw = gray, fill = gray, draw opacity=0.5, fill opacity = 0.5},
     v/.style={minimum size = 0.75cm}
    ]


\node[v] (s1) at (\xs, \y1) {$s_1$};
\node[v] (s2) at (\xs,\y2) {};
\node[v] (sdots) at (\xs,\ydots) {$\vdots$};
\node[v] (s3) at (\xs,\y3) {};
\node[v] (sn) at (\xs,\yn) {$s_n$};

\node[v] (adots) at (\xa,\ydots) {$a$};

\node[v] (b1) at (\xb, \y1) {$b_1$};
\node[v] (b2) at (\xb,\y2) {};
\node[v] (bdots) at (\xb,\ydots) {$\vdots$};
\node[v] (b3) at (\xb,\y3) {};
\node[v] (bn) at (\xb,\yn) {$b_n$};

\node[v] (cdots) at (\xc,\ydots) {$u$};

\node[v] (d1) at (\xd, \y1) {$w_1$};
\node[v] (d2) at (\xd,\y2) {};
\node[v] (ddots) at (\xd,\ydots) {$\vdots$};
\node[v] (d3) at (\xd,\y3) {};
\node[v] (dn) at (\xd,\yn) {$w_n$};

\node[v] (tdots) at (\xt,\ydots) {$t$};

\draw[-Latex] (s1)--(adots);
\draw[-Latex] (sn)--(adots);
\draw[-Latex] (adots)--(b1);
\draw[-Latex] (adots)--(bn);
\draw[-Latex] (b1)--(cdots);
\draw[-Latex] (bn)--(cdots);
\draw[-Latex] (cdots)--(d1);
\draw[-Latex] (cdots)--(dn);
\draw[-Latex] (d1)--(tdots);
\draw[-Latex] (dn)--(tdots);

\draw[-Latex, faint] (s2)--(adots);
\draw[-Latex, faint] (s3)--(adots);
\draw[-Latex, faint] (adots)--(b2);
\draw[-Latex, faint] (adots)--(b3);
\draw[-Latex, faint] (b2)--(cdots);
\draw[-Latex, faint] (b3)--(cdots);
\draw[-Latex, faint] (cdots)--(d2);
\draw[-Latex, faint] (cdots)--(d3);
\draw[-Latex, faint] (d2)--(tdots);
\draw[-Latex, faint] (d3)--(tdots);

\end{tikzpicture}
}
	\caption{\label{fig:middleout-approx} Example of a DAG $D$ where $\mo$ is an $\Omega(|C|^2)$-approximation. Let $C=\{b_1,\dots,b_n,w_1,\dots,w_n\}$. Then $\mo(D,C)$ gives elimination sequence $(a,u,b_1,\dots,b_n,w_1,\dots,w_n)$ and incurs cost $n^3+3n^2=\Theta(|C|^3)$, while the optimal elimination sequence $(b_1,\dots,b_n,w_1,\dots,w_n,u,a)$ incurs cost $3n+1=\Theta(|C|)$.}
\end{figure}

\subsection{Deferred Proofs from~\Cref{sec:approximation-hardness}}\label{appendix:proofs:approximation-hardness}

\minflopsapxhard*
\begin{proof}
  We reduce from cubic \VC{}. Let $G = (V, E)$ be a cubic graph. We construct a DAG $D = (V', E')$ as follows; note that this construction is identical to~\cite{bentert2025structural}. For each $v \in v$, we create five vertices $v_1, v_2, v_3, v_4, v_5 \in V'$. We add the arcs $(v_1, v_2)$, $(v_2, v_3)$, $(v_2, v_4)$, $(v_2, v_5)$, $(v_3, v_4)$, and $(v_3, v_5)$. Next, for each edge $uv \in E$, we add the arcs $(u_1, v_3), (u_2, v_3), (v_1, u_3), (v_2, u_3)$. This completes the construction. 
  Note that the maximum in-degree is $7$ (achieved by each $v_3$), the maximum out-degree is $6$ (achieved by each $v_2$), and the maximum total-degree is $9$ (achieved by each $v_3$).

  Let $n = |V|$ and $m = |E|$. Note that $|V'| = 5n$, so the reduction is computable in polynomial time. Let $\sigma$ be a total elimination ordering of cost $k$. We claim that it is possible, in polynomial time, to produce a vertex cover of $G$ with cardinality at most $k - 13n$. Let $C \subseteq V$ be the set of vertices such that for each $v \in C$, either $v_2$ precedes $v_3$ in $\sigma$, or $v_3$ precedes $u_2$ in $\sigma$ for some neighbor $u$ of $v$.

  First we claim that $C$ is a vertex cover in $G$. Assume toward a contradiction that there is some $uv \in E$ with $\{u, v\} \cap C = \emptyset$. Then $u_3$ is eliminated before $u_2$, $u_2$ is eliminated before $v_3$, $v_3$ is eliminated before $v_2$, and $v_2$ is eliminated before $u_3$. But this implies that $v_2$ is eliminated both before and after $u_3$, so we have a contradiction.

  Next we claim that $C$ has size at most $k - 13n$. First, we provide a lower bound on the cost of eliminating any vertex, regardless of order. Note that $v_3$ for any vertex $v \in V$ has exactly two out-neighbors $v_4$ and $v_5$, and this remains true for the duration of the elimination process, as $v_4, v_5$ are both sinks. Furthermore, $v_3$ has an in-neighbor $v_2$ (or $v_1$, if $v_2$ is eliminated before $v_3$) and for each neighbor $u$ of $v$, $v_3$ has an in-neighbor $u_1$. Thus, eliminating $v_3$ costs at least $2\deg(v) + 2$ for any elimination sequence. Now we analyze $v_2$. $v_1$ is always an in-neighbor of $v_2$, as it is in $D$ and it is a sink. Similarly, $v_4$ and $v_5$ are always out-neighbors of $v_2$. Finally, for each neighbor $u$ of $v$ $v_2$ has the out-neighbor $u_3$ (or $u_4$ and $u_5$ if $u_3$ is eliminated before $v_2$). Thus, eliminating $v_2$ costs at least $\deg(v) + 2$, regardless of order. Since these lower bounds hold for any vertex $v \in V$, we have that every total elimination sequence costs at least
  \begin{align*}
    &\sum_{v \in V} (\deg(v) + 2) + \sum_{v \in V} (2\deg(v) + 2) \\
    &= 4n + \sum_{v \in V} 3\deg(v) = 4n + 9n = 13n\ .
  \end{align*}
  Now we observe that for each $v \in C$, the cost of $\sigma$ increases by at least one over the analyzed lower bound. If $v_3$ is eliminated after $v_2$, then $v_3$ is an out-neighbor of $v_2$ at the time of $v_2$'s elimination, and we did not include this in our lower bound. If $v_3$ is eliminated before $u_2$ for some neighbor $u$ of $v$, then $u_2$ has out-neighbors $v_4$ and $v_5$, rather than $v_3$, at the time of its elimination.
  Thus, $|C| \leq k - 13n$, as otherwise the cost of $\sigma$ is larger than $k$, a contradiction.

  Let $\optG{D}$ denote the minimum cost of a total elimination sequence for $D$, and let $\tau(G)$ denote the vertex cover number of $G$. We have just proven that $\tau(G) \leq \optG{D} - 13n$. We now show that $\optG{D} \leq \tau(G) + 13n$. Let $C$ be a minimum vertex cover for $G$. Consider an elimination sequence in which we eliminate $v_2$ for all $v \in C$ first, followed by $u_3$ for all $u \in V\setminus C$, followed by $u_2$ for all $u \in V\setminus C$, and finally $v_3$ for each $v \in C$. 
  Observe that the cost of eliminating $v_2$ for $v \in C$ is $\deg(v) + 3$, as at the time of its elimination $v_2$ has a single in-neighbor $v_1$, three out-neighbors $v_3, v_4, v_5$, and $\deg(v)$ out-neighbors $u_3$ for each neighbor $u$ of $v$.
  The cost of eliminating $u_3$ for all $u \in V\setminus C$ is $2\deg(u) + 2$, as at the time of its elimination $u_3$ has two out-neighbors $u_4$ and $u_5$, one in-neighbor $u_2$, and $\deg(u)$ in-neighbors $v_1$ for each neighbor $v$ of $u$ (note that $N(u) \subseteq C$, so $v_2$ has already been eliminated).
  The cost of eliminating $u_2$ for all $u \in V\setminus C$ is $\deg(u) + 2$, as at the time of its elimination $u_2$ has one in-neighbor $u_1$, two out-neighbors $u_4, u_5$, and $\deg(u)$ out-neighbors $v_3$ for each neighbor $v$ of $u$.
  Finally, the cost of eliminating $v_3$ for all $v \in V$ is $2\deg(v) + 2$, as at the time of its elimination $v_3$ has two out-neighbors $v_4, v_5$, one in-neighbor $v_1$, and $\deg(v)$ in-neighbors $u_1$ for each neighbor $u$ of $v$.
  Summing across all vertices, we see that the cost of this elimination sequence is exactly
  \begin{align*}
    \sum_{v \in C} &(\deg(v) + 3) + \sum_{u \in V\setminus C} (2\deg(u) + 2)\\
    &\hphantom{hi}+  \sum_{u \in V\setminus C} (\deg(u) + 2)  + \sum_{v \in C} (2\deg(v) + 2) \\
    &= \sum_{v \in C} (3\deg(v) + 5) + \sum_{u \in V\setminus C} (3\deg(u) + 4) \\
    &= \sum_{v \in C} 14 + \sum_{u \in V\setminus C} 13 \\
    &= |C| + 13n = \tau(G) + 13n\ .
  \end{align*}

  We are now ready to complete the proof. Let $\varepsilon > 0$ be a constant such that a $(1 + \varepsilon)$-approximation for \VC{} in cubic graphs would imply $\Poly{} = \NP$.
  Such a constant exists, by~\cite{alimonti1997hardness}.
  Let $\varepsilon' = \varepsilon/27$, and suppose that there exists a $(1 + \varepsilon')$-approximation for \StrOJA{}. Then given a cubic graph $G$, we can construct a DAG $D$ as described above, and obtain a vertex cover $C$ of $G$ such that
  \begin{align*}
    |C| &\leq (1 + \varepsilon')\cdot\optG{D} - 13n \\
    &\leq (1 + \varepsilon')\cdot(\tau(G) + 13n) - 13n \\
    &= (1 + \varepsilon')\cdot\tau(G) +\varepsilon'\cdot13n\ .
  \end{align*}

  Since $G$ is cubic, $\tau(G) \geq m / 3 = n/2$. Combining this with the above, we obtain
  \begin{align*}    
    |C| &\leq (1+\varepsilon')\cdot\tau(G) + \varepsilon'\cdot26\cdot\tau(G) \\
    &= (1 + 27\varepsilon')\cdot\tau(G) = (1 + \varepsilon)\cdot\tau(G)\ .
  \end{align*}
  Recalling our choice of $\varepsilon$, the inequality above would imply that $\Poly{} = \NP$. This completes the proof.
\end{proof}

\maxisapxhardrestricted*
\begin{proof}[Proof of~\Cref{thm:maxIS-apx-hard-large-girth}]
  We reduce from \MaxIS{} in cubic graphs, which is \APX-hard~\cite{alimonti1997hardness}. Let $G = (V, E)$ be a cubic graph, and let $g \geq 3$. If $g$ is odd, replace it by $g+1$. We create a graph $G' = (V', E')$ as follows. For each vertex $v \in V(G)$, we create an \emph{original} vertex $v' \in V'$. For each edge $uv \in E$, we create $g$ \emph{subdivision} vertices, and edges such that these vertices induce a path, with one endpoint of the path adjacent to $u'$ and the other to $v'$. We refer to this path as the unique $(g+1)$-length $u'v'$-path in $G'$, and we say that its internal vertices are those subdivision vertices corresponding to the edge $uv$ in $G$. This completes the construction.
  Note that $G'$ has girth at least $g$ (in fact, at least $3g + 3$). Note also that $|V'| = |V| + g\cdot|E|$, and $|E'| = (g + 1)\cdot|E|$.
  Finally, note that in $G'$, every original vertex has degree $3$ and every subdivision vertex has degree $2$, and thus no pair of degree-$3$ vertices forms an edge.

  Let $\alpha(G), \alpha(G')$ denote the independence numbers of $G$ and $G'$, respectively. We now claim that $\alpha(G') \geq \alpha(G) + \frac{g}{2}\cdot|E|$. Let $X$ be an independent set in $G$. Then $X' = \{v' \mid v \in X\}$ is an independent set in $G'$. Furthermore, for each edge $uv \in E$, we may add to $X'$ $g/2$ subdivision vertices from the unique $(g+1)$-length $u'v'$-path in $G'$. We have now constructed an independent set $X'$ in $G'$ of size $|X| + \frac{g}{2}\cdot|E|$, which completes the proof of the claim.
  
  Now we claim that given any independent set $X'$ in $G'$, we may in polynomial time produce an independent set $X$ in $G$ of size at least $|X'| - \frac{g}{2}\cdot|E|$. First suppose that $X'$ contains two original vertices $u', v'$, with $uv \in E$. Then $X'$ must contain at most $\frac{g}{2} - 1$ subdivision vertices corresponding to the edge $uv$. We alter $X'$ by removing one of $u'$ or $v'$ (chosen arbitrarily) as well as any subdivision vertices corresponding to $uv$, and inserting exactly $g/2$ subdivision vertices corresponding to $uv$, chosen such they induce an independent set (it is possible that some removed subdivision vertices are reinserted). This operation does not reduce the size of $X'$. We repeat it exhaustively, and henceforth assume that $X'$ contains no pair of original vertices $u', v'$ such that $uv \in E$. Then $X = \{ v \mid v' \in X'\}$ is an independent set in $G$, and moreover it has size at least $|X'| - \frac{g}{2}\cdot|E|$, since $X'$ contains at most $g/2$ subdivision vertices corresponding to any edge in $E$.
  
  We are now ready to complete the proof. Let $\varepsilon > 0$ be a constant such that the existence of a polynomial-time $(1 - \varepsilon)$-approximation for \MaxIS{} in cubic graphs implies that \Poly = \NP; such a constant exists, by~\cite{alimonti1997hardness}. Let $\varepsilon' = \varepsilon/(3g + 1)$, and suppose that a $(1 - \varepsilon')$-approximation exists for \MaxIS{} restricted to graphs with maximum degree $3$, minimum degree $2$, and girth at least $g$. Then we apply this algorithm to generate an independent set $X'$ in $G'$. We may assume that $|X'| \geq 1 + \frac{g}{2}\cdot|E|$; otherwise, replace $X'$ with the set consisting of an arbitrary original vertex and $g/2$ subdivision vertices (chosen appropriately) corresponding to each edge in $E$.
  By our second claim, we produce an independent set $X$ in $G$ of size at least $|X'| - \frac{g}{2}\cdot|E|$.
  Then by our first claim, we have

  \begin{align*}
    |X| &\geq |X'| - \frac{g}{2}\cdot|E| \geq (1 - \varepsilon')\cdot\alpha(G') - \frac{g}{2}\cdot|E|\\
    & \geq (1 - \varepsilon')\cdot(\alpha(G) + \frac{g}{2}\cdot|E|) - \frac{g}{2}\cdot|E| \\
    &= (1 - \varepsilon')\cdot\alpha(G) - \varepsilon'\cdot\frac{g}{2}\cdot|E|.
  \end{align*}

  Observe that as $G$ is cubic, its independence number is at least $|V|/4$, and $|E| = \frac{3}{2}\cdot|V|$. Combining this with the above, we have
  \begin{align*}
    |X| &\geq (1 - \varepsilon')\cdot\alpha(G) - \varepsilon'\cdot\frac{g}{2}\cdot\frac{3}{2}\cdot|V| \\
    &\geq (1 - \varepsilon')\cdot\alpha(G) - 3\varepsilon'\cdot g\cdot \alpha(G) \\
    &= (1 - (3g + 1)\varepsilon')\cdot\alpha(G) \\
    &= (1 - \varepsilon)\cdot\alpha(G),
  \end{align*}
  which would imply \Poly{} = \NP, completing the proof. 
\end{proof}

\begin{table*}[h]
  \centering           
\begin{tabular}{|c|c|c||c|c|c|c|}    
  \hline
  Vertex Type & $\deg^-_G(v)$ & $\deg^+_G(v)$ & $|I_v|$ & $|O_v|$ & $|E(I_v, O_v)|$ & $\deg_D(v')$  \\  
  \hline
  1 & 0 & 3 & 4 & 6 & 24 & 13 \\
  \hline
  2a & 2 & 0 & 4 & 4 & 15 & 10 \\
  2b & 0 & 2 & 4 & 4 & 15 & 10 \\
  \hline
  3 & 1 & 1 & 4 & 4 & 16 & 10 \\
  \hline
\end{tabular}
\caption{\label{table:scarcity-reduction-params} Parameter settings for our construction of a DAG $D$, used to prove~\Cref{thm:scarcity-minimization-apx-hard,thm:scarcity-maximization-apx-hard}. The leftmost three columns are a reminder of the definitions of vertex types, where $\deg^-_G(v)$ and $\deg^+_G(v)$ are defined with respect to a vertex ordering described in the text. The next two columns define the cardinalities of $I_v$ and $O_v$. The second-from-right column defines the number of arcs created between $I_v$ and $O_v$, and the rightmost column records the degree of the replica vertex $v'$ associated with $v$.}
\end{table*}
\nicesolutionclaim*
\begin{proof}
  Suppose toward a contradiction that $X$ is not an independent set in $D$. Then let $u', v'$ be a pair of replica vertices in $X$ such that the corresponding vertices $u$ and $v$ are adjacent in $G$.
  Let $A = N_{D_{X \setminus \{v'\}}}(v') \setminus (I_v \cup O_v)$, and let $d = |A|$. Note that $d > \deg_G(v)$, since $\deg_G(v) \leq 3$ and $v'$ has four in-neighbors $I_u$ (if $u <_\pi v$) or four out-neighbors $O_u$ (if $v <_\pi u$) in $D_{X\setminus\{v'\}}$.
  Eliminating $v'$ in $D_{X\setminus \{v'\}}$ destroys $|A| + |I_v| + |O_v|$ arcs. However, this elimination creates at least $4d$ arcs, since every element of $A$ either gains out-neighbors $O_v$ or in-neighbors $I_v$. 
  If $v$ is type 1 we have $d \geq 4$, and thus the elimination destroys $d + 10$ arcs but creates $4d = d + 3d \geq d + 12$ arcs, contradicting that $X$ was a nice solution.
  So $v$ is not of type $1$. But then we have $d \geq 3$, and thus the elimination destroys $d + 8$ arcs but creates $4d = d + 3d \geq d + 9$ arcs, again contradicting that $X$ was a nice solution.   

  We have proven that any nice solution $X$ is an independent set in $D$. We now prove that $|X| = |E(D)| - |E(D_X)|$. We eliminate the vertices of $X$ in an arbitrary order.
  By independence, the degree of any vertex at the time of its elimination is equal to its degree in $D$. Now we analyze the number of arcs created by the elimination of $v' \in X$. Note that none of the replica neighbors of $v'$ are adjacent in $D$, as otherwise we have found a triangle in $G$. Furthermore, none of the replica neighbors of $v'$ are adjacent at the moment when $v'$ is eliminated, as otherwise we have found a $C_4$ in $G$.
  Then the elimination of $v'$ removes $\deg_D(v')$ arcs and creates 
  \begin{align*}
  \deg^-_G(v)&\cdot\deg^+_G(v) + |I_v|\cdot\deg^+(v) \\
  &+ |O_v|\cdot\deg^-(v) + |I_v|\cdot|O_v| - |E(I_v, O_v)|
  \end{align*}
  arcs. 
  Using the values in~\Cref{table:scarcity-reduction-params}, it is simple to check that for each vertex type, the latter quantity is equal to $\deg_D(v') - 1$. 
  Thus, the DAG has exactly one less arc after eliminating $v'$. Since this is true for every vertex in $X$, it follows that $|E(D)| - |E(D_X)| = |X_G|$, as desired.

  Finally, let $X$ be any independent subset of $D's$ internal vertices. By the argument above, $|E(D_X)| = |E(D_{X \setminus \{v'\}})| - 1$ for every $v' \in X$, so $X$ is a nice solution.
\end{proof}

\produceisclaim*
\begin{proof}
  If $X$ is a nice solution, then by Claim~\ref{claim:scarcity-reduction-nice-solutions} and our construction, the set $X_G = \{v \mid v' \in X\}$ is an independent set in $G$ of cardinality $k$. 
  If $X$ is not a nice solution, then identify a vertex $v' \in X$ such that $|E(D_{X\setminus \{v'\}})| < |E(D_{X})|$, and replace $X$ by $X' = X\setminus \{v'\}$. Repeat this operation until $X'$ is a nice solution, and note that $|E(D)| - |E(D_{X'})| > k$.
  Then apply the argument above to $X'$.
\end{proof}

\relateoptsclaim*
\begin{proof}
  From the definitions, $\optGmin{D} = |E(D)| - \optGmax{D}$, so it suffices to show that $\optGmax{D} = \alpha(G)$. By Claim~\ref{claim:scarcity-apx-produce-is}, $\alpha(G) \geq \optGmax{D}$. For the other direction, let $X_G$ be a maximum independent set in $G$, and let $X_D = \{v' \mid v \in X_G\}$, i.e., $X_D$ is the set of replica vertices corresponding to $X_G$.
  By our construction and Claim~\ref{claim:scarcity-reduction-nice-solutions}, $X_D$ is a nice solution, so $\optGmax{D} \geq |E(D)| - |E(D_{X_D})| = |X_D| = |X_G| = \alpha(G)$.
\end{proof}

\maxedgereductionapxhard*
\begin{proof}
  We reduce from \MaxIS{} in graphs with maximum degree 3, minimum degree 2, girth at least 5, and no two degree-$3$ vertices adjacent. Given an instance $G$, we create a DAG $D$ via the construction described above.
  Let $\varepsilon > 0$ be a constant such that a polynomial-time $(1 - \varepsilon)$-approximation for \MaxIS{} with the stated restrictions would imply \Poly = \NP; such a constant exists by~\Cref{thm:maxIS-apx-hard-large-girth}.
  Now suppose that there exists a $(1 - \varepsilon)$-approximation for \MaxER{}. Then we may obtain a vertex elimination set $X_D$ with $|E(D)| - |E(D_{X_D})| \geq (1 - \varepsilon)\cdot\optGmax{D}$.
  By Claim~\ref{claim:scarcity-apx-produce-is}, we obtain an independent set $X_G$ in $G$ of size at least $(1 - \varepsilon)\cdot\optGmax{D}$.
  But then by Claim~\ref{claim:scarcity-apx-relate-opts}, $|X| \geq (1 - \varepsilon)\cdot \alpha(G)$, which contradicts our choice of~$\varepsilon$.
\end{proof} 

\minedgecountapxhard*
\begin{proof}
  We reduce from \MaxIS{} in graphs with maximum degree $3$, minimum degree $2$, girth at least $5$, and no two degree-$3$ vertices adjacent. Given an instance $G$, we create a DAG $D$ via the construction described above. 
  Let $\varepsilon > 0$ be a constant such that a polynomial-time $(1 - \varepsilon)$-approximation for \MaxIS{} with the stated restrictions would imply \Poly = \NP; such a constant exists by~\Cref{thm:maxIS-apx-hard-large-girth}.
  Now let $\varepsilon' = \varepsilon/147$, and suppose that there exists a $(1 + \varepsilon')$-approximation for \MinEC{}.
  Then we obtain a vertex elimination set $X_D$ with $|E(D_{X_D})| \leq (1 + \varepsilon')\cdot\optGmin{D}$, where $\optGmin{D}$ denotes the optimum objective value for $\MinEC{}$ on $D$. Let $m'$ denote the number of edges in $D$, and observe that $\optGmin{D} = m' - \optGmax{D}$. Then by using Claim~\ref{claim:scarcity-apx-produce-is} and then Claim~\ref{claim:scarcity-apx-relate-opts}, we obtain an independent set $X_G$ in $G$ such that
  \begin{align*}
    |X_G|\geq m' - (1 + \varepsilon')&\cdot\optGmin{D} = \optGmax{D} - \varepsilon'\cdot\optGmin{D} \\
    &= \alpha(G) - \varepsilon'\cdot\optGmin{D} \\
    &= \alpha(G) - \varepsilon'\cdot(m' - \optGmax{D}) \\
    &= \alpha(G) - \varepsilon'\cdot(m' - \alpha(G)).
\end{align*}
Let $n$ denote the number of vertices in $G$. Observe that in the construction described above, $m' \leq 37n$. Observe also that as $G$ is cubic, $\alpha(G) \geq n/4$.
Then we have
\begin{align*}
  |X_G| &\geq \alpha(G) - \varepsilon'\cdot(37n - \alpha(G)) \\
  &\geq \alpha(G) - 147\varepsilon'\cdot\alpha(G) = (1 - \varepsilon)\cdot\alpha(G),
\end{align*}
which implies \Poly{} = \NP. This completes the proof.
\end{proof}

\maxedgereductionpolyapxbound*
\begin{proof}[Proof of~\Cref{thm:scarcity-poly-apx-hardness}]
  Let $G = (V, E)$ be a graph with girth $5$, and let $n = |V|$. Let $\pi$ be an arbitrary total order on $V$. For $v \in V$, we write $\deg^-_G(v)$ for the number of neighbors of $v$ which precede $v$ in $\pi$. We define $\deg^+_G(v)$ similarly.
  We create a DAG $D = (V', E')$ as follows.
  We add to $V'$ a set $T$ of $3n^2$ \emph{auxiliary} vertices.
  Then, for each $v \in V$, we add to $V'$ a \emph{replica} vertex $v'$, and two sets $I_v$ and $O_v$ of auxiliary vertices. Both $I_v$ and $O_v$ have cardinality $2n$.
  We add arcs from every vertex in $I_v$ to every vertex in~$\{v'\} \cup T$, and from $v'$ to every vertex in $T \cup O_v$.
  We also add $\gamma(v)$ arcs from $I_v$ to $O_v$ (it does not matter which arcs are chosen), where
  \[
    \gamma(v) = \deg^-_G(v)\cdot\deg^+_G(v) + (2n - 1)\cdot\deg_G(v) + n^2 - 4n + 1.
  \] 
  Note that this construction is always possible if $n \geq 4$, as we then have $0 < \gamma(v) < 4n^2 = |I_v|\cdot|O_v|$.
  Finally, for each edge $uv \in E$, we add the arc $(u', v')$ if $u$ precedes $v$ in $\pi$, or the arc $(v', u')$ if $v$ precedes $u$ in $\pi$.
  This completes the construction of $D$. Note that every auxiliary vertex is a source or a sink, and that every replica vertex is an internal vertex. Note also that $|V'| < 8n^2$, so the construction is computable in polynomial time.

  Let $X'$ be a subset of the replica vertices of $D$. As in~\Cref{claim:scarcity-reduction-nice-solutions}, we will say that $X'$ is a \emph{nice solution} if $D_{X'}$ has fewer arcs than $D_{X'\setminus\{v'\}}$ for every $v' \in X$.
  Note that given any solution $X'_1$, a nice solution $X'_2$ with $|E(D_{X'_2})| < |E(D_{X'_1})|$ may be computed in polynomial time by exhaustively removing a vertex $v'$ with $|E(D_{X'_1})| > |E(D_{X'_1\setminus\{v'\}})|$.

  Let $X'$ be a nice solution. We claim that $X = \{v \mid v' \in X'\}$ is an independent set in $G$.
  Otherwise, there are two vertices $u', v' \in X'$ such that $uv \in E$. Without loss of generality, assume $v$ precedes $u$ in $\pi$.
  Let $A = N_{D_{X' \setminus\{v'\}}}(v') \setminus (I_v \cup O_v \cup T)$, and let $d = |A|$. 
  Since every vertex of $O_u$ is an out-neighbor of $v'$ in $D_{X' \setminus\{v'\}}$, $d \geq 2n$.
  When $v'$ is eliminated in $D_{X' \setminus \{v'\}}$, exactly $d + |I_v| + |O_v| + |T| = d + 4n + 3n^2 \leq (\frac{3n}{2} + 3)d$ arcs are destroyed.
  Meanwhile, at least $2n\cdot d$ arcs are created, since every vertex in $A$ gains an arc to either every vertex in $I_v$ or to every vertex in $O_v$.
  As long as $n \geq 8$, the number of created arcs is larger than the number of destroyed arcs, contradicting that $X'$ is a nice solution. Thus, we may assume that $X'$ and $X$ are independent sets in $D$ and $G$, respectively.

  Now we prove that $|X| = |E(D)| - |E(D_{X'})|$. Consider the vertices of $X'$, eliminated in an arbitrary order. By independence, when $v' \in X'$ is eliminated, none of its neighbors have been eliminated. Then at the time of its elimination, $v'$ has degree $\deg_D(v') = \deg_G(v) + 3n^2 + 4n$. This is the number of arcs destroyed by the elimination of $v'$.
  Now we compute the number of arcs created by the elimination of $v'$.
  Note that no pair of replica neighbors of $v'$ is adjacent in $D$, as otherwise we have identified a triangle in $G$. Similarly, no pair of replica neighbors of $v'$ is adjacent at the moment $v'$ is eliminated, as otherwise we have identified a $C_4$ in $G$.
  Then the elimination of $v'$ creates precisely $\deg^-_G(v)\cdot\deg^+_G(v) + 2n\cdot\deg_G(v) + 4n^2 - \gamma(v)$ arcs.
  Substituting our choice for $\gamma(v)$, the latter quantity is $\deg_G(v) + 3n^2 + 4n - 1$.
  Comparing the numbers of destroyed and created arcs, we see that eliminating $v'$ reduces the total number of arcs by exactly $1$. Since this is true for every vertex in $X'$ and $|X| = |X'|$, it follows that $|X| = |E(D)| - |E(D_{X'})|$.

  Let $\alpha(G)$ denote the independence number of $G$ and let $\optGmax{D}$ denote the optimum objective value for \MaxER{} on $D$.
  Let $X$ be a maximum independent set in $G$, and let $X' = \{v' \mid v \in X\}$. By the argument in the previous paragraph, $X'$ is a nice solution and $|E(D)| - |E(D_{X'})| = |X'| = |X|$. It follows that $\optGmax{D} \geq \alpha(G)$.

  Now let $X'$ be subset of the internal vertices of $D$ such that $|E(D)| - |E(D_{X'})| \geq \frac{\optGmax{D}}{|V'|^c}$. By the arguments above, we may assume that $X'$ is a nice solution.
  Then $X = \{v \mid v' \in X'\}$ is an independent set in $G$, with
  \begin{align*}
    |X| = |X'| &= |E(D)| - |E(D_{X'})| \geq \frac{\optGmax{D}}{|V'|^c} \\
    &\geq \frac{\alpha(G)}{|V'|^c} \geq \frac{\alpha(G)}{(8n^2)^c} \geq \frac{\alpha(G)}{n^{2c + \varepsilon}}
  \end{align*}
  for any constant $\varepsilon > 0$, as long as $n$ is sufficiently large. Thus, a $n^{c}$-approximation for \MaxER{} implies a $n^{2c + \varepsilon}$-approximation for \MaxIS{} restricted to graphs of girth~$5$.
  The result follows.
\end{proof}


\begin{figure}[t]
  \centering
  \resizebox{0.5\columnwidth}{!}{%
  \begin{tikzpicture}
    \node[circle, draw] (v) at (0, 0) {$v'$};
    \node[circle, draw, minimum size=1cm, fill=blue!20, label=above:{$2n$}] (iv) at (-2.2, 2) {\large$I_v$};
    \node[circle, draw, minimum size=1cm, fill=blue!20, label=above:{$2n$}] (ov) at (2.2, 2) {\large$O_v$};
    \node[draw, rounded corners, fill = cyan!40!gray, minimum size=1.25cm, label=right:{$3n^2$}] (t) at (0, -2.6) {\large$T$};
    \node[] (m1) at (-2, 0.7) {};
    \node[] (m2) at (-2, 0) {};
    \node[] (m3) at (-2, -0.7) {};
    \node[] (n1) at (2, 0.7) {};
    \node[] (n2) at (2, 0) {};
    \node[] (n3) at (2, -0.7) {};
    \draw[->,line width=1.5pt] (iv)--node[above, yshift=2pt]{$\ \ 2n$}(v);
    \draw[->,line width=1.5pt] (iv) -- node[above]{$\gamma(v)$} (ov);
    \draw[->,line width=1.5pt] (iv) edge[bend right=60] node[left]{$6n^3$} (t);
    \draw[->,line width=1.5pt] (v)--node[above, xshift=-2pt]{$2n\ $}(ov);
    \draw[->,line width=1.5pt] (v) edge [] node[right, yshift=3pt]{$3n^2$}(t);
    \draw[->,dashed] (m1)--(v);
    \draw[->,dashed] (m2)--(v);
    \draw[->,dashed] (m3)--(v);
    \draw[->,dashed] (v)--(n1);
    \draw[->,dashed] (v)--(n2);
    \draw[->,dashed] (v)--(n3);
  \end{tikzpicture}%
  }
  \caption{\label{fig:scarcity-poly-apx-hard-gadget}Gadget for vertex $v$ used in the proof of \Cref{thm:scarcity-poly-apx-hardness}. The vertex $v'$ is a replica vertex, and the sets $I_v$ and $O_v$ contain auxiliary vertices; the cardinalities of these sets are indicated above. Bolded arrows represent multiple edges. Dashed arrows represent edges to or from the replica vertices in other gadgets. $T$ is a global set of cardinality $3n^2$, shared by all vertex gadgets.}
\end{figure}

\smallskip
\noindent\textbf{A new preprocessing rule.} We conclude by introducing a new safe preprocessing rule for~\StrOJA{}. To our knowledge, the only other known rule is~\Cref{prop:subdivision-preprocessing}.
Recall that we write $\mu_D^*(v) \vcentcolon = \cutsize{S}{v}\cdot\cutsize{v}{T}$ for the minimum possible Markowitz degree of $v$. By~\Cref{obs:sep-lower-bound},
the cost of eliminating $v$ in any elimination sequence is at least $\mu_D^*(v)$.
We write $\cost(\sigma)$ for the cost of an elimination sequence $\sigma$.

For two DAGs $D = (V=\sources\uplus\internals\uplus\sinks,E)$ and $D'= (V'=\sources'\uplus\internals'\uplus\sinks',E)$, we relax the subgraph relation by saying that $D$ is \emph{included in} $D'$ (denoted $D \subseteq D'$) if $V \subseteq V'$ and $(E \setminus (S \times T)) \subseteq E'$. 
Finally, for every internal node $v \in \internals$, we write 
\begin{align*}
  P^+_D(v) \vcentcolon = \{ w \in \outnbrs{D}{v} &\mid \exists u\in \innbrs{D}{v}:\\
  & (u, w) \notin E \cup (S \times T)\}\ .
\end{align*}

Intuitively, $P^+_D(v)$ is the set of out-neighbors of $v$ which cause $\gElimVertex{D}{v} \not\subseteq D$.
The following definition is symmetric.
\begin{align*}
  P^-_D(v) \vcentcolon = \{ u \in \innbrs{D}{v} &\mid \exists w\in \outnbrs{D}{v} :\\
  &(u, w) \notin E \cup (S \times T)\}\ .
\end{align*}

We can now state our result.
\begin{theorem}
  Let $D = (V = \sources\uplus\internals\uplus\sinks, E)$ be a DAG, and let $v \in I$ be such that $\mu_D(v) = \mu_D^*(v)$, $P^+_D(v) \subseteq T$, and $|P^+_D(v)| \leq 1$.
  Then there exists a minimum-cost total elimination sequence for $D$ in which $v$ appears first.
\end{theorem}

The rule can also be applied in the case where $\markdeg{v} = \mu_D^*(v)$, $P^-_D(v) \subseteq S$, and $|P^-_D(v)| \leq 1$.

\begin{proof}
  Let $\sigma^* = (v_1, v_2, \ldots, v_{|I|})$ be a minimum-cost total elimination sequence, and let $v_k = v$.
  Let $\sigma = (v, v_1, \ldots, v_{k-1}, v_{k+1}, \ldots, v_{|I|})$. We will show that $\cost(\sigma) \leq \cost(\sigma^*)$.
  
  We write $X_i = \{v_1, v_2, \ldots v_{i-1}\}$ for the set consisting of the first $i-1$ vertices eliminated by $\sigma^*$.
  Observe that for all $i > k$, we have $\mu_{\gElimSet{D}{X_i \cup \{v\}}}(v_i) = \mu_{\gElimSet{D}{X_i}}(v_i)$, since $X_i = X_i \cup \{v\}$.
  Next, we observe that by our hypotheses, for all $i < k$ we have
  \[
  \innbrs{D_{X_i \cup \{v\}}}{v_i} = \begin{cases}
	\innbrs{X_i}{v_i} \setminus \{v\} & \text{ if } v\in \innbrs{X_i}{v_i}\\
	\innbrs{X_i}{v_i} & \text{ o.w. }
\end{cases}
\]
\begin{align*}
&\hphantom{hello}\outnbrs{D_{X_i \cup \{v\}}}{v_i} = \\
&\begin{cases}
 \outnbrs{D_{X_i}}{v_i} \cup \{t\}\setminus \{v\} & \text{if } v\in \outnbrs{D_{X_i}}{v_i}, P^+_{D}(v)=\{t\}\\
    \outnbrs{D_{X_i}}{v_i} \setminus \{v\} & \text{if } v\in \outnbrs{D_{X_i}}{v_i}, P^+_{D}(v)=\emptyset\\
    N^+_{D_{X_i}}(v_i) & \text{o.w.}
\end{cases}
\end{align*}

Thus, $\forall i<k$, $\mu_{D_{X_i \cup \{v\}}}(v_i) \leq \mu_{D_{X_i}}(v_i)$. Lastly, since $\mu_D(v) = \mu^*_{D}(v) \leq \mu_{D_{X_k}}(v)$, we have:
\begin{align*}
\cost(\sigma) &= \mu_D(v) +  \underset{i\neq k}{\underset{i=1}{\overset{|I|}{\sum}}}\mu_{D_{X_i \cup \{v\}}}(v_i) \\
&\leq \mu_{D_{X_k}}(v) +  \underset{i\neq k}{\underset{i=1}{\overset{n}{\sum}}}\mu_{D_{X_i}}(v_i) \leq \cost(\sigma^*)\ ,
\end{align*}
as desired. This completes the proof.
\end{proof}
\section{Computational Results}\label{appendix:experiments}
\subsection{ILP Timing}

The timing data presented in~\Cref{tab:ilp-timing-flops-all-but-nls,tab:ilp-timing-flops-TEST-NLS,tab:ilp-timing-scarcity-all-but-nls,tab:ilp-timing-scarcity-TEST-NLS} was obtained using Gurobi version 13, single threaded, on an Intel(R) Xeon(R) Gold 6230 CPU @ 2.10GHz with a 10 minute timeout.
\begin{table*}[!htbp]
\caption{Running time and relative improvement (RI) percentages for our \StrOJA{} ILP on graphs from \CHMUset{} (top), \GWset{} (second from top), \EGset{} (second from bottom), and \AGset{} (bottom), run on a single core with a 10 minute timeout. RI percentages are reported with respect to the fastest version of the ILP from~\cite{chen2012integer,chen2012scarcity}. A RI percentage of $\infty$ indicates that our ILP found an optimal solution, but every version of the previous ILP timed out. A RI percentage of n/a indicates that every version of both ILPs timed out.}
\label{tab:ilp-timing-flops-all-but-nls}
\centering
{\small\setlength{\tabcolsep}{3pt}\begin{tabular}{lll|llll|llll}
\toprule
\multicolumn{3}{c}{} & \multicolumn{4}{c}{\ilpVE{} time (s)} & \multicolumn{4}{c}{\ilpVE{} Relative Improvement (\%)} \\
Graph & n & m & \va{} & \vb{} & \vc{} & \vd{} & \va{} & \vb{} & \vc{} & \vd{} \\
\midrule
fig10.4 & 10 & 12 & 0.03 & 0.04 & 0.04 & 0.28 & 38.30 & 12.77 & 8.51 & -506.38 \\
fig10.1 & 11 & 12 & 0.23 & 0.37 & 0.51 & 0.37 & 67.55 & 47.69 & 29.09 & 48.25 \\
ex10.8 & 12 & 14 & 0.12 & 0.07 & 0.08 & 0.05 & 86.33 & 92.00 & 90.27 & 94.32 \\
revbound & 12 & 21 & 0.90 & 0.27 & 0.19 & 0.03 & 76.80 & 93.01 & 95.09 & 99.23 \\
butterfly & 16 & 24 & 0.42 & 0.11 & 0.11 & 0.07 & 99.80 & 99.95 & 99.95 & 99.97 \\
\midrule
fig10.2 & 6 & 6 & 0.30 & 0.34 & 0.29 & 0.03 & -274.07 & -319.75 & -264.20 & 62.96 \\
fig10.1.orig & 11 & 12 & 0.22 & 0.10 & 0.11 & 0.09 & 72.62 & 86.69 & 85.55 & 88.21 \\
fig10.3 & 18 & 20 & 0.46 & 0.11 & 0.14 & 0.09 & ∞ & ∞ & ∞ & ∞ \\
fig10.6 & 34 & 70 & t & 275.71 & t & t & n/a & ∞ & n/a & n/a \\
\midrule
2d4x2x2 & 32 & 96 & 0.73 & 0.29 & 0.41 & 0.36 & ∞ & ∞ & ∞ & ∞ \\
2d3x3x2 & 36 & 135 & 1.13 & 0.43 & 0.58 & 0.56 & ∞ & ∞ & ∞ & ∞ \\
2d4x2x3 & 40 & 128 & 190.95 & 24.08 & 107.98 & 107.17 & ∞ & ∞ & ∞ & ∞ \\
2d3x3x3 & 45 & 180 & 431.56 & 105.85 & 159.94 & 157.69 & ∞ & ∞ & ∞ & ∞ \\
2d5x5x2 & 100 & 375 & 41.19 & 10.66 & 45.94 & 43.77 & ∞ & ∞ & ∞ & ∞ \\
2d5x5x3 & 125 & 500 & t & t & t & t & n/a & n/a & n/a & n/a \\
2d5x5x5 & 175 & 750 & t & t & t & t & n/a & n/a & n/a & n/a \\
2d10x10x2 & 400 & 1500 & t & t & t & t & n/a & n/a & n/a & n/a \\
2d10x10x5 & 700 & 3000 & t & t & t & t & n/a & n/a & n/a & n/a \\
2d10x10x10 & 1200 & 5500 & t & t & t & t & n/a & n/a & n/a & n/a \\
\midrule
Helmholtz & 7 & 8 & 0.08 & 0.05 & 0.03 & 0.03 & 27.27 & 54.55 & 73.64 & 71.82 \\
PropaneComb. & 27 & 29 & 5.07 & 0.67 & 0.22 & 0.18 & ∞ & ∞ & ∞ & ∞ \\
Robot-6DOF & 34 & 43 & t & t & 0.57 & 0.32 & n/a & n/a & ∞ & ∞ \\
BlackScholes & 35 & 44 & t & t & 1.49 & 0.45 & n/a & n/a & ∞ & ∞ \\
HumanHeart & 36 & 43 & t & t & 0.66 & 0.42 & n/a & n/a & ∞ & ∞ \\
Perceptron & 44 & 47 & t & t & 352.60 & 1.05 & n/a & n/a & ∞ & ∞ \\
f & 46 & 61 & t & t & t & t & n/a & n/a & n/a & n/a \\
Encoder & 98 & 114 & t & t & t & 101.49 & n/a & n/a & n/a & ∞ \\
RoeFlux-1d & 100 & 164 & t & t & t & t & n/a & n/a & n/a & n/a \\
g & 102 & 155 & t & t & t & t & n/a & n/a & n/a & n/a \\
RoeFlux-3d & 140 & 226 & t & t & t & t & n/a & n/a & n/a & n/a \\
\bottomrule
\end{tabular}}
\end{table*}

\begin{table*}[!htbp]
\caption{Running time and relative improvement (RI) percentages for our \StrOJA{} ILP on \NLSset{} graphs, run on a single core with a 10 minute timeout. RI percentages are reported with respect to the fastest version of the ILP from~\cite{chen2012integer,chen2012scarcity}. A RI percentage of $\infty$ indicates that our ILP found an optimal solution, but every version of the previous ILP timed out. A RI percentage of n/a indicates that every version of both ILPs timed out.}
\label{tab:ilp-timing-flops-TEST-NLS}
\centering
{\small\setlength{\tabcolsep}{3pt}\begin{tabular}{lll|llll|llll}
\toprule
\multicolumn{3}{c}{} & \multicolumn{4}{c}{\ilpVE{} time (s)} & \multicolumn{4}{c}{\ilpVE{} Relative Improvement (\%)} \\
Graph & n & m & \va{} & \vb{} & \vc{} & \vd{} & \va{} & \vb{} & \vc{} & \vd{} \\
\midrule
nls04 & 8 & 7 & 0.06 & 0.06 & 0.05 & 0.03 & 39.18 & 41.24 & 46.39 & 69.07 \\
nls07 & 12 & 16 & 0.12 & 0.08 & 0.28 & 0.04 & 83.19 & 88.30 & 58.33 & 94.15 \\
nls26 & 12 & 13 & 0.07 & 0.05 & 0.05 & 0.04 & 82.35 & 87.47 & 87.47 & 89.77 \\
nls05 & 15 & 13 & 1.36 & 0.16 & 0.07 & 0.05 & 34.91 & 92.25 & 96.60 & 97.61 \\
nls06 & 16 & 16 & 0.10 & 0.08 & 0.08 & 0.06 & 90.02 & 91.98 & 91.78 & 94.03 \\
nls21 & 20 & 29 & 1.02 & 0.25 & 0.16 & 0.10 & 98.26 & 99.57 & 99.73 & 99.83 \\
nls22 & 31 & 51 & 16.25 & 1.72 & 0.49 & 0.27 & ∞ & ∞ & ∞ & ∞ \\
nls03 & 48 & 47 & t & t & t & 2.17 & n/a & n/a & n/a & ∞ \\
nls01 & 50 & 59 & t & t & t & 460.83 & n/a & n/a & n/a & ∞ \\
nls13 & 52 & 60 & 38.50 & 3.15 & 2.21 & 0.91 & ∞ & ∞ & ∞ & ∞ \\
nls16 & 59 & 76 & t & t & t & 435.93 & n/a & n/a & n/a & ∞ \\
nls02 & 60 & 59 & t & t & t & 28.98 & n/a & n/a & n/a & ∞ \\
nls12 & 63 & 80 & 43.80 & 4.35 & 4.41 & 1.93 & ∞ & ∞ & ∞ & ∞ \\
nls08 & 63 & 90 & 25.04 & 3.60 & 3.63 & 2.30 & ∞ & ∞ & ∞ & ∞ \\
nls19 & 77 & 90 & t & 138.99 & 9.15 & 3.14 & n/a & ∞ & ∞ & ∞ \\
nls09 & 81 & 110 & 598.39 & 25.17 & 11.03 & 4.99 & ∞ & ∞ & ∞ & ∞ \\
nls14 & 84 & 140 & 37.78 & 9.19 & 11.00 & 8.12 & ∞ & ∞ & ∞ & ∞ \\
nls10 & 99 & 128 & t & t & 22.32 & 8.31 & n/a & n/a & ∞ & ∞ \\
nls23 & 115 & 189 & t & t & t & 212.65 & n/a & n/a & n/a & ∞ \\
nls20 & 123 & 150 & t & t & 44.43 & 14.57 & n/a & n/a & ∞ & ∞ \\
nls24 & 147 & 213 & t & t & 71.87 & 36.94 & n/a & n/a & ∞ & ∞ \\
nls15 & 172 & 340 & t & t & t & t & n/a & n/a & n/a & n/a \\
nls25 & 237 & 324 & t & t & t & 144.94 & n/a & n/a & n/a & ∞ \\
nls17 & 395 & 520 & t & t & t & t & n/a & n/a & n/a & n/a \\
nls11 & 593 & 1019 & t & t & t & t & n/a & n/a & n/a & n/a \\
nls18 & 1751 & 2383 & t & t & t & t & n/a & n/a & n/a & n/a \\
\bottomrule
\end{tabular}}
\end{table*}

\begin{table*}
\caption{Running time and relative improvement (RI) percentages for our \MinEC{} ILP on graphs from \CHMUset{} (top), \GWset{} (second from top), \EGset{} (second from bottom), and \AGset{} (bottom), run on a single core with a 10 minute timeout. RI percentages are reported with respect to the ILP from~\cite{chen2012integer}. A RI percentage of $\infty$ indicates that our ILP found an optimal solution, but the ILP of~\cite{chen2012integer} timed out. A RI percentage of n/a indicates that both ILPs timed out.}
\label{tab:ilp-timing-scarcity-all-but-nls}
\centering
{\small\setlength{\tabcolsep}{3pt}\begin{tabular}{lll|l|l}
\toprule
\multicolumn{3}{c}{} & \multicolumn{2}{c}{\ilpSC{}} \\
Graph & n & m & time (s) & RI (\%) \\
\midrule
fig10.4 & 10 & 12 & 0.03 & 57.14 \\
fig10.1 & 11 & 12 & 0.06 & 95.14 \\
ex10.8 & 12 & 14 & 0.05 & 96.14 \\
revbound & 12 & 21 & 0.04 & 99.36 \\
butterfly & 16 & 24 & 0.30 & 99.92 \\
\midrule
fig10.2 & 6 & 6 & 0.04 & 66.92 \\
fig10.1.orig & 11 & 12 & 0.03 & 97.27 \\
fig10.3 & 18 & 20 & 0.10 & ∞ \\
fig10.6 & 34 & 70 & 0.15 & ∞ \\
\midrule
2d4x2x2 & 32 & 96 & 0.09 & ∞ \\
2d3x3x2 & 36 & 135 & 0.11 & ∞ \\
2d4x2x3 & 40 & 128 & 0.21 & ∞ \\
2d3x3x3 & 45 & 180 & 0.32 & ∞ \\
2d5x5x2 & 100 & 375 & 0.75 & ∞ \\
2d5x5x3 & 125 & 500 & 11.87 & ∞ \\
2d5x5x5 & 175 & 750 & t & n/a \\
2d10x10x2 & 400 & 1500 & 3.87 & ∞ \\
2d10x10x5 & 700 & 3000 & t & n/a \\
2d10x10x10 & 1200 & 5500 & t & n/a \\
\midrule
Helmholtz & 7 & 8 & 0.03 & 82.84 \\
PropaneComb. & 27 & 29 & 0.12 & ∞ \\
Robot-6DOF & 34 & 43 & 0.05 & ∞ \\
BlackScholes & 35 & 44 & 0.07 & ∞ \\
HumanHeart & 36 & 43 & 0.06 & ∞ \\
Perceptron & 44 & 47 & 0.96 & ∞ \\
f & 46 & 61 & 0.12 & ∞ \\
Encoder & 98 & 114 & 33.15 & ∞ \\
RoeFlux-1d & 100 & 164 & 1.10 & ∞ \\
g & 102 & 155 & 7.27 & ∞ \\
RoeFlux-3d & 140 & 226 & 2.40 & ∞ \\
\bottomrule
\end{tabular}}
\end{table*}

\begin{table*}
\caption{Running time and relative improvement (RI) percentages for our \MinEC{} ILP on \NLSset{} graphs, run on a single core with a 10 minute timeout. RI percentages are reported with respect to the ILP from~\cite{chen2012integer,chen2012scarcity}. A RI percentage of $\infty$ indicates that our ILP found an optimal solution, but ILP of~\cite{chen2012integer} timed out. A RI percentage of n/a indicates that both ILPs timed out.}
\label{tab:ilp-timing-scarcity-TEST-NLS}
\centering
{\small\setlength{\tabcolsep}{3pt}\begin{tabular}{lll|l|l}
\toprule
\multicolumn{3}{c}{} & \multicolumn{2}{c}{\ilpSC{}} \\
Graph & n & m & time (s) & RI (\%) \\
\midrule
nls04 & 8 & 7 & 0.31 & -135.88 \\
nls07 & 12 & 16 & 0.05 & 97.45 \\
nls26 & 12 & 13 & 0.04 & 95.40 \\
nls05 & 15 & 13 & 0.06 & 99.99 \\
nls06 & 16 & 16 & 0.04 & 99.39 \\
nls21 & 20 & 29 & 0.05 & ∞ \\
nls22 & 31 & 51 & 0.07 & ∞ \\
nls03 & 48 & 47 & 0.37 & ∞ \\
nls01 & 50 & 59 & 0.46 & ∞ \\
nls13 & 52 & 60 & 0.07 & ∞ \\
nls16 & 59 & 76 & 0.54 & ∞ \\
nls02 & 60 & 59 & 0.46 & ∞ \\
nls12 & 63 & 80 & 0.10 & ∞ \\
nls08 & 63 & 90 & 0.10 & ∞ \\
nls19 & 77 & 90 & 0.09 & ∞ \\
nls09 & 81 & 110 & 0.40 & ∞ \\
nls14 & 84 & 140 & 0.12 & ∞ \\
nls10 & 99 & 128 & 0.42 & ∞ \\
nls23 & 115 & 189 & 0.26 & ∞ \\
nls20 & 123 & 150 & 0.25 & ∞ \\
nls24 & 147 & 213 & 0.24 & ∞ \\
nls15 & 172 & 340 & t & n/a \\
nls25 & 237 & 324 & 0.72 & ∞ \\
nls17 & 395 & 520 & 1.52 & ∞ \\
nls11 & 593 & 1019 & 5.55 & ∞ \\
nls18 & 1751 & 2383 & 46.59 & ∞ \\
\bottomrule
\end{tabular}}
\end{table*}

\subsection{Optimal Objective Values}

We report the optimal objective values found by ILPs. The values in~\Cref{tab:ilp-vals-flops-all-but-NLS,tab:ilp-vals-flops-TEST-NLS,tab:ilp-vals-scarcity-all-but-NLS,tab:ilp-vals-scarcity-TEST-NLS} were obtained using Gurobi version 13, single threaded, on an Intel(R) Xeon(R) Gold 6230 CPU @ 2.10GHz with a 10 minute timeout.
For graphs not solved optimally during the time limit, we attempted to obtain optimal solutions using longer timeouts and multiple threads.
This strategy resulted in optimal solutions to~\StrOJA{} for an additional four graphs.

\begin{itemize}
    \item The nls15 graph from dataset~\NLSset{} has an optimal objective value of $388$. This was identified using variant~\vd{} of our ILP with 4 threads. The computation required 1260.267 seconds.
    \item The nls17 graph from dataset~\NLSset{} has an optimal objective value of $487$. This was identified using variant~\vd{} of our ILP with 4 threads. The computation required 662.000 seconds.
    \item The 2d10x10x2 graph from dataset~\EGset{} has an optimal objective value of $9000$. This was identified using variant~\vb{} of our ILP with 12 threads. The computation required 1230.963 seconds.
    \item The f graph from dataset~\AGset{} has an optimal objective value of $71$. This was identified using variant~\vd{} of our ILP with 4 threads. The computation required 162.068 seconds.
\end{itemize}

This strategy also resulted in optimal solutions to~\MinEC{} for an additional two graphs.
\begin{itemize}
    \item The nls15 graph from dataset~\NLSset{} has an optimal objective value of $72$. This was identified using our ILP on 8 threads. The computation required 1537.630 seconds.
    \item The 2d10x10x5 graph from dataset~\EGset{} has an optimal objective value of $625$. This was identified using our ILP on 1 thread. The computation required 4905.374 seconds.
\end{itemize}

\begin{table*}
\caption{Minimum costs achieved by versions of our~\StrOJA{} ILP on graphs from datasets~\CHMUset{} (top),~\GWset{} (second from top),~\EGset{} (second from bottom), and~\AGset{} (bottom), run on a single core with a 10 minute timeout. A cost 't' indicates that the method timed out before finding an optimal solution.}
\label{tab:ilp-vals-flops-all-but-NLS}
\centering
{\small\setlength{\tabcolsep}{5pt}\begin{tabular}{lll|llll}
\toprule
\multicolumn{3}{c}{} & \multicolumn{4}{c}{\ilpVE{} Objective} \\
Graph & n & m & \va{} & \vb{} & \vc{} & \vd{} \\
\midrule
fig10.4 & 10 & 12 & 18 & 18 & 18 & 18 \\
fig10.1 & 11 & 12 & 15 & 15 & 15 & 15 \\
ex10.8 & 12 & 14 & 22 & 22 & 22 & 22 \\
revbound & 12 & 21 & 10 & 10 & 10 & 10 \\
butterfly & 16 & 24 & 48 & 48 & 48 & 48 \\
\midrule
fig10.2 & 6 & 6 & 4 & 4 & 4 & 4 \\
fig10.1.orig & 11 & 12 & 15 & 15 & 15 & 15 \\
fig10.3 & 18 & 20 & 26 & 26 & 26 & 26 \\
fig10.6 & 34 & 70 & t & 132 & t & t \\
\midrule
2d4x2x2 & 32 & 96 & 352 & 352 & 352 & 352 \\
2d3x3x2 & 36 & 135 & 630 & 630 & 630 & 630 \\
2d4x2x3 & 40 & 128 & 608 & 608 & 608 & 608 \\
2d3x3x3 & 45 & 180 & 1035 & 1035 & 1035 & 1035 \\
2d5x5x2 & 100 & 375 & 2250 & 2250 & 2250 & 2250 \\
2d5x5x3 & 125 & 500 & t & t & t & t \\
2d5x5x5 & 175 & 750 & t & t & t & t \\
2d10x10x2 & 400 & 1500 & t & t & t & t \\
2d10x10x5 & 700 & 3000 & t & t & t & t \\
2d10x10x10 & 1200 & 5500 & t & t & t & t \\
\midrule
Helmholtz & 7 & 8 & 5 & 5 & 5 & 5 \\
PropaneComb. & 27 & 29 & 25 & 25 & 25 & 25 \\
Robot-6DOF & 34 & 43 & t & t & 41 & 41 \\
BlackScholes & 35 & 44 & t & t & 42 & 42 \\
HumanHeart & 36 & 43 & t & t & 42 & 42 \\
Perceptron & 44 & 47 & t & t & 43 & 43 \\
f & 46 & 61 & t & t & t & t \\
Encoder & 98 & 114 & t & t & t & 104 \\
RoeFlux-1d & 100 & 164 & t & t & t & t \\
g & 102 & 155 & t & t & t & t \\
RoeFlux-3d & 140 & 226 & t & t & t & t \\
\bottomrule
\end{tabular}}
\end{table*}

\begin{table*}
\caption{Minimum costs achieved by versions of our~\StrOJA{{}} ILP on~\NLSset{{}} graphs, run on a single core with a 10 minute timeout. A cost 't' indicates that the method timed out before finding an optimal solution.}
\label{tab:ilp-vals-flops-TEST-NLS}
\centering
{\small\setlength{\tabcolsep}{5pt}\begin{tabular}{lll|llll}
\toprule
\multicolumn{3}{c}{} & \multicolumn{4}{c}{\ilpVE{} Objective} \\
Graph & n & m & \va{} & \vb{} & \vc{} & \vd{} \\
\midrule
nls04 & 8 & 7 & 5 & 5 & 5 & 5 \\
nls07 & 12 & 16 & 10 & 10 & 10 & 10 \\
nls26 & 12 & 13 & 10 & 10 & 10 & 10 \\
nls05 & 15 & 13 & 10 & 10 & 10 & 10 \\
nls06 & 16 & 16 & 12 & 12 & 12 & 12 \\
nls21 & 20 & 29 & 28 & 28 & 28 & 28 \\
nls22 & 31 & 51 & 53 & 53 & 53 & 53 \\
nls03 & 48 & 47 & t & t & t & 101 \\
nls01 & 50 & 59 & t & t & t & 147 \\
nls13 & 52 & 60 & 50 & 50 & 50 & 50 \\
nls16 & 59 & 76 & t & t & t & 154 \\
nls02 & 60 & 59 & t & t & t & 147 \\
nls12 & 63 & 80 & 70 & 70 & 70 & 70 \\
nls08 & 63 & 90 & 75 & 75 & 75 & 75 \\
nls19 & 77 & 90 & t & 75 & 75 & 75 \\
nls09 & 81 & 110 & 99 & 99 & 99 & 99 \\
nls14 & 84 & 140 & 120 & 120 & 120 & 120 \\
nls10 & 99 & 128 & t & t & 112 & 112 \\
nls23 & 115 & 189 & t & t & t & 226 \\
nls20 & 123 & 150 & t & t & 135 & 135 \\
nls24 & 147 & 213 & t & t & 189 & 189 \\
nls15 & 172 & 340 & t & t & t & t \\
nls25 & 237 & 324 & t & t & t & 299 \\
nls17 & 395 & 520 & t & t & t & t \\
nls11 & 593 & 1019 & t & t & t & t \\
nls18 & 1751 & 2383 & t & t & t & t \\
\bottomrule
\end{tabular}}
\end{table*}

\begin{table*}
\caption{Objective values achieved by our~\MinEC{} ILP, and that of~\cite{chen2012integer}, on graphs from datasets~\CHMUset{} (top),~\GWset{} (second from top),~\EGset{} (second from bottom), and~\AGset{} (bottom), run on a single core with a 10 minute timeout. A value 't' indicates that the method timed out before finding an optimal solution.}
\label{tab:ilp-vals-scarcity-all-but-NLS}
\centering
{\small\setlength{\tabcolsep}{5pt}\begin{tabular}{lll|ll}
\toprule
\multicolumn{3}{c}{} & \multicolumn{2}{c}{Minimum Representation} \\
Graph & n & m & \chenSC{} & \ilpSC{} \\
\midrule
fig10.4 & 10 & 12 & 11 & 11 \\
fig10.1 & 11 & 12 & 7 & 7 \\
ex10.8 & 12 & 14 & 11 & 11 \\
revbound & 12 & 21 & 1 & 1 \\
butterfly & 16 & 24 & 16 & 16 \\
\midrule
fig10.2 & 6 & 6 & 1 & 1 \\
fig10.1.orig & 11 & 12 & 7 & 7 \\
fig10.3 & 18 & 20 & t & 11 \\
fig10.6 & 34 & 70 & t & 20 \\
\midrule
2d4x2x2 & 32 & 96 & t & 64 \\
2d3x3x2 & 36 & 135 & t & 81 \\
2d4x2x3 & 40 & 128 & t & 64 \\
2d3x3x3 & 45 & 180 & t & 81 \\
2d5x5x2 & 100 & 375 & t & 375 \\
2d5x5x3 & 125 & 500 & t & 500 \\
2d5x5x5 & 175 & 750 & t & t \\
2d10x10x2 & 400 & 1500 & t & 1500 \\
2d10x10x5 & 700 & 3000 & t & t \\
2d10x10x10 & 1200 & 5500 & t & t \\
\midrule
Helmholtz & 7 & 8 & 1 & 1 \\
PropaneComb. & 27 & 29 & t & 15 \\
Robot-6DOF & 34 & 43 & t & 4 \\
BlackScholes & 35 & 44 & t & 5 \\
HumanHeart & 36 & 43 & t & 8 \\
Perceptron & 44 & 47 & t & 8 \\
f & 46 & 61 & t & 6 \\
Encoder & 98 & 114 & t & 16 \\
RoeFlux-1d & 100 & 164 & t & 12 \\
g & 102 & 155 & t & 96 \\
RoeFlux-3d & 140 & 226 & t & 12 \\
\bottomrule
\end{tabular}}
\end{table*}

\begin{table*}
\caption{Objective values achieved by our~\MinEC{{}} ILP, and that of~\cite{{chen2012integer}}, on \NLSset{} graphs, run on a single core with a 10 minute timeout. A value 't' indicates that the method timed out before finding an optimal solution.}
\label{tab:ilp-vals-scarcity-TEST-NLS}
\centering
{\small\setlength{\tabcolsep}{5pt}\begin{tabular}{lll|ll}
\toprule
\multicolumn{3}{c}{} & \multicolumn{2}{c}{Minimum Representation} \\
Graph & n & m & \chenSC{} & \ilpSC{} \\
\midrule
nls04 & 8 & 7 & 3 & 3 \\
nls07 & 12 & 16 & 4 & 4 \\
nls26 & 12 & 13 & 4 & 4 \\
nls05 & 15 & 13 & 4 & 4 \\
nls06 & 16 & 16 & 8 & 8 \\
nls21 & 20 & 29 & t & 6 \\
nls22 & 31 & 51 & t & 12 \\
nls03 & 48 & 47 & t & 16 \\
nls01 & 50 & 59 & t & 30 \\
nls13 & 52 & 60 & t & 20 \\
nls16 & 59 & 76 & t & 38 \\
nls02 & 60 & 59 & t & 20 \\
nls12 & 63 & 80 & t & 30 \\
nls08 & 63 & 90 & t & 45 \\
nls19 & 77 & 90 & t & 30 \\
nls09 & 81 & 110 & t & 44 \\
nls14 & 84 & 140 & t & 80 \\
nls10 & 99 & 128 & t & 48 \\
nls23 & 115 & 189 & t & 50 \\
nls20 & 123 & 150 & t & 45 \\
nls24 & 147 & 213 & t & 94 \\
nls15 & 172 & 340 & t & t \\
nls25 & 237 & 324 & t & 117 \\
nls17 & 395 & 520 & t & 163 \\
nls11 & 593 & 1019 & t & 495 \\
nls18 & 1751 & 2383 & t & 711 \\
\bottomrule
\end{tabular}}
\end{table*}


\subsection{Computational Results for \StrOJA{} Heuristics}

We ran the \ag{} software using default settings described in~\cite{lohoff2024optimizing} (50 MCTS simulations; 5000 episodes).
These tests were run on a NVIDIA RTX A6000.
We acknowledge that the software was designed for a more general setting than that considered in this paper. Namely, in this paper we assume that all values associated with edges in a linearized computational graph are scalars. 
We also present computational results for the other heuristics considered in this paper.
We implemented these algorithms in C++, and ran tests on a Intel(R) Xeon(R) Gold 6230 CPU @ 2.10GHz.

\begin{table*}[!htbp]
\caption{\StrOJA{} results from \ag{} on graphs from datasets~\CHMUset{} (top),~\GWset{} (second from top),~\EGset{} (second from bottom), and~\AGset{} (bottom). The 'obj' column records the lowest cost identified via~\fo{},~\re{},~\greedymark{}, and learning. The 'learned obj' column records the lowest cost obtained only through learning. We were unable to obtain solutions on some graphs. These graphs have '-1' entries. One running time is marked with an asterisk; this is to indicate that the software crashed, rather than terminating. In cases where the optimal solution is not known, we mark the ratio columns with '?'.}
\label{table:alphagrad-results-all-but-NLS}
\centering
{\small\setlength{\tabcolsep}{3pt}\begin{tabular}{lll|ll|ll|l}
\toprule
 & n & m & obj & ratio & learned obj & learned ratio & time \\
graph &  &  &  &  &  &  &  \\
\midrule
fig10.4 & 10 & 12 & 18 & 1.00 & 18 & 1.00 & 8m 43s \\
fig10.1 & 11 & 12 & 15 & 1.00 & 15 & 1.00 & 12m 46s \\
ex10.8 & 12 & 14 & 22 & 1.00 & 22 & 1.00 & 12m 54s \\
revbound & 12 & 21 & 10 & 1.00 & 11 & 1.10 & 24m 49s \\
butterfly & 16 & 24 & 48 & 1.00 & 48 & 1.00 & 21m 13s \\
\midrule
fig10.2 & 6 & 6 & 4 & 1.00 & 4 & 1.00 & 10m 9s \\
fig10.1.orig & 11 & 12 & 15 & 1.00 & 15 & 1.00 & 12m 52s \\
fig10.3 & 18 & 20 & 26 & 1.00 & 26 & 1.00 & 28m 5s \\
fig10.6 & 34 & 70 & 148 & 1.12 & 150 & 1.14 & 1h 36m 13s \\
\midrule
2d4x2x2 & 32 & 96 & 352 & 1.00 & 352 & 1.00 & 54m 35s \\
2d3x3x2 & 36 & 135 & 630 & 1.00 & 630 & 1.00 & 1h 7m 29s \\
2d4x2x3 & 40 & 128 & 608 & 1.00 & 648 & 1.07 & 1h 38m 4s \\
2d3x3x3 & 45 & 180 & 1035 & 1.00 & 1165 & 1.13 & 1h 58m 31s \\
2d5x5x2 & 100 & 375 & 2250 & 1.00 & 2485 & 1.10 & 43m 4s* \\
2d5x5x3 & 125 & 500 & 4875 & ? & 6820 & ? & 11h 59m 35s (timeout) \\
2d5x5x5 & 175 & 750 & 11125 & ? & 25923 & ? & 11h 58m 35s (timeout) \\
2d10x10x2 & 400 & 1500 & 9000 & 1.00 & 10370 & 1.15 & 11h 53m 35s (timeout) \\
2d10x10x5 & 700 & 3000 & -1 & ? & -1 & ? & -1 \\
2d10x10x10 & 1200 & 5500 & -1 & ? & -1 & ? & -1 \\
\midrule
Helmholtz & 7 & 8 & 5 & 1.00 & 5 & 1.00 & 11m 44s \\
PropaneComb. & 27 & 29 & 25 & 1.00 & 25 & 1.00 & 38m 26s \\
Robot-6DOF & 34 & 43 & 41 & 1.00 & 41 & 1.00 & 1h 51m 17s \\
BlackScholes & 35 & 44 & 42 & 1.00 & 45 & 1.07 & 1h 54m 46s \\
HumanHeart & 36 & 43 & 42 & 1.00 & 42 & 1.00 & 1h 37m 55s \\
Perceptron & 44 & 47 & 46 & 1.07 & 49 & 1.14 & 2h 33m 33s \\
f & 46 & 61 & 76 & 1.07 & 78 & 1.10 & 3h 34m 12s \\
Encoder & 98 & 114 & 113 & 1.09 & 136 & 1.31 & 11h 59m 21s (timeout) \\
RoeFlux-1d & 100 & 164 & 259 & ? & 327 & ? & 11h 59m 21s (timeout) \\
g & 102 & 155 & 376 & ? & 376 & ? & 11h 59m 21s (timeout) \\
RoeFlux-3d & 140 & 226 & 349 & ? & 482 & ? & 11h 58m 35s (timeout) \\
\bottomrule
\end{tabular}}
\end{table*}

\begin{table*}[!htbp]
\caption{\StrOJA{} results from \ag{} on graphs from dataset~\NLSset{}. The 'obj' column records the lowest cost identified via~\fo{},~\re{},~\greedymark{}, and learning. The 'learned obj' column records the lowest cost obtained only through learning. We were unable to obtain solutions on some graphs. These graphs have '-1' entries. In cases where the optimal solution is not known, we mark the ratio columns with '?'.}
\label{table:alphagrad-results-TEST-NLS}
\centering
{\small\setlength{\tabcolsep}{3pt}\begin{tabular}{lll|ll|ll|l}
\toprule
 & n & m & obj & ratio & learned obj & learned ratio & time \\
graph &  &  &  &  &  &  &  \\
\midrule
nls01 & 50 & 59 & 148 & 1.01 & 148 & 1.01 & 2h 23m 54s \\
nls02 & 60 & 59 & 148 & 1.01 & 148 & 1.01 & 3h 46m 23s \\
nls03 & 48 & 47 & 101 & 1.00 & 101 & 1.00 & 2h 33m 5s \\
nls04 & 8 & 7 & 5 & 1.00 & 5 & 1.00 & 10m 24s \\
nls05 & 15 & 13 & 10 & 1.00 & 10 & 1.00 & 23m 13s \\
nls06 & 16 & 16 & 12 & 1.00 & 12 & 1.00 & 20m 54s \\
nls07 & 12 & 16 & 10 & 1.00 & 10 & 1.00 & 19m 55s \\
nls08 & 63 & 90 & 75 & 1.00 & 76 & 1.01 & 4h 59m 29s \\
nls09 & 81 & 110 & 99 & 1.00 & 105 & 1.06 & 11h 1m 15s \\
nls10 & 99 & 128 & 112 & 1.00 & 120 & 1.07 & 11h 59m 35s (timeout) \\
nls11 & 593 & 1019 & -1 & ? & -1 & ? & -1 \\
nls12 & 63 & 80 & 70 & 1.00 & 70 & 1.00 & 5h 34m 24s \\
nls13 & 52 & 60 & 50 & 1.00 & 50 & 1.00 & 3h 50m 27s \\
nls14 & 84 & 140 & 120 & 1.00 & 125 & 1.04 & 9h 32m 29s \\
nls15 & 172 & 340 & 588 & 1.52 & 706 & 1.82 & 11h 57m 51s (timeout) \\
nls16 & 59 & 76 & 163 & 1.06 & 163 & 1.06 & 3h 56m 15s \\
nls17 & 395 & 520 & -1 & - & -1 & - & -1 \\
nls18 & 1751 & 2383 & -1 & ? & -1 & ? & -1 \\
nls19 & 77 & 90 & 75 & 1.00 & 75 & 1.00 & 9h 37m 1s \\
nls20 & 123 & 150 & 135 & 1.00 & 142 & 1.05 & 11h 59m 6s (timeout) \\
nls21 & 20 & 29 & 28 & 1.00 & 28 & 1.00 & 41m 18s \\
nls22 & 31 & 51 & 53 & 1.00 & 53 & 1.00 & 1h 33m 4s \\
nls23 & 115 & 189 & 261 & 1.15 & 261 & 1.15 & 11h 59m 20s (timeout) \\
nls24 & 147 & 213 & 189 & 1.00 & 214 & 1.13 & 11h 58m 5s (timeout) \\
nls25 & 237 & 324 & 299 & 1.00 & 354 & 1.18 & 11h 55m 6s (timeout) \\
nls26 & 12 & 13 & 10 & 1.00 & 10 & 1.00 & 17m 44s \\
\bottomrule
\end{tabular}}
\end{table*}

\begin{sidewaystable*}
\centering
\caption{\label{table:flops-approximations-all-but-NLS-sideways1}\StrOJA{} heuristic performances on graphs from \CHMUset{} (top), \GWset{} (second from top), \EGset{} (second from bottom), and \AGset{} (bottom). In cases where the optimal objective value is unknown, the ratio is marked '?'.}
\centering
{\footnotesize\setlength{\tabcolsep}{3pt}\begin{tabular}{lll|lll|lll|lll|lll|lll}
\toprule
\multicolumn{3}{c}{} & \multicolumn{3}{c}{\fo{}} & \multicolumn{3}{c}{\re{}} & \multicolumn{3}{c}{\greedymark{}} & \multicolumn{3}{c}{\relmark{}} & \multicolumn{3}{c}{\mtmr{}} \\
Graph & n & m & obj & ratio & time (s) & obj & ratio & time (s) & obj & ratio & time (s) & obj & ratio & time (s) & obj & ratio & time (s) \\
\midrule
fig10.4 & 10 & 12 & 22 & 1.22 & 0.000 & 18 & 1.00 & 0.000 & 22 & 1.22 & 0.000 & 18 & 1.00 & 0.000 & 18 & 1.00 & 0.000 \\
fig10.1 & 11 & 12 & 20 & 1.33 & 0.000 & 18 & 1.20 & 0.000 & 16 & 1.07 & 0.000 & 15 & 1.00 & 0.000 & 23 & 1.53 & 0.000 \\
ex10.8 & 12 & 14 & 28 & 1.27 & 0.000 & 24 & 1.09 & 0.000 & 26 & 1.18 & 0.000 & 22 & 1.00 & 0.000 & 28 & 1.27 & 0.000 \\
revbound & 12 & 21 & 10 & 1.00 & 0.000 & 19 & 1.90 & 0.000 & 10 & 1.00 & 0.000 & 10 & 1.00 & 0.000 & 10 & 1.00 & 0.000 \\
butterfly & 16 & 24 & 48 & 1.00 & 0.000 & 48 & 1.00 & 0.000 & 48 & 1.00 & 0.000 & 48 & 1.00 & 0.000 & 48 & 1.00 & 0.000 \\
\midrule
fig10.2 & 6 & 6 & 5 & 1.25 & 0.000 & 5 & 1.25 & 0.000 & 4 & 1.00 & 0.000 & 4 & 1.00 & 0.000 & 6 & 1.50 & 0.000 \\
fig10.1.orig & 11 & 12 & 20 & 1.33 & 0.000 & 18 & 1.20 & 0.000 & 16 & 1.07 & 0.000 & 15 & 1.00 & 0.000 & 23 & 1.53 & 0.000 \\
fig10.3 & 18 & 20 & 42 & 1.62 & 0.000 & 30 & 1.15 & 0.000 & 30 & 1.15 & 0.000 & 36 & 1.38 & 0.000 & 30 & 1.15 & 0.000 \\
fig10.6 & 34 & 70 & 259 & 1.96 & 0.000 & 202 & 1.53 & 0.000 & 148 & 1.12 & 0.000 & 135 & 1.02 & 0.000 & 188 & 1.42 & 0.000 \\
\midrule
2d4x2x2 & 32 & 96 & 352 & 1.00 & 0.000 & 352 & 1.00 & 0.000 & 352 & 1.00 & 0.000 & 352 & 1.00 & 0.000 & 416 & 1.18 & 0.000 \\
2d3x3x2 & 36 & 135 & 630 & 1.00 & 0.000 & 630 & 1.00 & 0.000 & 630 & 1.00 & 0.000 & 630 & 1.00 & 0.000 & 770 & 1.22 & 0.000 \\
2d4x2x3 & 40 & 128 & 608 & 1.00 & 0.000 & 608 & 1.00 & 0.000 & 648 & 1.07 & 0.000 & 608 & 1.00 & 0.000 & 676 & 1.11 & 0.000 \\
2d3x3x3 & 45 & 180 & 1035 & 1.00 & 0.000 & 1035 & 1.00 & 0.000 & 1179 & 1.14 & 0.000 & 1035 & 1.00 & 0.000 & 1274 & 1.23 & 0.000 \\
2d5x5x2 & 100 & 375 & 2250 & 1.00 & 0.000 & 2250 & 1.00 & 0.000 & 2250 & 1.00 & 0.000 & 2250 & 1.00 & 0.000 & 2865 & 1.27 & 0.002 \\
2d5x5x3 & 125 & 500 & 4875 & ? & 0.000 & 4875 & ? & 0.000 & 5475 & ? & 0.000 & 4875 & ? & 0.002 & 6130 & ? & 0.006 \\
2d5x5x5 & 175 & 750 & 11125 & ? & 0.001 & 11125 & ? & 0.001 & 14225 & ? & 0.001 & 11125 & ? & 0.008 & 16984 & ? & 0.018 \\
2d10x10x2 & 400 & 1500 & 9000 & 1.00 & 0.000 & 9000 & 1.00 & 0.000 & 9000 & 1.00 & 0.001 & 9000 & 1.00 & 0.009 & 11490 & 1.28 & 0.024 \\
2d10x10x5 & 700 & 3000 & 71500 & ? & 0.014 & 71500 & ? & 0.014 & 77700 & ? & 0.011 & 71500 & ? & 0.171 & 112601 & ? & 0.439 \\
2d10x10x10 & 1200 & 5500 & 299500 & ? & 0.072 & 299500 & ? & 0.082 & 512000 & ? & 0.051 & 299500 & ? & 2.610 & 937128 & ? & 4.408 \\
\midrule
Helmholtz & 7 & 8 & 5 & 1.00 & 0.000 & 6 & 1.20 & 0.000 & 5 & 1.00 & 0.000 & 5 & 1.00 & 0.000 & 5 & 1.00 & 0.000 \\
PropaneComb. & 27 & 29 & 63 & 2.52 & 0.000 & 25 & 1.00 & 0.000 & 38 & 1.52 & 0.000 & 25 & 1.00 & 0.000 & 25 & 1.00 & 0.000 \\
Robot-6DOF & 34 & 43 & 56 & 1.37 & 0.000 & 41 & 1.00 & 0.000 & 49 & 1.20 & 0.000 & 44 & 1.07 & 0.000 & 41 & 1.00 & 0.000 \\
BlackScholes & 35 & 44 & 108 & 2.57 & 0.000 & 42 & 1.00 & 0.000 & 50 & 1.19 & 0.000 & 48 & 1.14 & 0.000 & 45 & 1.07 & 0.000 \\
HumanHeart & 36 & 43 & 59 & 1.40 & 0.000 & 42 & 1.00 & 0.000 & 51 & 1.21 & 0.000 & 42 & 1.00 & 0.000 & 42 & 1.00 & 0.000 \\
Perceptron & 44 & 47 & 162 & 3.77 & 0.000 & 46 & 1.07 & 0.000 & 48 & 1.12 & 0.000 & 44 & 1.02 & 0.000 & 43 & 1.00 & 0.000 \\
f & 46 & 61 & 140 & 1.97 & 0.000 & 76 & 1.07 & 0.000 & 78 & 1.10 & 0.000 & 76 & 1.07 & 0.000 & 89 & 1.25 & 0.000 \\
Encoder & 98 & 114 & 789 & 7.59 & 0.000 & 113 & 1.09 & 0.000 & 134 & 1.29 & 0.000 & 108 & 1.04 & 0.001 & 104 & 1.00 & 0.000 \\
RoeFlux-1d & 100 & 164 & 558 & ? & 0.000 & 259 & ? & 0.000 & 341 & ? & 0.000 & 309 & ? & 0.001 & 284 & ? & 0.000 \\
g & 102 & 155 & 536 & ? & 0.000 & 447 & ? & 0.000 & 377 & ? & 0.000 & 401 & ? & 0.001 & 387 & ? & 0.000 \\
RoeFlux-3d & 140 & 226 & 818 & ? & 0.000 & 349 & ? & 0.000 & 435 & ? & 0.000 & 371 & ? & 0.004 & 372 & ? & 0.002 \\
\bottomrule
\end{tabular}}
\end{sidewaystable*}
\begin{sidewaystable*}
\centering
\caption{\label{table:flops-approximations-all-but-NLS-sideways2}\StrOJA{} heuristic performances on graphs from \CHMUset{} (top), \GWset{} (second from top), \EGset{} (second from bottom), and \AGset{} (bottom). In cases where the optimal objective value is unknown, the ratio is marked '?'.}
\centering
{\footnotesize\setlength{\tabcolsep}{3pt}\begin{tabular}{lll|lll|lll|lll|lll|lll}
\toprule
\multicolumn{3}{c}{} & \multicolumn{3}{c}{\lmmd{}} & \multicolumn{3}{c}{\pathlen{}} & \multicolumn{3}{c}{\moshort{}} & \multicolumn{3}{c}{\sa{}} & \multicolumn{3}{c}{\diffmincost{}} \\
Graph & n & m & obj & ratio & time (s) & obj & ratio & time (s) & obj & ratio & time (s) & obj & ratio & time (s) & obj & ratio & time (s) \\
\midrule
fig10.4 & 10 & 12 & 22 & 1.22 & 0.000 & 22 & 1.22 & 0.000 & 22 & 1.22 & 0.000 & 18 & 1.00 & 0.001 & 18 & 1.00 & 0.000 \\
fig10.1 & 11 & 12 & 16 & 1.07 & 0.000 & 22 & 1.47 & 0.000 & 15 & 1.00 & 0.000 & 15 & 1.00 & 0.002 & 19 & 1.27 & 0.000 \\
ex10.8 & 12 & 14 & 26 & 1.18 & 0.000 & 28 & 1.27 & 0.000 & 28 & 1.27 & 0.000 & 22 & 1.00 & 0.005 & 24 & 1.09 & 0.000 \\
revbound & 12 & 21 & 10 & 1.00 & 0.000 & 10 & 1.00 & 0.000 & 10 & 1.00 & 0.000 & 11 & 1.10 & 0.008 & 19 & 1.90 & 0.000 \\
butterfly & 16 & 24 & 48 & 1.00 & 0.000 & 48 & 1.00 & 0.000 & 48 & 1.00 & 0.000 & 48 & 1.00 & 0.000 & 48 & 1.00 & 0.000 \\
\midrule
fig10.2 & 6 & 6 & 4 & 1.00 & 0.000 & 6 & 1.50 & 0.000 & 5 & 1.25 & 0.000 & 4 & 1.00 & 0.001 & 5 & 1.25 & 0.000 \\
fig10.1.orig & 11 & 12 & 16 & 1.07 & 0.000 & 22 & 1.47 & 0.000 & 15 & 1.00 & 0.000 & 15 & 1.00 & 0.004 & 18 & 1.20 & 0.000 \\
fig10.3 & 18 & 20 & 30 & 1.15 & 0.000 & 46 & 1.77 & 0.000 & 30 & 1.15 & 0.000 & 26 & 1.00 & 0.006 & 30 & 1.15 & 0.000 \\
fig10.6 & 34 & 70 & 145 & 1.10 & 0.000 & 244 & 1.85 & 0.002 & 132 & 1.00 & 0.000 & 143 & 1.08 & 0.052 & 202 & 1.53 & 0.000 \\
\midrule
2d4x2x2 & 32 & 96 & 352 & 1.00 & 0.000 & 352 & 1.00 & 0.001 & 352 & 1.00 & 0.000 & 352 & 1.00 & 0.032 & 352 & 1.00 & 0.000 \\
2d3x3x2 & 36 & 135 & 630 & 1.00 & 0.000 & 630 & 1.00 & 0.001 & 630 & 1.00 & 0.000 & 630 & 1.00 & 0.044 & 630 & 1.00 & 0.000 \\
2d4x2x3 & 40 & 128 & 648 & 1.07 & 0.000 & 648 & 1.07 & 0.002 & 608 & 1.00 & 0.000 & 608 & 1.00 & 0.116 & 608 & 1.00 & 0.000 \\
2d3x3x3 & 45 & 180 & 1179 & 1.14 & 0.000 & 1179 & 1.14 & 0.006 & 1035 & 1.00 & 0.000 & 1035 & 1.00 & 0.197 & 1035 & 1.00 & 0.000 \\
2d5x5x2 & 100 & 375 & 2250 & 1.00 & 0.002 & 2250 & 1.00 & 0.039 & 2250 & 1.00 & 0.000 & 2250 & 1.00 & 1.519 & 2250 & 1.00 & 0.003 \\
2d5x5x3 & 125 & 500 & 5475 & ? & 0.006 & 5475 & ? & 0.116 & 4875 & ? & 0.001 & 4875 & ? & 4.521 & 4875 & ? & 0.005 \\
2d5x5x5 & 175 & 750 & 15525 & ? & 0.021 & 14225 & ? & 0.460 & 11125 & ? & 0.002 & 11125 & ? & 14.084 & 11125 & ? & 0.013 \\
2d10x10x2 & 400 & 1500 & 9000 & 1.00 & 0.040 & 9000 & 1.00 & 2.192 & 9000 & 1.00 & 0.006 & 9000 & 1.00 & 7.225 & 9000 & 1.00 & 0.038 \\
2d10x10x5 & 700 & 3000 & 80500 & ? & 0.627 & 77700 & ? & 25.757 & 71500 & ? & 0.028 & 71500 & ? & 144.387 & 71500 & ? & 0.200 \\
2d10x10x10 & 1200 & 5500 & 560000 & ? & 6.600 & 512000 & ? & 179.261 & 299500 & ? & 0.108 & 299500 & ? & 1535.214 & 299500 & ? & 0.825 \\
\midrule
Helmholtz & 7 & 8 & 5 & 1.00 & 0.000 & 6 & 1.20 & 0.000 & 5 & 1.00 & 0.000 & 5 & 1.00 & 0.000 & 6 & 1.20 & 0.000 \\
PropaneComb. & 27 & 29 & 25 & 1.00 & 0.000 & 25 & 1.00 & 0.000 & 25 & 1.00 & 0.000 & 25 & 1.00 & 0.007 & 25 & 1.00 & 0.000 \\
Robot-6DOF & 34 & 43 & 41 & 1.00 & 0.000 & 49 & 1.20 & 0.002 & 41 & 1.00 & 0.000 & 41 & 1.00 & 0.019 & 41 & 1.00 & 0.000 \\
BlackScholes & 35 & 44 & 42 & 1.00 & 0.000 & 53 & 1.26 & 0.002 & 42 & 1.00 & 0.000 & 43 & 1.02 & 0.016 & 42 & 1.00 & 0.000 \\
HumanHeart & 36 & 43 & 42 & 1.00 & 0.000 & 42 & 1.00 & 0.002 & 48 & 1.14 & 0.000 & 42 & 1.00 & 0.020 & 42 & 1.00 & 0.000 \\
Perceptron & 44 & 47 & 43 & 1.00 & 0.000 & 87 & 2.02 & 0.008 & 52 & 1.21 & 0.000 & 47 & 1.09 & 0.028 & 46 & 1.07 & 0.000 \\
f & 46 & 61 & 71 & 1.00 & 0.000 & 150 & 2.11 & 0.010 & 76 & 1.07 & 0.000 & 74 & 1.04 & 0.053 & 76 & 1.07 & 0.000 \\
Encoder & 98 & 114 & 104 & 1.00 & 0.000 & 420 & 4.04 & 0.081 & 127 & 1.22 & 0.000 & 119 & 1.14 & 0.150 & 113 & 1.09 & 0.002 \\
RoeFlux-1d & 100 & 164 & 251 & ? & 0.001 & 654 & ? & 0.110 & 267 & ? & 0.000 & 261 & ? & 0.285 & 259 & ? & 0.004 \\
g & 102 & 155 & 346 & ? & 0.000 & 475 & ? & 0.074 & 441 & ? & 0.000 & 362 & ? & 0.296 & 447 & ? & 0.003 \\
RoeFlux-3d & 140 & 226 & 370 & ? & 0.004 & 1728 & ? & 0.297 & 357 & ? & 0.000 & 396 & ? & 0.447 & 349 & ? & 0.008 \\
\bottomrule
\end{tabular}}
\end{sidewaystable*}
\begin{sidewaystable*}
\centering
\caption{\label{table:flops-approximations-all-but-NLS-sideways3}\StrOJA{} heuristic performances on graphs from \CHMUset{} (top), \GWset{} (second from top), \EGset{} (second from bottom), and \AGset{} (bottom). In cases where the optimal objective value is unknown, the ratio is marked '?'.}
\centering
{\footnotesize\setlength{\tabcolsep}{3pt}\begin{tabular}{lll|lll|lll|lll}
\toprule
\multicolumn{3}{c}{} & \multicolumn{3}{c}{\edgereduction{}} & \multicolumn{3}{c}{\markanddegree{}} & \multicolumn{3}{c}{\pc} \\
Graph & n & m & obj & ratio & time (s) & obj & ratio & time (s) & obj & ratio & time (s) \\
\midrule
fig10.4 & 10 & 12 & 22 & 1.22 & 0.000 & 22 & 1.22 & 0.000 & 22 & 1.22 & 0.000 \\
fig10.1 & 11 & 12 & 16 & 1.07 & 0.000 & 16 & 1.07 & 0.000 & 16 & 1.07 & 0.000 \\
ex10.8 & 12 & 14 & 26 & 1.18 & 0.000 & 26 & 1.18 & 0.000 & 26 & 1.18 & 0.000 \\
revbound & 12 & 21 & 10 & 1.00 & 0.000 & 10 & 1.00 & 0.000 & 18 & 1.80 & 0.000 \\
butterfly & 16 & 24 & 48 & 1.00 & 0.000 & 48 & 1.00 & 0.000 & 48 & 1.00 & 0.000 \\
\midrule
fig10.2 & 6 & 6 & 4 & 1.00 & 0.000 & 4 & 1.00 & 0.000 & 4 & 1.00 & 0.000 \\
fig10.1.orig & 11 & 12 & 16 & 1.07 & 0.000 & 16 & 1.07 & 0.000 & 16 & 1.07 & 0.000 \\
fig10.3 & 18 & 20 & 30 & 1.15 & 0.000 & 30 & 1.15 & 0.000 & 30 & 1.15 & 0.000 \\
fig10.6 & 34 & 70 & 132 & 1.00 & 0.000 & 148 & 1.12 & 0.000 & 132 & 1.00 & 0.000 \\
\midrule
2d4x2x2 & 32 & 96 & 352 & 1.00 & 0.000 & 352 & 1.00 & 0.000 & 352 & 1.00 & 0.000 \\
2d3x3x2 & 36 & 135 & 630 & 1.00 & 0.000 & 630 & 1.00 & 0.000 & 630 & 1.00 & 0.000 \\
2d4x2x3 & 40 & 128 & 648 & 1.07 & 0.000 & 648 & 1.07 & 0.000 & 648 & 1.07 & 0.000 \\
2d3x3x3 & 45 & 180 & 1179 & 1.14 & 0.000 & 1179 & 1.14 & 0.000 & 1179 & 1.14 & 0.000 \\
2d5x5x2 & 100 & 375 & 2250 & 1.00 & 0.001 & 2250 & 1.00 & 0.000 & 2250 & 1.00 & 0.001 \\
2d5x5x3 & 125 & 500 & 5475 & ? & 0.003 & 5475 & ? & 0.000 & 5475 & ? & 0.004 \\
2d5x5x5 & 175 & 750 & 14225 & ? & 0.011 & 14225 & ? & 0.002 & 14225 & ? & 0.012 \\
2d10x10x2 & 400 & 1500 & 9000 & 1.00 & 0.012 & 9000 & 1.00 & 0.003 & 9000 & 1.00 & 0.020 \\
2d10x10x5 & 700 & 3000 & 77700 & ? & 0.257 & 77700 & ? & 0.027 & 77700 & ? & 0.084 \\
2d10x10x10 & 1200 & 5500 & 512000 & ? & 2.553 & 512000 & ? & 0.122 & 512000 & ? & 0.281 \\
\midrule
Helmholtz & 7 & 8 & 5 & 1.00 & 0.000 & 5 & 1.00 & 0.000 & 5 & 1.00 & 0.000 \\
PropaneComb. & 27 & 29 & 38 & 1.52 & 0.000 & 38 & 1.52 & 0.000 & 25 & 1.00 & 0.000 \\
Robot-6DOF & 34 & 43 & 45 & 1.10 & 0.000 & 45 & 1.10 & 0.000 & 41 & 1.00 & 0.000 \\
BlackScholes & 35 & 44 & 51 & 1.21 & 0.000 & 50 & 1.19 & 0.000 & 44 & 1.05 & 0.000 \\
HumanHeart & 36 & 43 & 51 & 1.21 & 0.000 & 51 & 1.21 & 0.000 & 42 & 1.00 & 0.000 \\
Perceptron & 44 & 47 & 48 & 1.12 & 0.000 & 48 & 1.12 & 0.000 & 46 & 1.07 & 0.000 \\
f & 46 & 61 & 82 & 1.15 & 0.000 & 78 & 1.10 & 0.000 & 75 & 1.06 & 0.000 \\
Encoder & 98 & 114 & 135 & 1.30 & 0.000 & 131 & 1.26 & 0.000 & 117 & 1.12 & 0.001 \\
RoeFlux-1d & 100 & 164 & 311 & ? & 0.001 & 338 & ? & 0.000 & 267 & ? & 0.002 \\
g & 102 & 155 & 370 & ? & 0.000 & 371 & ? & 0.000 & 364 & ? & 0.001 \\
RoeFlux-3d & 140 & 226 & 392 & ? & 0.001 & 430 & ? & 0.001 & 343 & ? & 0.007 \\
\bottomrule
\end{tabular}}
\end{sidewaystable*}
\begin{sidewaystable*}
\centering
\caption{\label{table:flops-approximations-TEST-NLS-sideways1}\StrOJA{{}} heuristic performances on graphs from~\NLSset{}. In cases where the optimal objective value is unknown, the ratio is marked '?'.}
\centering
{\footnotesize\setlength{\tabcolsep}{3pt}\begin{tabular}{lll|lll|lll|lll|lll|lll}
\toprule
\multicolumn{3}{c}{} & \multicolumn{3}{c}{\fo{}} & \multicolumn{3}{c}{\re{}} & \multicolumn{3}{c}{\greedymark{}} & \multicolumn{3}{c}{\relmark{}} & \multicolumn{3}{c}{\mtmr{}} \\
Graph & n & m & obj & ratio & time (s) & obj & ratio & time (s) & obj & ratio & time (s) & obj & ratio & time (s) & obj & ratio & time (s) \\
\midrule
nls04 & 8 & 7 & 6 & 1.20 & 0.000 & 5 & 1.00 & 0.000 & 5 & 1.00 & 0.000 & 5 & 1.00 & 0.000 & 5 & 1.00 & 0.000 \\
nls07 & 12 & 16 & 10 & 1.00 & 0.000 & 14 & 1.40 & 0.000 & 10 & 1.00 & 0.000 & 10 & 1.00 & 0.000 & 12 & 1.20 & 0.000 \\
nls26 & 12 & 13 & 10 & 1.00 & 0.000 & 10 & 1.00 & 0.000 & 10 & 1.00 & 0.000 & 10 & 1.00 & 0.000 & 10 & 1.00 & 0.000 \\
nls05 & 15 & 13 & 11 & 1.10 & 0.000 & 10 & 1.00 & 0.000 & 10 & 1.00 & 0.000 & 10 & 1.00 & 0.000 & 10 & 1.00 & 0.000 \\
nls06 & 16 & 16 & 16 & 1.33 & 0.000 & 12 & 1.00 & 0.000 & 12 & 1.00 & 0.000 & 12 & 1.00 & 0.000 & 12 & 1.00 & 0.000 \\
nls21 & 20 & 29 & 31 & 1.11 & 0.000 & 42 & 1.50 & 0.000 & 31 & 1.11 & 0.000 & 28 & 1.00 & 0.000 & 31 & 1.11 & 0.000 \\
nls22 & 31 & 51 & 69 & 1.30 & 0.000 & 92 & 1.74 & 0.000 & 57 & 1.08 & 0.000 & 57 & 1.08 & 0.000 & 61 & 1.15 & 0.000 \\
nls03 & 48 & 47 & 228 & 2.26 & 0.000 & 200 & 1.98 & 0.000 & 106 & 1.05 & 0.000 & 101 & 1.00 & 0.000 & 137 & 1.36 & 0.000 \\
nls01 & 50 & 59 & 345 & 2.35 & 0.000 & 220 & 1.50 & 0.000 & 159 & 1.08 & 0.000 & 147 & 1.00 & 0.000 & 211 & 1.44 & 0.000 \\
nls13 & 52 & 60 & 50 & 1.00 & 0.000 & 50 & 1.00 & 0.000 & 50 & 1.00 & 0.000 & 50 & 1.00 & 0.000 & 50 & 1.00 & 0.000 \\
nls16 & 59 & 76 & 380 & 2.47 & 0.000 & 218 & 1.42 & 0.000 & 178 & 1.16 & 0.000 & 154 & 1.00 & 0.000 & 210 & 1.36 & 0.000 \\
nls02 & 60 & 59 & 355 & 2.41 & 0.000 & 310 & 2.11 & 0.000 & 159 & 1.08 & 0.000 & 147 & 1.00 & 0.000 & 211 & 1.44 & 0.000 \\
nls12 & 63 & 80 & 80 & 1.14 & 0.000 & 70 & 1.00 & 0.000 & 80 & 1.14 & 0.000 & 70 & 1.00 & 0.000 & 70 & 1.00 & 0.000 \\
nls08 & 63 & 90 & 105 & 1.40 & 0.000 & 75 & 1.00 & 0.000 & 90 & 1.20 & 0.000 & 75 & 1.00 & 0.000 & 75 & 1.00 & 0.000 \\
nls19 & 77 & 90 & 90 & 1.20 & 0.000 & 75 & 1.00 & 0.000 & 75 & 1.00 & 0.000 & 75 & 1.00 & 0.000 & 75 & 1.00 & 0.000 \\
nls09 & 81 & 110 & 154 & 1.56 & 0.000 & 99 & 1.00 & 0.000 & 121 & 1.22 & 0.000 & 99 & 1.00 & 0.000 & 99 & 1.00 & 0.000 \\
nls14 & 84 & 140 & 160 & 1.33 & 0.000 & 120 & 1.00 & 0.000 & 160 & 1.33 & 0.000 & 120 & 1.00 & 0.000 & 120 & 1.00 & 0.000 \\
nls10 & 99 & 128 & 176 & 1.57 & 0.000 & 112 & 1.00 & 0.000 & 128 & 1.14 & 0.000 & 112 & 1.00 & 0.000 & 112 & 1.00 & 0.000 \\
nls23 & 115 & 189 & 395 & 1.75 & 0.000 & 710 & 3.14 & 0.000 & 289 & 1.28 & 0.000 & 266 & 1.18 & 0.001 & 291 & 1.29 & 0.001 \\
nls20 & 123 & 150 & 180 & 1.33 & 0.000 & 135 & 1.00 & 0.000 & 150 & 1.11 & 0.000 & 135 & 1.00 & 0.000 & 135 & 1.00 & 0.001 \\
nls24 & 147 & 213 & 305 & 1.61 & 0.000 & 189 & 1.00 & 0.000 & 235 & 1.24 & 0.000 & 189 & 1.00 & 0.001 & 189 & 1.00 & 0.001 \\
nls15 & 172 & 340 & 640 & 1.65 & 0.000 & 1044 & 2.69 & 0.000 & 588 & 1.52 & 0.000 & 461 & 1.19 & 0.004 & 428 & 1.10 & 0.005 \\
nls25 & 237 & 324 & 529 & 1.77 & 0.000 & 299 & 1.00 & 0.000 & 345 & 1.15 & 0.000 & 299 & 1.00 & 0.003 & 299 & 1.00 & 0.007 \\
nls17 & 395 & 520 & 746 & 1.53 & 0.000 & 487 & 1.00 & 0.000 & 616 & 1.26 & 0.000 & 487 & 1.00 & 0.011 & 487 & 1.00 & 0.031 \\
nls11 & 593 & 1019 & 2554 & ? & 0.001 & 930 & ? & 0.000 & 1220 & ? & 0.001 & 930 & ? & 0.028 & 930 & ? & 0.106 \\
nls18 & 1751 & 2383 & 4689 & ? & 0.003 & 2315 & ? & 0.001 & 3272 & ? & 0.009 & 2315 & ? & 0.426 & 2315 & ? & 3.099 \\
\bottomrule
\end{tabular}}
\end{sidewaystable*}
\begin{sidewaystable*}
\centering
\caption{\label{table:flops-approximations-TEST-NLS-sideways2}\StrOJA{{}} heuristic performances on graphs from~\NLSset{}. In cases where the optimal objective value is unknown, the ratio is marked '?'.}
\centering
{\footnotesize\setlength{\tabcolsep}{3pt}\begin{tabular}{lll|lll|lll|lll|lll|lll}
\toprule
\multicolumn{3}{c}{} & \multicolumn{3}{c}{\lmmd{}} & \multicolumn{3}{c}{\pathlen{}} & \multicolumn{3}{c}{\moshort{}} & \multicolumn{3}{c}{\sa{}} & \multicolumn{3}{c}{\diffmincost{}} \\
Graph & n & m & obj & ratio & time (s) & obj & ratio & time (s) & obj & ratio & time (s) & obj & ratio & time (s) & obj & ratio & time (s) \\
\midrule
nls04 & 8 & 7 & 5 & 1.00 & 0.000 & 5 & 1.00 & 0.000 & 6 & 1.20 & 0.000 & 5 & 1.00 & 0.002 & 5 & 1.00 & 0.000 \\
nls07 & 12 & 16 & 10 & 1.00 & 0.000 & 12 & 1.20 & 0.000 & 10 & 1.00 & 0.000 & 10 & 1.00 & 0.003 & 14 & 1.40 & 0.000 \\
nls26 & 12 & 13 & 10 & 1.00 & 0.000 & 10 & 1.00 & 0.000 & 10 & 1.00 & 0.000 & 10 & 1.00 & 0.003 & 10 & 1.00 & 0.000 \\
nls05 & 15 & 13 & 10 & 1.00 & 0.000 & 10 & 1.00 & 0.000 & 11 & 1.10 & 0.000 & 10 & 1.00 & 0.003 & 10 & 1.00 & 0.000 \\
nls06 & 16 & 16 & 12 & 1.00 & 0.000 & 12 & 1.00 & 0.000 & 16 & 1.33 & 0.000 & 12 & 1.00 & 0.004 & 12 & 1.00 & 0.000 \\
nls21 & 20 & 29 & 28 & 1.00 & 0.000 & 31 & 1.11 & 0.000 & 31 & 1.11 & 0.000 & 28 & 1.00 & 0.006 & 42 & 1.50 & 0.000 \\
nls22 & 31 & 51 & 53 & 1.00 & 0.000 & 61 & 1.15 & 0.001 & 69 & 1.30 & 0.000 & 53 & 1.00 & 0.026 & 92 & 1.74 & 0.000 \\
nls03 & 48 & 47 & 110 & 1.09 & 0.000 & 249 & 2.47 & 0.008 & 101 & 1.00 & 0.000 & 101 & 1.00 & 0.078 & 200 & 1.98 & 0.000 \\
nls01 & 50 & 59 & 159 & 1.08 & 0.000 & 382 & 2.60 & 0.005 & 345 & 2.35 & 0.000 & 148 & 1.01 & 0.078 & 220 & 1.50 & 0.000 \\
nls13 & 52 & 60 & 50 & 1.00 & 0.000 & 50 & 1.00 & 0.008 & 50 & 1.00 & 0.000 & 50 & 1.00 & 0.002 & 50 & 1.00 & 0.000 \\
nls16 & 59 & 76 & 166 & 1.08 & 0.000 & 290 & 1.88 & 0.011 & 380 & 2.47 & 0.000 & 158 & 1.03 & 0.095 & 218 & 1.42 & 0.000 \\
nls02 & 60 & 59 & 159 & 1.08 & 0.000 & 391 & 2.66 & 0.017 & 147 & 1.00 & 0.000 & 148 & 1.01 & 0.105 & 310 & 2.11 & 0.000 \\
nls12 & 63 & 80 & 70 & 1.00 & 0.000 & 70 & 1.00 & 0.016 & 80 & 1.14 & 0.000 & 70 & 1.00 & 0.084 & 70 & 1.00 & 0.001 \\
nls08 & 63 & 90 & 75 & 1.00 & 0.000 & 75 & 1.00 & 0.013 & 105 & 1.40 & 0.000 & 75 & 1.00 & 0.076 & 75 & 1.00 & 0.001 \\
nls19 & 77 & 90 & 75 & 1.00 & 0.000 & 75 & 1.00 & 0.033 & 90 & 1.20 & 0.000 & 75 & 1.00 & 0.108 & 75 & 1.00 & 0.001 \\
nls09 & 81 & 110 & 99 & 1.00 & 0.000 & 99 & 1.00 & 0.039 & 154 & 1.56 & 0.000 & 99 & 1.00 & 0.134 & 99 & 1.00 & 0.001 \\
nls14 & 84 & 140 & 120 & 1.00 & 0.000 & 120 & 1.00 & 0.031 & 160 & 1.33 & 0.000 & 120 & 1.00 & 0.175 & 120 & 1.00 & 0.001 \\
nls10 & 99 & 128 & 112 & 1.00 & 0.001 & 112 & 1.00 & 0.070 & 176 & 1.57 & 0.000 & 115 & 1.03 & 0.162 & 112 & 1.00 & 0.002 \\
nls23 & 115 & 189 & 246 & 1.09 & 0.002 & 750 & 3.32 & 0.141 & 395 & 1.75 & 0.000 & 248 & 1.10 & 0.278 & 710 & 3.14 & 0.004 \\
nls20 & 123 & 150 & 135 & 1.00 & 0.001 & 135 & 1.00 & 0.151 & 180 & 1.33 & 0.000 & 139 & 1.03 & 0.179 & 135 & 1.00 & 0.004 \\
nls24 & 147 & 213 & 189 & 1.00 & 0.003 & 189 & 1.00 & 0.243 & 305 & 1.61 & 0.000 & 204 & 1.08 & 0.329 & 189 & 1.00 & 0.006 \\
nls15 & 172 & 340 & 388 & 1.00 & 0.010 & 428 & 1.10 & 0.421 & 640 & 1.65 & 0.000 & 567 & 1.46 & 0.505 & 1044 & 2.69 & 0.011 \\
nls25 & 237 & 324 & 299 & 1.00 & 0.013 & 299 & 1.00 & 1.154 & 529 & 1.77 & 0.001 & 343 & 1.15 & 0.504 & 299 & 1.00 & 0.016 \\
nls17 & 395 & 520 & 487 & 1.00 & 0.060 & 487 & 1.00 & 5.463 & 746 & 1.53 & 0.004 & 547 & 1.12 & 0.922 & 487 & 1.00 & 0.041 \\
nls11 & 593 & 1019 & 930 & ? & 0.202 & 930 & ? & 18.636 & 2554 & ? & 0.009 & 1212 & ? & 3.203 & 930 & ? & 0.119 \\
nls18 & 1751 & 2383 & 2315 & ? & 6.135 & 2315 & ? & 504.425 & 4689 & ? & 0.062 & 2781 & ? & 7.320 & 2315 & ? & 1.027 \\
\bottomrule
\end{tabular}}
\end{sidewaystable*}
\begin{sidewaystable*}
\centering
\caption{\label{table:flops-approximations-TEST-NLS-sideways3}\StrOJA{{}} heuristic performances on graphs from~\NLSset{}. In cases where the optimal objective value is unknown, the ratio is marked '?'.}
\centering
{\footnotesize\setlength{\tabcolsep}{3pt}\begin{tabular}{lll|lll|lll|lll}
\toprule
\multicolumn{3}{c}{} & \multicolumn{3}{c}{\edgereduction{}} & \multicolumn{3}{c}{\markanddegree{}} & \multicolumn{3}{c}{\pc} \\
Graph & n & m & obj & ratio & time (s) & obj & ratio & time (s) & obj & ratio & time (s) \\
\midrule
nls04 & 8 & 7 & 5 & 1.00 & 0.000 & 5 & 1.00 & 0.000 & 5 & 1.00 & 0.000 \\
nls07 & 12 & 16 & 10 & 1.00 & 0.000 & 10 & 1.00 & 0.000 & 12 & 1.20 & 0.000 \\
nls26 & 12 & 13 & 10 & 1.00 & 0.000 & 10 & 1.00 & 0.000 & 10 & 1.00 & 0.000 \\
nls05 & 15 & 13 & 10 & 1.00 & 0.000 & 10 & 1.00 & 0.000 & 10 & 1.00 & 0.000 \\
nls06 & 16 & 16 & 12 & 1.00 & 0.000 & 12 & 1.00 & 0.000 & 12 & 1.00 & 0.000 \\
nls21 & 20 & 29 & 28 & 1.00 & 0.000 & 31 & 1.11 & 0.000 & 28 & 1.00 & 0.000 \\
nls22 & 31 & 51 & 57 & 1.08 & 0.000 & 57 & 1.08 & 0.000 & 53 & 1.00 & 0.000 \\
nls03 & 48 & 47 & 106 & 1.05 & 0.000 & 106 & 1.05 & 0.000 & 101 & 1.00 & 0.000 \\
nls01 & 50 & 59 & 159 & 1.08 & 0.000 & 159 & 1.08 & 0.000 & 147 & 1.00 & 0.000 \\
nls13 & 52 & 60 & 50 & 1.00 & 0.000 & 50 & 1.00 & 0.000 & 50 & 1.00 & 0.001 \\
nls16 & 59 & 76 & 178 & 1.16 & 0.000 & 178 & 1.16 & 0.000 & 154 & 1.00 & 0.001 \\
nls02 & 60 & 59 & 159 & 1.08 & 0.000 & 159 & 1.08 & 0.000 & 147 & 1.00 & 0.001 \\
nls12 & 63 & 80 & 80 & 1.14 & 0.000 & 80 & 1.14 & 0.000 & 70 & 1.00 & 0.002 \\
nls08 & 63 & 90 & 90 & 1.20 & 0.000 & 90 & 1.20 & 0.000 & 75 & 1.00 & 0.000 \\
nls19 & 77 & 90 & 75 & 1.00 & 0.000 & 75 & 1.00 & 0.000 & 75 & 1.00 & 0.000 \\
nls09 & 81 & 110 & 121 & 1.22 & 0.000 & 121 & 1.22 & 0.000 & 99 & 1.00 & 0.000 \\
nls14 & 84 & 140 & 160 & 1.33 & 0.000 & 160 & 1.33 & 0.000 & 120 & 1.00 & 0.002 \\
nls10 & 99 & 128 & 128 & 1.14 & 0.000 & 128 & 1.14 & 0.001 & 112 & 1.00 & 0.001 \\
nls23 & 115 & 189 & 282 & 1.25 & 0.001 & 289 & 1.28 & 0.000 & 226 & 1.00 & 0.003 \\
nls20 & 123 & 150 & 150 & 1.11 & 0.001 & 150 & 1.11 & 0.002 & 135 & 1.00 & 0.002 \\
nls24 & 147 & 213 & 235 & 1.24 & 0.002 & 235 & 1.24 & 0.001 & 189 & 1.00 & 0.005 \\
nls15 & 172 & 340 & 468 & 1.21 & 0.001 & 428 & 1.10 & 0.000 & 523 & 1.35 & 0.006 \\
nls25 & 237 & 324 & 345 & 1.15 & 0.005 & 345 & 1.15 & 0.003 & 299 & 1.00 & 0.011 \\
nls17 & 395 & 520 & 616 & 1.26 & 0.018 & 616 & 1.26 & 0.011 & 487 & 1.00 & 0.028 \\
nls11 & 593 & 1019 & 1220 & ? & 0.035 & 1220 & ? & 0.026 & 930 & ? & 0.037 \\
nls18 & 1751 & 2383 & 3272 & ? & 0.353 & 3272 & ? & 0.276 & 2315 & ? & 0.415 \\
\bottomrule
\end{tabular}}
\end{sidewaystable*}

\subsection{Computational Results for \MinEC{} Heuristics.}
We implemented all~\MinEC{} heuristics tested in this paper in C++, and ran all tests on a Intel(R) Xeon(R) Gold 6230 CPU @ 2.10GHz.
\begin{sidewaystable*}
\centering
\caption{\label{table:scarcity-approximations-all-but-NLS-sideways1}\MinEC{} heuristic performances on graphs from \CHMUset{} (top), \GWset{} (second from top), \EGset{} (second from bottom), and \AGset{} (bottom). In cases where the optimal objective value is unknown, the ratio is marked '?'.}
\centering
{\footnotesize\setlength{\tabcolsep}{3pt}\begin{tabular}{lll|lll|lll|lll|lll|lll}
\toprule
\multicolumn{3}{c}{} & \multicolumn{3}{c}{\fo{}} & \multicolumn{3}{c}{\re{}} & \multicolumn{3}{c}{\greedymark{}} & \multicolumn{3}{c}{\relmark{}} & \multicolumn{3}{c}{\mtmr{}} \\
Graph & n & m & obj & ratio & time (s) & obj & ratio & time (s) & obj & ratio & time (s) & obj & ratio & time (s) & obj & ratio & time (s) \\
\midrule
fig10.4 & 10 & 12 & 12 & 1.09 & 0.000 & 11 & 1.00 & 0.000 & 12 & 1.09 & 0.000 & 11 & 1.00 & 0.000 & 18 & 1.64 & 0.000 \\
fig10.1 & 11 & 12 & 7 & 1.00 & 0.000 & 8 & 1.14 & 0.000 & 7 & 1.00 & 0.000 & 7 & 1.00 & 0.000 & 23 & 3.29 & 0.000 \\
ex10.8 & 12 & 14 & 12 & 1.09 & 0.000 & 11 & 1.00 & 0.000 & 12 & 1.09 & 0.000 & 11 & 1.00 & 0.000 & 28 & 2.55 & 0.000 \\
revbound & 12 & 21 & 1 & 1.00 & 0.000 & 1 & 1.00 & 0.000 & 1 & 1.00 & 0.000 & 1 & 1.00 & 0.000 & 10 & 10.00 & 0.000 \\
butterfly & 16 & 24 & 16 & 1.00 & 0.000 & 16 & 1.00 & 0.000 & 16 & 1.00 & 0.000 & 16 & 1.00 & 0.000 & 48 & 3.00 & 0.000 \\
\midrule
fig10.2 & 6 & 6 & 1 & 1.00 & 0.000 & 1 & 1.00 & 0.000 & 1 & 1.00 & 0.000 & 1 & 1.00 & 0.000 & 6 & 6.00 & 0.000 \\
fig10.1.orig & 11 & 12 & 7 & 1.00 & 0.000 & 8 & 1.14 & 0.000 & 7 & 1.00 & 0.000 & 7 & 1.00 & 0.000 & 23 & 3.29 & 0.000 \\
fig10.3 & 18 & 20 & 12 & 1.09 & 0.000 & 12 & 1.09 & 0.000 & 12 & 1.09 & 0.000 & 11 & 1.00 & 0.000 & 30 & 2.73 & 0.000 \\
fig10.6 & 34 & 70 & 20 & 1.00 & 0.000 & 20 & 1.00 & 0.000 & 20 & 1.00 & 0.000 & 20 & 1.00 & 0.000 & 188 & 9.40 & 0.000 \\
\midrule
2d4x2x2 & 32 & 96 & 64 & 1.00 & 0.000 & 64 & 1.00 & 0.000 & 64 & 1.00 & 0.000 & 64 & 1.00 & 0.000 & 416 & 6.50 & 0.000 \\
2d3x3x2 & 36 & 135 & 81 & 1.00 & 0.000 & 81 & 1.00 & 0.000 & 81 & 1.00 & 0.000 & 81 & 1.00 & 0.000 & 770 & 9.51 & 0.000 \\
2d4x2x3 & 40 & 128 & 64 & 1.00 & 0.000 & 64 & 1.00 & 0.000 & 64 & 1.00 & 0.000 & 64 & 1.00 & 0.000 & 676 & 10.56 & 0.000 \\
2d3x3x3 & 45 & 180 & 81 & 1.00 & 0.000 & 81 & 1.00 & 0.000 & 81 & 1.00 & 0.000 & 81 & 1.00 & 0.000 & 1274 & 15.73 & 0.000 \\
2d5x5x2 & 100 & 375 & 375 & 1.00 & 0.000 & 375 & 1.00 & 0.000 & 375 & 1.00 & 0.000 & 375 & 1.00 & 0.000 & 2865 & 7.64 & 0.002 \\
2d5x5x3 & 125 & 500 & 500 & 1.00 & 0.000 & 500 & 1.00 & 0.000 & 500 & 1.00 & 0.000 & 500 & 1.00 & 0.002 & 6130 & 12.26 & 0.006 \\
2d5x5x5 & 175 & 750 & 625 & 1.00 & 0.002 & 625 & 1.00 & 0.003 & 625 & 1.00 & 0.001 & 625 & 1.00 & 0.011 & 16984 & 27.17 & 0.021 \\
2d10x10x2 & 400 & 1500 & 1500 & 1.00 & 0.001 & 1500 & 1.00 & 0.001 & 1500 & 1.00 & 0.001 & 1500 & 1.00 & 0.009 & 11490 & 7.66 & 0.026 \\
2d10x10x5 & 700 & 3000 & 3000 & ? & 0.014 & 3000 & ? & 0.013 & 3000 & ? & 0.011 & 3000 & ? & 0.171 & 112601 & ? & 0.485 \\
2d10x10x10 & 1200 & 5500 & 5500 & ? & 0.068 & 5500 & ? & 0.077 & 5500 & ? & 0.053 & 5500 & ? & 2.638 & 937128 & ? & 4.806 \\
\midrule
Helmholtz & 7 & 8 & 1 & 1.00 & 0.000 & 1 & 1.00 & 0.000 & 1 & 1.00 & 0.000 & 1 & 1.00 & 0.000 & 5 & 5.00 & 0.000 \\
PropaneComb. & 27 & 29 & 15 & 1.00 & 0.000 & 15 & 1.00 & 0.000 & 15 & 1.00 & 0.000 & 15 & 1.00 & 0.000 & 25 & 1.67 & 0.000 \\
Robot-6DOF & 34 & 43 & 4 & 1.00 & 0.000 & 4 & 1.00 & 0.000 & 4 & 1.00 & 0.000 & 4 & 1.00 & 0.000 & 41 & 10.25 & 0.000 \\
BlackScholes & 35 & 44 & 5 & 1.00 & 0.000 & 5 & 1.00 & 0.000 & 5 & 1.00 & 0.000 & 5 & 1.00 & 0.000 & 45 & 9.00 & 0.000 \\
HumanHeart & 36 & 43 & 8 & 1.00 & 0.000 & 8 & 1.00 & 0.000 & 8 & 1.00 & 0.000 & 8 & 1.00 & 0.000 & 42 & 5.25 & 0.000 \\
Perceptron & 44 & 47 & 8 & 1.00 & 0.000 & 8 & 1.00 & 0.000 & 8 & 1.00 & 0.000 & 8 & 1.00 & 0.000 & 43 & 5.38 & 0.000 \\
f & 46 & 61 & 6 & 1.00 & 0.000 & 6 & 1.00 & 0.000 & 6 & 1.00 & 0.000 & 6 & 1.00 & 0.000 & 89 & 14.83 & 0.000 \\
Encoder & 98 & 114 & 16 & 1.00 & 0.000 & 16 & 1.00 & 0.000 & 16 & 1.00 & 0.000 & 16 & 1.00 & 0.001 & 104 & 6.50 & 0.000 \\
RoeFlux-1d & 100 & 164 & 12 & 1.00 & 0.000 & 12 & 1.00 & 0.000 & 12 & 1.00 & 0.000 & 12 & 1.00 & 0.001 & 284 & 23.67 & 0.000 \\
g & 102 & 155 & 129 & 1.34 & 0.000 & 129 & 1.34 & 0.000 & 107 & 1.11 & 0.000 & 119 & 1.24 & 0.001 & 387 & 4.03 & 0.000 \\
RoeFlux-3d & 140 & 226 & 12 & 1.00 & 0.000 & 12 & 1.00 & 0.000 & 12 & 1.00 & 0.000 & 12 & 1.00 & 0.004 & 372 & 31.00 & 0.001 \\
\bottomrule
\end{tabular}}
\end{sidewaystable*}
\begin{sidewaystable*}
\centering
\caption{\label{table:scarcity-approximations-all-but-NLS-sideways2}\MinEC{} heuristic performances on graphs from \CHMUset{} (top), \GWset{} (second from top), \EGset{} (second from bottom), and \AGset{} (bottom). In cases where the optimal objective value is unknown, the ratio is marked '?'.}
\centering
{\footnotesize\setlength{\tabcolsep}{3pt}\begin{tabular}{lll|lll|lll|lll|lll|lll}
\toprule
\multicolumn{3}{c}{} & \multicolumn{3}{c}{\lmmd{}} & \multicolumn{3}{c}{\pathlen{}} & \multicolumn{3}{c}{\moshort{}} & \multicolumn{3}{c}{\mcmc{}} & \multicolumn{3}{c}{\diffmincost{}} \\
Graph & n & m & obj & ratio & time (s) & obj & ratio & time (s) & obj & ratio & time (s) & obj & ratio & time (s) & obj & ratio & time (s) \\
\midrule
fig10.4 & 10 & 12 & 22 & 2.00 & 0.000 & 12 & 1.09 & 0.000 & 12 & 1.09 & 0.000 & 11 & 1.00 & 0.004 & 11 & 1.00 & 0.000 \\
fig10.1 & 11 & 12 & 16 & 2.29 & 0.000 & 8 & 1.14 & 0.000 & 7 & 1.00 & 0.000 & 7 & 1.00 & 0.003 & 8 & 1.14 & 0.000 \\
ex10.8 & 12 & 14 & 26 & 2.36 & 0.000 & 12 & 1.09 & 0.000 & 12 & 1.09 & 0.000 & 11 & 1.00 & 0.006 & 11 & 1.00 & 0.000 \\
revbound & 12 & 21 & 10 & 10.00 & 0.000 & 1 & 1.00 & 0.000 & 1 & 1.00 & 0.000 & 1 & 1.00 & 0.003 & 1 & 1.00 & 0.000 \\
butterfly & 16 & 24 & 48 & 3.00 & 0.000 & 16 & 1.00 & 0.000 & 16 & 1.00 & 0.000 & 16 & 1.00 & 0.006 & 16 & 1.00 & 0.000 \\
\midrule
fig10.2 & 6 & 6 & 4 & 4.00 & 0.000 & 1 & 1.00 & 0.000 & 1 & 1.00 & 0.000 & 1 & 1.00 & 0.001 & 1 & 1.00 & 0.000 \\
fig10.1.orig & 11 & 12 & 16 & 2.29 & 0.000 & 8 & 1.14 & 0.000 & 7 & 1.00 & 0.000 & 7 & 1.00 & 0.004 & 8 & 1.14 & 0.000 \\
fig10.3 & 18 & 20 & 30 & 2.73 & 0.000 & 12 & 1.09 & 0.000 & 12 & 1.09 & 0.000 & 11 & 1.00 & 0.008 & 12 & 1.09 & 0.000 \\
fig10.6 & 34 & 70 & 145 & 7.25 & 0.000 & 20 & 1.00 & 0.002 & 20 & 1.00 & 0.000 & 20 & 1.00 & 0.026 & 20 & 1.00 & 0.000 \\
\midrule
2d4x2x2 & 32 & 96 & 352 & 5.50 & 0.000 & 64 & 1.00 & 0.001 & 64 & 1.00 & 0.000 & 96 & 1.50 & 0.056 & 64 & 1.00 & 0.000 \\
2d3x3x2 & 36 & 135 & 630 & 7.78 & 0.000 & 81 & 1.00 & 0.002 & 81 & 1.00 & 0.000 & 135 & 1.67 & 0.116 & 81 & 1.00 & 0.000 \\
2d4x2x3 & 40 & 128 & 648 & 10.12 & 0.000 & 64 & 1.00 & 0.002 & 64 & 1.00 & 0.000 & 128 & 2.00 & 0.099 & 64 & 1.00 & 0.000 \\
2d3x3x3 & 45 & 180 & 1179 & 14.56 & 0.000 & 81 & 1.00 & 0.009 & 81 & 1.00 & 0.000 & 180 & 2.22 & 0.226 & 81 & 1.00 & 0.000 \\
2d5x5x2 & 100 & 375 & 2250 & 6.00 & 0.002 & 375 & 1.00 & 0.038 & 375 & 1.00 & 0.000 & 375 & 1.00 & 0.854 & 375 & 1.00 & 0.003 \\
2d5x5x3 & 125 & 500 & 5475 & 10.95 & 0.006 & 500 & 1.00 & 0.110 & 500 & 1.00 & 0.001 & 500 & 1.00 & 1.630 & 500 & 1.00 & 0.005 \\
2d5x5x5 & 175 & 750 & 15525 & 24.84 & 0.024 & 625 & 1.00 & 0.436 & 625 & 1.00 & 0.002 & 750 & 1.20 & 3.970 & 625 & 1.00 & 0.013 \\
2d10x10x2 & 400 & 1500 & 9000 & 6.00 & 0.043 & 1500 & 1.00 & 2.109 & 1500 & 1.00 & 0.006 & 1500 & 1.00 & 13.488 & 1500 & 1.00 & 0.038 \\
2d10x10x5 & 700 & 3000 & 80500 & ? & 0.695 & 3000 & ? & 26.494 & 3000 & ? & 0.029 & 3000 & ? & 62.509 & 3000 & ? & 0.195 \\
2d10x10x10 & 1200 & 5500 & 560000 & ? & 7.193 & 5500 & ? & 170.151 & 5500 & ? & 0.100 & 5500 & ? & 222.730 & 5500 & ? & 0.749 \\
\midrule
Helmholtz & 7 & 8 & 5 & 5.00 & 0.000 & 1 & 1.00 & 0.000 & 1 & 1.00 & 0.000 & 1 & 1.00 & 0.000 & 1 & 1.00 & 0.000 \\
PropaneComb. & 27 & 29 & 25 & 1.67 & 0.000 & 15 & 1.00 & 0.000 & 15 & 1.00 & 0.000 & 15 & 1.00 & 0.005 & 15 & 1.00 & 0.000 \\
Robot-6DOF & 34 & 43 & 41 & 10.25 & 0.000 & 4 & 1.00 & 0.004 & 4 & 1.00 & 0.000 & 4 & 1.00 & 0.007 & 4 & 1.00 & 0.000 \\
BlackScholes & 35 & 44 & 42 & 8.40 & 0.000 & 5 & 1.00 & 0.002 & 5 & 1.00 & 0.000 & 5 & 1.00 & 0.018 & 5 & 1.00 & 0.000 \\
HumanHeart & 36 & 43 & 42 & 5.25 & 0.000 & 8 & 1.00 & 0.002 & 8 & 1.00 & 0.000 & 8 & 1.00 & 0.008 & 8 & 1.00 & 0.000 \\
Perceptron & 44 & 47 & 43 & 5.38 & 0.000 & 8 & 1.00 & 0.006 & 8 & 1.00 & 0.000 & 8 & 1.00 & 0.010 & 8 & 1.00 & 0.000 \\
f & 46 & 61 & 71 & 11.83 & 0.000 & 6 & 1.00 & 0.010 & 6 & 1.00 & 0.000 & 6 & 1.00 & 0.014 & 6 & 1.00 & 0.000 \\
Encoder & 98 & 114 & 104 & 6.50 & 0.000 & 16 & 1.00 & 0.081 & 16 & 1.00 & 0.000 & 16 & 1.00 & 0.052 & 16 & 1.00 & 0.003 \\
RoeFlux-1d & 100 & 164 & 251 & 20.92 & 0.001 & 12 & 1.00 & 0.110 & 12 & 1.00 & 0.000 & 12 & 1.00 & 0.072 & 12 & 1.00 & 0.004 \\
g & 102 & 155 & 346 & 3.60 & 0.001 & 129 & 1.34 & 0.072 & 128 & 1.33 & 0.000 & 96 & 1.00 & 0.113 & 127 & 1.32 & 0.003 \\
RoeFlux-3d & 140 & 226 & 370 & 30.83 & 0.003 & 12 & 1.00 & 0.344 & 12 & 1.00 & 0.000 & 12 & 1.00 & 0.158 & 12 & 1.00 & 0.008 \\
\bottomrule
\end{tabular}}
\end{sidewaystable*}
\begin{sidewaystable*}
\centering
\caption{\label{table:scarcity-approximations-all-but-NLS-sideways3}\MinEC{} heuristic performances on graphs from \CHMUset{} (top), \GWset{} (second from top), \EGset{} (second from bottom), and \AGset{} (bottom). In cases where the optimal objective value is unknown, the ratio is marked '?'.}
\centering
{\footnotesize\setlength{\tabcolsep}{3pt}\begin{tabular}{lll|lll|lll|lll}
\toprule
\multicolumn{3}{c}{} & \multicolumn{3}{c}{\edgereduction{}} & \multicolumn{3}{c}{\markanddegree{}} & \multicolumn{3}{c}{\pc} \\
Graph & n & m & obj & ratio & time (s) & obj & ratio & time (s) & obj & ratio & time (s) \\
\midrule
fig10.4 & 10 & 12 & 12 & 1.09 & 0.000 & 12 & 1.09 & 0.000 & 12 & 1.09 & 0.000 \\
fig10.1 & 11 & 12 & 7 & 1.00 & 0.000 & 7 & 1.00 & 0.000 & 7 & 1.00 & 0.000 \\
ex10.8 & 12 & 14 & 12 & 1.09 & 0.000 & 12 & 1.09 & 0.000 & 12 & 1.09 & 0.000 \\
revbound & 12 & 21 & 1 & 1.00 & 0.000 & 1 & 1.00 & 0.000 & 1 & 1.00 & 0.000 \\
butterfly & 16 & 24 & 16 & 1.00 & 0.000 & 16 & 1.00 & 0.000 & 16 & 1.00 & 0.000 \\
\midrule
fig10.2 & 6 & 6 & 1 & 1.00 & 0.000 & 1 & 1.00 & 0.000 & 1 & 1.00 & 0.000 \\
fig10.1.orig & 11 & 12 & 7 & 1.00 & 0.000 & 7 & 1.00 & 0.000 & 7 & 1.00 & 0.000 \\
fig10.3 & 18 & 20 & 12 & 1.09 & 0.000 & 12 & 1.09 & 0.000 & 12 & 1.09 & 0.000 \\
fig10.6 & 34 & 70 & 20 & 1.00 & 0.000 & 20 & 1.00 & 0.000 & 20 & 1.00 & 0.000 \\
\midrule
2d4x2x2 & 32 & 96 & 64 & 1.00 & 0.000 & 64 & 1.00 & 0.000 & 64 & 1.00 & 0.000 \\
2d3x3x2 & 36 & 135 & 81 & 1.00 & 0.000 & 81 & 1.00 & 0.000 & 81 & 1.00 & 0.000 \\
2d4x2x3 & 40 & 128 & 64 & 1.00 & 0.000 & 64 & 1.00 & 0.000 & 64 & 1.00 & 0.000 \\
2d3x3x3 & 45 & 180 & 81 & 1.00 & 0.000 & 81 & 1.00 & 0.000 & 81 & 1.00 & 0.000 \\
2d5x5x2 & 100 & 375 & 375 & 1.00 & 0.001 & 375 & 1.00 & 0.001 & 375 & 1.00 & 0.001 \\
2d5x5x3 & 125 & 500 & 500 & 1.00 & 0.003 & 500 & 1.00 & 0.003 & 500 & 1.00 & 0.005 \\
2d5x5x5 & 175 & 750 & 625 & 1.00 & 0.012 & 625 & 1.00 & 0.012 & 625 & 1.00 & 0.012 \\
2d10x10x2 & 400 & 1500 & 1500 & 1.00 & 0.012 & 1500 & 1.00 & 0.012 & 1500 & 1.00 & 0.020 \\
2d10x10x5 & 700 & 3000 & 3000 & ? & 0.259 & 3000 & ? & 0.259 & 3000 & ? & 0.085 \\
2d10x10x10 & 1200 & 5500 & 5500 & ? & 2.554 & 5500 & ? & 2.555 & 5500 & ? & 0.285 \\
\midrule
Helmholtz & 7 & 8 & 1 & 1.00 & 0.000 & 1 & 1.00 & 0.000 & 1 & 1.00 & 0.000 \\
PropaneComb. & 27 & 29 & 15 & 1.00 & 0.000 & 15 & 1.00 & 0.000 & 15 & 1.00 & 0.000 \\
Robot-6DOF & 34 & 43 & 4 & 1.00 & 0.000 & 4 & 1.00 & 0.000 & 4 & 1.00 & 0.000 \\
BlackScholes & 35 & 44 & 5 & 1.00 & 0.000 & 5 & 1.00 & 0.000 & 5 & 1.00 & 0.000 \\
HumanHeart & 36 & 43 & 8 & 1.00 & 0.000 & 8 & 1.00 & 0.000 & 8 & 1.00 & 0.000 \\
Perceptron & 44 & 47 & 8 & 1.00 & 0.000 & 8 & 1.00 & 0.000 & 8 & 1.00 & 0.000 \\
f & 46 & 61 & 6 & 1.00 & 0.000 & 6 & 1.00 & 0.000 & 6 & 1.00 & 0.000 \\
Encoder & 98 & 114 & 16 & 1.00 & 0.000 & 16 & 1.00 & 0.000 & 16 & 1.00 & 0.001 \\
RoeFlux-1d & 100 & 164 & 12 & 1.00 & 0.001 & 12 & 1.00 & 0.001 & 12 & 1.00 & 0.004 \\
g & 102 & 155 & 96 & 1.00 & 0.000 & 96 & 1.00 & 0.000 & 118 & 1.23 & 0.001 \\
RoeFlux-3d & 140 & 226 & 12 & 1.00 & 0.002 & 12 & 1.00 & 0.002 & 12 & 1.00 & 0.002 \\
\bottomrule
\end{tabular}}
\end{sidewaystable*}
\begin{sidewaystable*}
\centering
\caption{\label{table:scarcity-approximations-TEST-NLS-sideways1}\MinEC{{}} heuristic performances on graphs from~\NLSset{}. In cases where the optimal objective value is unknown, the ratio is marked '?'.}
\centering
{\footnotesize\setlength{\tabcolsep}{3pt}\begin{tabular}{lll|lll|lll|lll|lll|lll}
\toprule
\multicolumn{3}{c}{} & \multicolumn{3}{c}{\fo{}} & \multicolumn{3}{c}{\re{}} & \multicolumn{3}{c}{\greedymark{}} & \multicolumn{3}{c}{\relmark{}} & \multicolumn{3}{c}{\mtmr{}} \\
Graph & n & m & obj & ratio & time (s) & obj & ratio & time (s) & obj & ratio & time (s) & obj & ratio & time (s) & obj & ratio & time (s) \\
\midrule
nls04 & 8 & 7 & 3 & 1.00 & 0.000 & 3 & 1.00 & 0.000 & 3 & 1.00 & 0.000 & 3 & 1.00 & 0.000 & 5 & 1.67 & 0.000 \\
nls07 & 12 & 16 & 4 & 1.00 & 0.000 & 4 & 1.00 & 0.000 & 4 & 1.00 & 0.000 & 4 & 1.00 & 0.000 & 12 & 3.00 & 0.000 \\
nls26 & 12 & 13 & 4 & 1.00 & 0.000 & 4 & 1.00 & 0.000 & 4 & 1.00 & 0.000 & 4 & 1.00 & 0.000 & 10 & 2.50 & 0.000 \\
nls05 & 15 & 13 & 4 & 1.00 & 0.000 & 4 & 1.00 & 0.000 & 4 & 1.00 & 0.000 & 4 & 1.00 & 0.000 & 10 & 2.50 & 0.000 \\
nls06 & 16 & 16 & 8 & 1.00 & 0.000 & 8 & 1.00 & 0.000 & 8 & 1.00 & 0.000 & 8 & 1.00 & 0.000 & 12 & 1.50 & 0.000 \\
nls21 & 20 & 29 & 6 & 1.00 & 0.000 & 6 & 1.00 & 0.000 & 6 & 1.00 & 0.000 & 6 & 1.00 & 0.000 & 31 & 5.17 & 0.000 \\
nls22 & 31 & 51 & 12 & 1.00 & 0.000 & 12 & 1.00 & 0.000 & 12 & 1.00 & 0.000 & 12 & 1.00 & 0.000 & 61 & 5.08 & 0.000 \\
nls03 & 48 & 47 & 32 & 2.00 & 0.000 & 31 & 1.94 & 0.000 & 16 & 1.00 & 0.000 & 16 & 1.00 & 0.000 & 137 & 8.56 & 0.000 \\
nls01 & 50 & 59 & 50 & 1.67 & 0.000 & 39 & 1.30 & 0.000 & 30 & 1.00 & 0.000 & 30 & 1.00 & 0.000 & 211 & 7.03 & 0.000 \\
nls13 & 52 & 60 & 20 & 1.00 & 0.000 & 20 & 1.00 & 0.000 & 20 & 1.00 & 0.000 & 20 & 1.00 & 0.000 & 50 & 2.50 & 0.000 \\
nls16 & 59 & 76 & 58 & 1.53 & 0.000 & 56 & 1.47 & 0.000 & 38 & 1.00 & 0.000 & 38 & 1.00 & 0.000 & 210 & 5.53 & 0.000 \\
nls02 & 60 & 59 & 40 & 2.00 & 0.000 & 39 & 1.95 & 0.000 & 20 & 1.00 & 0.000 & 20 & 1.00 & 0.000 & 211 & 10.55 & 0.000 \\
nls12 & 63 & 80 & 30 & 1.00 & 0.000 & 30 & 1.00 & 0.000 & 30 & 1.00 & 0.000 & 30 & 1.00 & 0.000 & 70 & 2.33 & 0.000 \\
nls08 & 63 & 90 & 45 & 1.00 & 0.000 & 45 & 1.00 & 0.000 & 45 & 1.00 & 0.000 & 45 & 1.00 & 0.000 & 75 & 1.67 & 0.000 \\
nls19 & 77 & 90 & 30 & 1.00 & 0.000 & 30 & 1.00 & 0.000 & 30 & 1.00 & 0.000 & 30 & 1.00 & 0.000 & 75 & 2.50 & 0.000 \\
nls09 & 81 & 110 & 44 & 1.00 & 0.000 & 44 & 1.00 & 0.000 & 44 & 1.00 & 0.000 & 44 & 1.00 & 0.000 & 99 & 2.25 & 0.000 \\
nls14 & 84 & 140 & 80 & 1.00 & 0.000 & 80 & 1.00 & 0.000 & 80 & 1.00 & 0.000 & 80 & 1.00 & 0.000 & 120 & 1.50 & 0.000 \\
nls10 & 99 & 128 & 48 & 1.00 & 0.000 & 48 & 1.00 & 0.000 & 48 & 1.00 & 0.000 & 48 & 1.00 & 0.000 & 112 & 2.33 & 0.000 \\
nls23 & 115 & 189 & 50 & 1.00 & 0.000 & 50 & 1.00 & 0.000 & 50 & 1.00 & 0.000 & 50 & 1.00 & 0.001 & 291 & 5.82 & 0.001 \\
nls20 & 123 & 150 & 45 & 1.00 & 0.000 & 45 & 1.00 & 0.000 & 45 & 1.00 & 0.000 & 45 & 1.00 & 0.000 & 135 & 3.00 & 0.001 \\
nls24 & 147 & 213 & 94 & 1.00 & 0.000 & 94 & 1.00 & 0.000 & 94 & 1.00 & 0.000 & 94 & 1.00 & 0.001 & 189 & 2.01 & 0.001 \\
nls15 & 172 & 340 & 72 & 1.00 & 0.000 & 72 & 1.00 & 0.000 & 72 & 1.00 & 0.000 & 72 & 1.00 & 0.004 & 428 & 5.94 & 0.002 \\
nls25 & 237 & 324 & 117 & 1.00 & 0.000 & 117 & 1.00 & 0.000 & 117 & 1.00 & 0.000 & 117 & 1.00 & 0.003 & 299 & 2.56 & 0.007 \\
nls17 & 395 & 520 & 163 & 1.00 & 0.000 & 163 & 1.00 & 0.000 & 163 & 1.00 & 0.000 & 163 & 1.00 & 0.011 & 487 & 2.99 & 0.032 \\
nls11 & 593 & 1019 & 495 & 1.00 & 0.002 & 495 & 1.00 & 0.000 & 495 & 1.00 & 0.002 & 495 & 1.00 & 0.028 & 930 & 1.88 & 0.112 \\
nls18 & 1751 & 2383 & 711 & 1.00 & 0.009 & 711 & 1.00 & 0.005 & 711 & 1.00 & 0.014 & 711 & 1.00 & 0.431 & 2315 & 3.26 & 3.101 \\
\bottomrule
\end{tabular}}
\end{sidewaystable*}
\begin{sidewaystable*}
\centering
\caption{\label{table:scarcity-approximations-TEST-NLS-sideways2}\MinEC{{}} heuristic performances on graphs from~\NLSset{}. In cases where the optimal objective value is unknown, the ratio is marked '?'.}
\centering
{\footnotesize\setlength{\tabcolsep}{3pt}\begin{tabular}{lll|lll|lll|lll|lll|lll}
\toprule
\multicolumn{3}{c}{} & \multicolumn{3}{c}{\lmmd{}} & \multicolumn{3}{c}{\pathlen{}} & \multicolumn{3}{c}{\moshort{}} & \multicolumn{3}{c}{\mcmc{}} & \multicolumn{3}{c}{\diffmincost{}} \\
Graph & n & m & obj & ratio & time (s) & obj & ratio & time (s) & obj & ratio & time (s) & obj & ratio & time (s) & obj & ratio & time (s) \\
\midrule
nls04 & 8 & 7 & 5 & 1.67 & 0.000 & 3 & 1.00 & 0.000 & 3 & 1.00 & 0.000 & 3 & 1.00 & 0.002 & 3 & 1.00 & 0.000 \\
nls07 & 12 & 16 & 10 & 2.50 & 0.000 & 4 & 1.00 & 0.000 & 4 & 1.00 & 0.000 & 4 & 1.00 & 0.002 & 4 & 1.00 & 0.000 \\
nls26 & 12 & 13 & 10 & 2.50 & 0.000 & 4 & 1.00 & 0.000 & 4 & 1.00 & 0.000 & 4 & 1.00 & 0.002 & 4 & 1.00 & 0.000 \\
nls05 & 15 & 13 & 10 & 2.50 & 0.000 & 4 & 1.00 & 0.000 & 4 & 1.00 & 0.000 & 4 & 1.00 & 0.005 & 4 & 1.00 & 0.000 \\
nls06 & 16 & 16 & 12 & 1.50 & 0.000 & 8 & 1.00 & 0.000 & 8 & 1.00 & 0.000 & 8 & 1.00 & 0.006 & 8 & 1.00 & 0.000 \\
nls21 & 20 & 29 & 28 & 4.67 & 0.000 & 6 & 1.00 & 0.000 & 6 & 1.00 & 0.000 & 6 & 1.00 & 0.008 & 6 & 1.00 & 0.000 \\
nls22 & 31 & 51 & 53 & 4.42 & 0.000 & 12 & 1.00 & 0.002 & 12 & 1.00 & 0.000 & 12 & 1.00 & 0.014 & 12 & 1.00 & 0.000 \\
nls03 & 48 & 47 & 110 & 6.88 & 0.000 & 47 & 2.94 & 0.012 & 16 & 1.00 & 0.000 & 16 & 1.00 & 0.027 & 31 & 1.94 & 0.000 \\
nls01 & 50 & 59 & 159 & 5.30 & 0.000 & 59 & 1.97 & 0.005 & 50 & 1.67 & 0.000 & 30 & 1.00 & 0.027 & 39 & 1.30 & 0.000 \\
nls13 & 52 & 60 & 50 & 2.50 & 0.000 & 20 & 1.00 & 0.009 & 20 & 1.00 & 0.000 & 20 & 1.00 & 0.018 & 20 & 1.00 & 0.000 \\
nls16 & 59 & 76 & 166 & 4.37 & 0.000 & 76 & 2.00 & 0.011 & 58 & 1.53 & 0.000 & 38 & 1.00 & 0.032 & 47 & 1.24 & 0.000 \\
nls02 & 60 & 59 & 159 & 7.95 & 0.000 & 59 & 2.95 & 0.011 & 20 & 1.00 & 0.000 & 20 & 1.00 & 0.030 & 39 & 1.95 & 0.000 \\
nls12 & 63 & 80 & 70 & 2.33 & 0.000 & 30 & 1.00 & 0.019 & 30 & 1.00 & 0.000 & 30 & 1.00 & 0.027 & 30 & 1.00 & 0.001 \\
nls08 & 63 & 90 & 75 & 1.67 & 0.000 & 45 & 1.00 & 0.015 & 45 & 1.00 & 0.000 & 45 & 1.00 & 0.032 & 45 & 1.00 & 0.001 \\
nls19 & 77 & 90 & 75 & 2.50 & 0.000 & 30 & 1.00 & 0.034 & 30 & 1.00 & 0.000 & 30 & 1.00 & 0.040 & 30 & 1.00 & 0.001 \\
nls09 & 81 & 110 & 99 & 2.25 & 0.000 & 44 & 1.00 & 0.042 & 44 & 1.00 & 0.000 & 44 & 1.00 & 0.043 & 44 & 1.00 & 0.001 \\
nls14 & 84 & 140 & 120 & 1.50 & 0.000 & 80 & 1.00 & 0.034 & 80 & 1.00 & 0.000 & 80 & 1.00 & 0.063 & 80 & 1.00 & 0.001 \\
nls10 & 99 & 128 & 112 & 2.33 & 0.001 & 48 & 1.00 & 0.066 & 48 & 1.00 & 0.000 & 48 & 1.00 & 0.063 & 48 & 1.00 & 0.002 \\
nls23 & 115 & 189 & 246 & 4.92 & 0.001 & 50 & 1.00 & 0.137 & 50 & 1.00 & 0.000 & 50 & 1.00 & 0.102 & 50 & 1.00 & 0.004 \\
nls20 & 123 & 150 & 135 & 3.00 & 0.003 & 45 & 1.00 & 0.143 & 45 & 1.00 & 0.000 & 45 & 1.00 & 0.084 & 45 & 1.00 & 0.004 \\
nls24 & 147 & 213 & 189 & 2.01 & 0.003 & 94 & 1.00 & 0.234 & 94 & 1.00 & 0.000 & 94 & 1.00 & 0.158 & 94 & 1.00 & 0.006 \\
nls15 & 172 & 340 & 388 & 5.39 & 0.005 & 72 & 1.00 & 0.415 & 72 & 1.00 & 0.000 & 72 & 1.00 & 0.284 & 72 & 1.00 & 0.010 \\
nls25 & 237 & 324 & 299 & 2.56 & 0.013 & 117 & 1.00 & 1.093 & 117 & 1.00 & 0.001 & 117 & 1.00 & 0.353 & 117 & 1.00 & 0.016 \\
nls17 & 395 & 520 & 487 & 2.99 & 0.062 & 163 & 1.00 & 5.158 & 163 & 1.00 & 0.003 & 163 & 1.00 & 0.851 & 163 & 1.00 & 0.042 \\
nls11 & 593 & 1019 & 930 & 1.88 & 0.207 & 495 & 1.00 & 18.454 & 495 & 1.00 & 0.009 & 495 & 1.00 & 2.775 & 495 & 1.00 & 0.119 \\
nls18 & 1751 & 2383 & 2315 & 3.26 & 7.229 & 711 & 1.00 & 579.894 & 711 & 1.00 & 0.067 & 711 & 1.00 & 16.894 & 711 & 1.00 & 1.034 \\
\bottomrule
\end{tabular}}
\end{sidewaystable*}
\begin{sidewaystable*}
\centering
\caption{\label{table:scarcity-approximations-TEST-NLS-sideways3}\MinEC{{}} heuristic performances on graphs from~\NLSset{}. In cases where the optimal objective value is unknown, the ratio is marked '?'.}
\centering
{\footnotesize\setlength{\tabcolsep}{3pt}\begin{tabular}{lll|lll|lll|lll}
\toprule
\multicolumn{3}{c}{} & \multicolumn{3}{c}{\edgereduction{}} & \multicolumn{3}{c}{\markanddegree{}} & \multicolumn{3}{c}{\pc} \\
Graph & n & m & obj & ratio & time (s) & obj & ratio & time (s) & obj & ratio & time (s) \\
\midrule
nls04 & 8 & 7 & 3 & 1.00 & 0.000 & 3 & 1.00 & 0.000 & 3 & 1.00 & 0.000 \\
nls07 & 12 & 16 & 4 & 1.00 & 0.000 & 4 & 1.00 & 0.000 & 4 & 1.00 & 0.000 \\
nls26 & 12 & 13 & 4 & 1.00 & 0.000 & 4 & 1.00 & 0.000 & 4 & 1.00 & 0.000 \\
nls05 & 15 & 13 & 4 & 1.00 & 0.000 & 4 & 1.00 & 0.000 & 4 & 1.00 & 0.000 \\
nls06 & 16 & 16 & 8 & 1.00 & 0.000 & 8 & 1.00 & 0.000 & 8 & 1.00 & 0.000 \\
nls21 & 20 & 29 & 6 & 1.00 & 0.000 & 6 & 1.00 & 0.000 & 6 & 1.00 & 0.000 \\
nls22 & 31 & 51 & 12 & 1.00 & 0.000 & 12 & 1.00 & 0.000 & 12 & 1.00 & 0.000 \\
nls03 & 48 & 47 & 16 & 1.00 & 0.000 & 16 & 1.00 & 0.000 & 16 & 1.00 & 0.000 \\
nls01 & 50 & 59 & 30 & 1.00 & 0.000 & 30 & 1.00 & 0.000 & 30 & 1.00 & 0.000 \\
nls13 & 52 & 60 & 20 & 1.00 & 0.000 & 20 & 1.00 & 0.000 & 20 & 1.00 & 0.001 \\
nls16 & 59 & 76 & 38 & 1.00 & 0.000 & 38 & 1.00 & 0.000 & 38 & 1.00 & 0.000 \\
nls02 & 60 & 59 & 20 & 1.00 & 0.000 & 20 & 1.00 & 0.000 & 20 & 1.00 & 0.001 \\
nls12 & 63 & 80 & 30 & 1.00 & 0.000 & 30 & 1.00 & 0.000 & 30 & 1.00 & 0.002 \\
nls08 & 63 & 90 & 45 & 1.00 & 0.000 & 45 & 1.00 & 0.000 & 45 & 1.00 & 0.000 \\
nls19 & 77 & 90 & 30 & 1.00 & 0.000 & 30 & 1.00 & 0.000 & 30 & 1.00 & 0.000 \\
nls09 & 81 & 110 & 44 & 1.00 & 0.000 & 44 & 1.00 & 0.000 & 44 & 1.00 & 0.000 \\
nls14 & 84 & 140 & 80 & 1.00 & 0.000 & 80 & 1.00 & 0.000 & 80 & 1.00 & 0.002 \\
nls10 & 99 & 128 & 48 & 1.00 & 0.000 & 48 & 1.00 & 0.000 & 48 & 1.00 & 0.001 \\
nls23 & 115 & 189 & 50 & 1.00 & 0.001 & 50 & 1.00 & 0.001 & 50 & 1.00 & 0.003 \\
nls20 & 123 & 150 & 45 & 1.00 & 0.001 & 45 & 1.00 & 0.001 & 45 & 1.00 & 0.002 \\
nls24 & 147 & 213 & 94 & 1.00 & 0.002 & 94 & 1.00 & 0.001 & 94 & 1.00 & 0.006 \\
nls15 & 172 & 340 & 72 & 1.00 & 0.001 & 72 & 1.00 & 0.001 & 72 & 1.00 & 0.007 \\
nls25 & 237 & 324 & 117 & 1.00 & 0.005 & 117 & 1.00 & 0.005 & 117 & 1.00 & 0.011 \\
nls17 & 395 & 520 & 163 & 1.00 & 0.020 & 163 & 1.00 & 0.021 & 163 & 1.00 & 0.023 \\
nls11 & 593 & 1019 & 495 & 1.00 & 0.036 & 495 & 1.00 & 0.036 & 495 & 1.00 & 0.037 \\
nls18 & 1751 & 2383 & 711 & 1.00 & 0.358 & 711 & 1.00 & 0.355 & 711 & 1.00 & 0.421 \\
\bottomrule
\end{tabular}}
\end{sidewaystable*}

\end{document}